\documentclass[superscriptaddress,onecolumn,floatfix,aps,10pt,prx]{revtex4-2}
\usepackage{adjustbox,amsmath,amsfonts,amsthm, appendix,dsfont}
\usepackage{algorithm}
\usepackage{algpseudocode}
\usepackage[normalem]{ulem}
\usepackage{bbm}
\usepackage{booktabs}
\usepackage{enumitem}
\usepackage{graphicx,wasysym}
\usepackage{hyperref}
 \hypersetup{
     colorlinks=true,
     linkcolor=blue,
     filecolor=blue,
     citecolor = blue,      
     urlcolor=cyan,
     }
\usepackage{lineno}
\usepackage{lipsum}
\usepackage{mathtools}
\usepackage{pgfplots}
    \pgfplotsset{compat=1.8}
\usepackage{physics}
\usepackage{pifont}
\usepackage{quantikz}
\usepackage[mode=buildnew]{standalone}
\usepackage[caption=false]{subfig}
\usepackage{tablefootnote}
\usepackage{tikz}
\usepackage{times}
\usepackage{threeparttable}
\usepackage{xcolor}
\definecolor{myred}{RGB}{200, 36, 35}
\definecolor{mygreen}{RGB}{27, 121, 57}
\definecolor{myblue}{RGB}{47, 127, 193}
\definecolor{mypink}{RGB}{116, 40, 129}
\definecolor{mypurple}{RGB}{181, 83, 132}
\usepackage{xparse}

\newtheorem{theorem}{Theorem}
\newtheorem*{theorem*}{Theorem}

\newtheorem*{definition*}{Definition}

\newtheorem{atheorem}{Theorem}[section]
\newtheorem{adefinition}[atheorem]{Definition}
\newtheorem{alemma}[atheorem]{Lemma}

\renewcommand{\Vec}[1]{\mathbf{#1}}

\usepackage{xifthen}
\newcommand{\PP}[1][]{
  \ifthenelse{\isempty{#1}}
    {\mathbbm{P}}
    {\mathbbm{P}\left[#1\right]}
}
\newcommand{\EE}[1][]{
  \ifthenelse{\isempty{#1}}
    {\mathbbm{E}}
    {\mathbbm{E}\left[#1\right]}
}

\begin{document}
\setlength{\columnsep}{25pt}
\title{A Pathway to Practical Quantum Advantage in Solving Navier-Stokes Equations}

\author{Xi-Ning Zhuang}
\affiliation{Laboratory of Quantum Information, University of Science and Technology of China, Hefei, 230026, China}
\author{Zhao-Yun Chen}
\email{chenzhaoyun@iai.ustc.edu.cn}
\affiliation{Institute of Artificial Intelligence, Hefei Comprehensive National Science Center}
\author{Ming-Yang Tan}
\affiliation{Laboratory of Quantum Information, University of Science and Technology of China, Hefei, 230026, China}
\author{Jiaxuan Zhang}
\affiliation{Institute of Artificial Intelligence, Hefei Comprehensive National Science Center}
\author{Chuang-Chao Ye}
\affiliation{Origin Quantum Computing,  Hefei,  China}
\author{Tian-Hao Wei}
\affiliation{Laboratory of Quantum Information, University of Science and Technology of China, Hefei, 230026, China}
\author{Teng-Yang Ma}
\affiliation{Origin Quantum Computing,  Hefei,  China}
\author{Cheng Xue}
\affiliation{Institute of Artificial Intelligence, Hefei Comprehensive National Science Center}
\author{Huan-Yu~Liu}
\affiliation{Institute of Artificial Intelligence, Hefei Comprehensive National Science Center}
\author{Qing-Song Li}
\affiliation{Laboratory of Quantum Information, University of Science and Technology of China, Hefei, 230026, China}
\author{Tai-Ping Sun}
\affiliation{Laboratory of Quantum Information, University of Science and Technology of China, Hefei, 230026, China}
\author{Xiao-Fan Xu}
\affiliation{Laboratory of Quantum Information, University of Science and Technology of China, Hefei, 230026, China}
\author{Yun-Jie Wang}
\affiliation{Institute of Advanced Technology, University of Science and Technology of China, Hefei, Anhui, 230031, China}
\author{Yu-Chun Wu}
\email{wuyuchun@ustc.edu.cn}
\affiliation{Laboratory of Quantum Information, University of Science and Technology of China, Hefei, 230026, China}
\affiliation{Anhui Province Key Laboratory of Quantum Network, University of Science and Technology of China, Hefei, 230026, China}
\affiliation{Institute of Artificial Intelligence, Hefei Comprehensive National Science Center}
\author{Guo-Ping Guo}
\email{gpguo@ustc.edu.cn}
\affiliation{Laboratory of Quantum Information, University of Science and Technology of China, Hefei, 230026, China}
\affiliation{Anhui Province Key Laboratory of Quantum Network, University of Science and Technology of China, Hefei, 230026, China}
\affiliation{Institute of Artificial Intelligence, Hefei Comprehensive National Science Center}
\affiliation{Origin Quantum Computing,  Hefei,  China}

\begin{abstract}
The advent of fault-tolerant quantum computing (FTQC) promises to tackle classically intractable problems. A key milestone is solving the Navier-Stokes equations (NSE), which has remained formidable for quantum algorithms due to their high input-output overhead and nonlinearity. Here, we establish a full-stack framework that charts a practical pathway to a quantum advantage for large-scale NSE simulation. Our approach integrates a spectral-based input/output algorithm, an explicit and synthesized quantum circuit, and a refined error-correction protocol. The algorithm achieves an end-to-end exponential speedup in asymptotic complexity, meeting the lower bound for general quantum linear system solvers. Through symmetry-based circuit synthesis and optimized error correction, we reduce the required logical and physical resources by two orders of magnitude. Our concrete resource analysis demonstrates that solving NSE on a $2^{80}$-grid is feasible with 8.71 million physical qubits (at an error rate of $5\times10^{-4}$) in 42.6 days—outperforming a state-of-the-art supercomputer, which would require over a century. This work bridges the gap between theoretical quantum speedup and the practical deployment of high-performance scientific computing.

\end{abstract}

\maketitle

\section{Introduction}

Quantum computing is rapidly emerging as a revolutionary computational paradigm with transformative potential for solving classically intractable problems~\cite{shor1999polynomial, arute2019quantum, kim2023evidence}. The field is now witnessing early evidence of a transition into the fault-tolerant quantum computing (FTQC) regime, where large-scale, error-corrected qubit arrays enable the reliable execution of deep quantum circuits, fueling ambitions for a clear pathway to practical quantum advantage~\cite{shor1996fault,preskill1998fault,gottesman2022opportunities, bluvstein2024logical, google2025quantum, katabarwa2024early, yoshioka2024hunting}. 
At this crossroad, one must ask whether and where an FTQC quantum computer can outperform a classical one in solving more practical problems beyond quantum simulation and factorization~\cite{lee2021even, gidney2021factor, yoshioka2024hunting, gidney2025factor}.
The fundamental difficulty often originates from the non-linear behaviors in the classical world and from a limited information conversion bandwidth.

Nevertheless, the pursuit of practical quantum advantage faces significant challenges since potential speedups are often undermined by inefficient input/output (I/O) protocols, suboptimal circuit synthesis, costly QEC, or unfavorable algorithmic prefactors and problem hyperparameters~\cite{dalzell2023end}. 
Bridging the gap between asymptotic speedups and a tangible quantum-classical crossover in the FTQC era requires a concerted full-stack effort encompassing:
Firstly, in the algorithm design phase to reduce complexity and quantum-classical data conversion overhead~\cite{gottesman2022opportunities, katabarwa2024early}, the I/O bottleneck is particularly critical: theoretical lower bounds show an exponential scaling with problem size, revealing an unavoidable spacetime trade-off for encoding large-scale classical data into quantum states~\cite{zhang2022quantum,clader2023quantum}. Conventional approaches, such as quantum random access memory (qRAM), mitigate time overhead at the cost of an exponential increase in qubit number~\cite{giovannetti2008quantum,di2020fault}. Furthermore, Holevo's bound fundamentally limits the amount of classical information extractable from a quantum state, necessitating a large number of repetitions of the underlying algorithm~\cite{holevo1973bounds}. This scalability barrier is exacerbated in I/O-intensive applications (e.g., iterative algorithms, quantum machine learning), where the no-cloning theorem necessitates costly state re-preparation and tomography between computational stages~\cite{dalzell2023end,kerenidis2020quantum}.
Secondly, in the circuit deployment phase to optimize logical and physical resources, FTQC imposes stringent constraints: non-Clifford gates—essential for universality—dominate resource overheads, while the search space for quantum architecture grows exponentially. Implementing high-fidelity non-Clifford gates requires costly magic state distillation and non-transversal implementations~\cite{ge2024quantum, van2023optimising, fowler2012surface,gottesman2022opportunities}. 
Finally, in the resource analysis phase to integrate these components into a realistic estimate, one must account for problem-specific prefactors and hyperparameters, often requiring extensive numerical experiments and multi-objective mixed-integer optimization~\cite{lee2021even, dalzell2023end, yoshioka2024hunting}.

A concrete quantum advantage in computational fluid dynamics (CFD) represents one of the most exciting yet challenging milestones for practical quantum advantage in the FTQC era. On one hand, solving the fundamental Navier-Stokes equations (NSE) is a central challenge in mathematics and physics, with critical applications across aerospace, automotive, energy, and environmental engineering~\cite{blazek2015computational}. On the other hand, NSE embodies all the aforementioned challenges: its nonlinear dynamics and exceptionally heavy I/O requirements make it significantly more demanding than FTQC benchmarks established in condensed matter physics, quantum chemistry, or cryptography~\cite{lee2021even, gidney2021factor, yoshioka2024hunting, gidney2025factor}.

\begin{figure}
\centering
\begin{tikzpicture}
\node[] at (0, 0) {\includegraphics[width=0.975\textwidth]{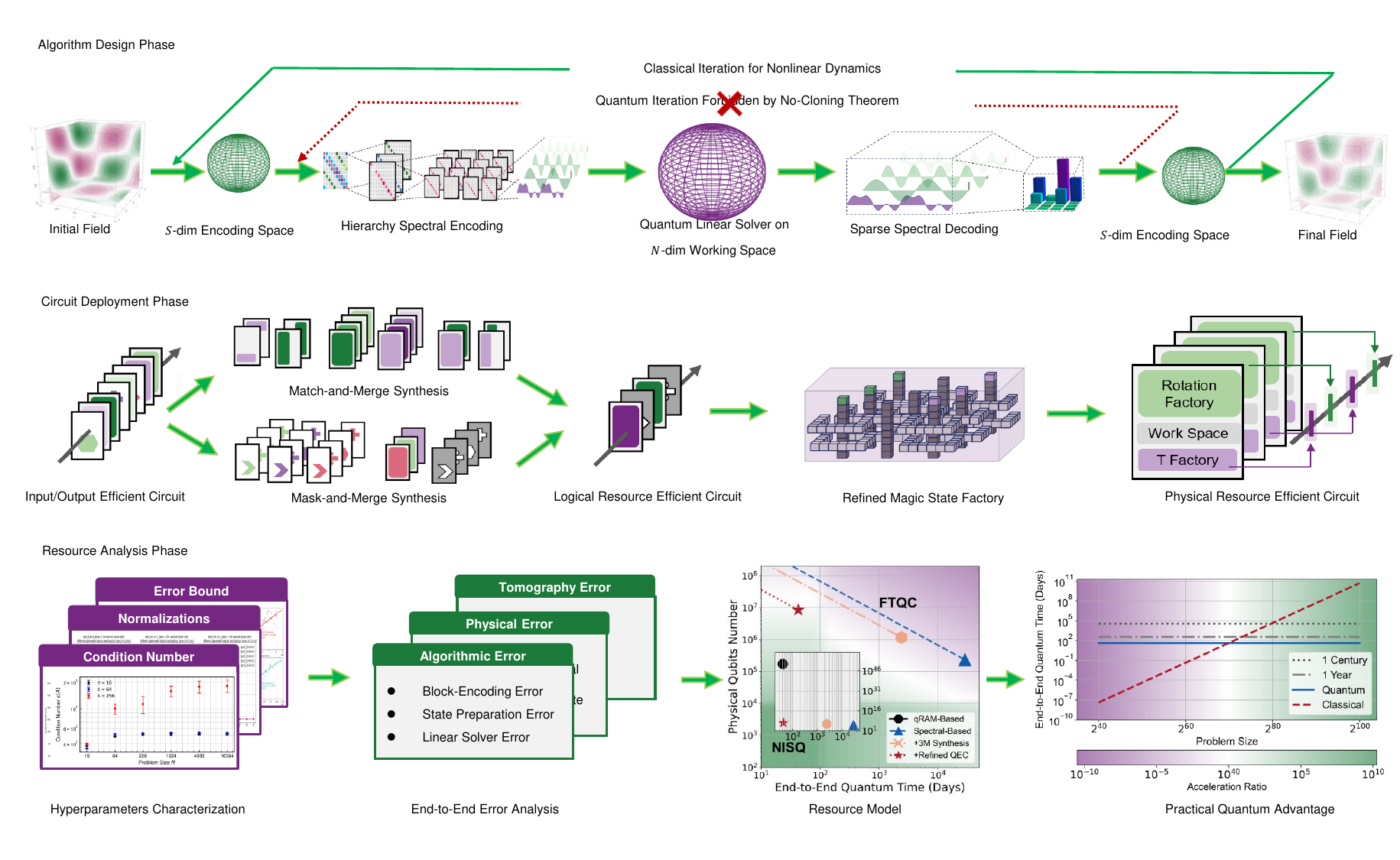}};
\node[] at (-8.5, 4.85) {\textbf{a}};
\node[] at (-8.5, 1.64) {\textbf{b}};
\node[] at (-8.5, -1.52) {\textbf{c}};
\end{tikzpicture}
\caption{Hierarchical Spectral Quantum Navier-Stokes Solver (QNSS) Framework. (\textbf{a}) \textbf{Algorithm Design Phase:} An initial field is translated into an $\mathcal{S}$-dimensional encoding space to mitigate the input bottleneck. Our hierarchical spectral encoding method expands this Hilbert space into an $N$-dimensional working space suitable for a quantum linear solver. A sparse spectral decoding method translates the dense $N$-dim state back to an $\mathcal{S}$-dim state, respecting Holevo's bound. Due to nonlinearity, copies of the state are required; while economic quantum iteration is forbidden by the no-cloning theorem, classical iteration is permitted. The final field is read out from the encoding space. (\textbf{b}) \textbf{Circuit Deployment Phase:} We apply match-and-merge and mask-and-merge synthesis (3M Synthesis) methods to reduce logical resources by exploiting global and artificial symmetries. A gate-teleportation-based QEC protocol, combined with a refined magic state factory, further reduces physical resources. (\textbf{c}) \textbf{Resource Analysis Phase:} We perform extensive numerical experiments to validate hyperparameters, and based on  an end-to-end error analysis (algorithmic, physical, and tomography errors), we derive a resource model for practical quantum advantage: QNSS outperforms state-of-the-art supercomputers on $2^{70}$-grids, and can solve the $2^{80}$-grids case within 42.6 days, while a supercomputer will utilize more than a century.
Inset: Spectral-Based I/O enables $10^{42}\times$ physical qubits number reduction due to an exponential reduction in both the number of logical qubits and rotation counts.}
\label{fig:schematic}
\end{figure}
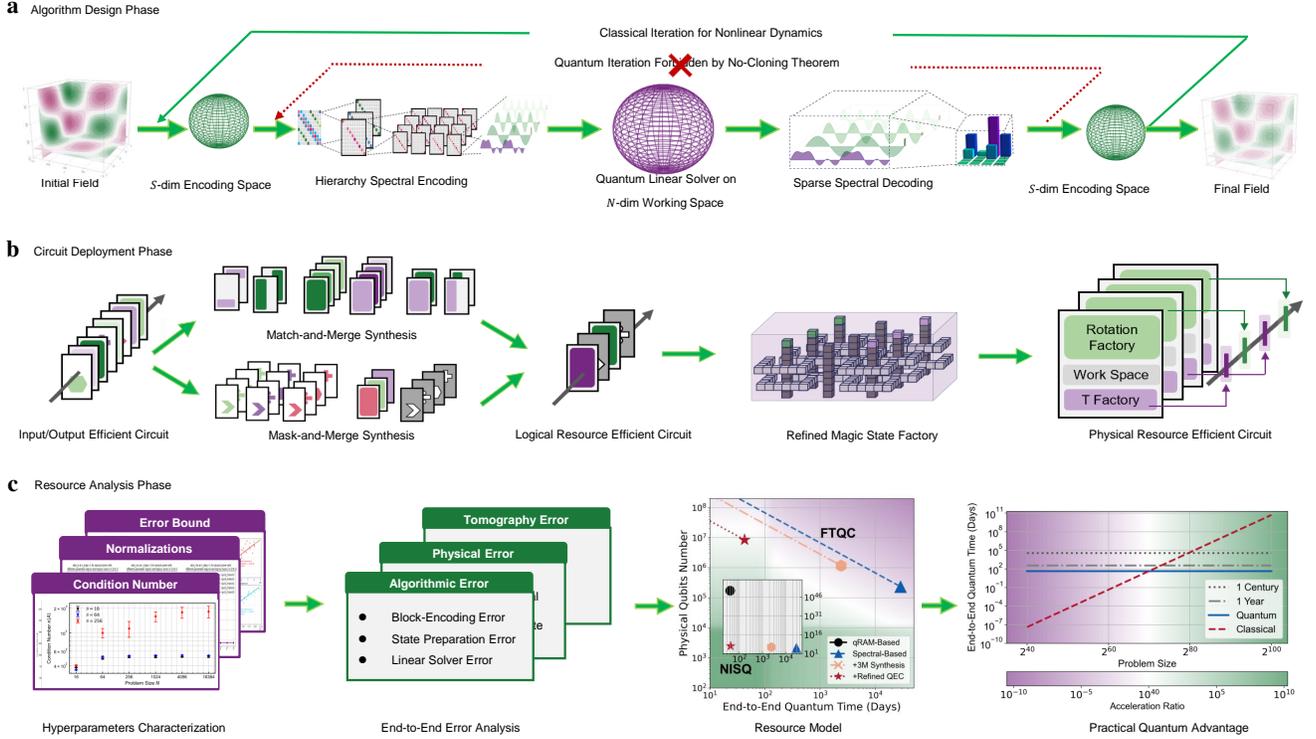

In this work, we establish a concrete pathway to practical quantum advantage for CFD in the FTQC era by addressing these critical bottlenecks. We present an end-to-end quantum algorithm that efficiently solves NSE on large-scale grids—a task that is classically intractable due to prohibitive spacetime complexity. The NSE is transformed into a series of iterative linear systems via the implicit Euler method for temporal discretization and the finite volume method (FVM) for spatial discretization~\cite{blazek2015computational}, and then is solved exponentially faster with a quantum linear solver, provided the I/O bottleneck is overcome.
In fact, we prove that our algorithm's end-to-end complexity reaches the theoretical lower bound for a general iterative quantum linear system model.

In the algorithm design phase (Fig.~\ref{fig:schematic}a), we circumvent the classical-quantum conversion bandwidth limitation by designing a qRAM-free I/O protocol that leverages prior knowledge of the problem's structure—including the variant's spectral structure and the matrix's hierarchical block structure induced by numerical schemes. The core idea is to expand the Hilbert space using a quantum circuit that encodes prior spectral and structural information, rather than relying on a costly qRAM architecture. Consequently, the information that must be injected into or extracted from a quantum state scales only with the spectral sparsity $\mathcal{S}$, rather than the full grid size $N$. This allows us to solve the NSE in $\mathcal{O}(\mathcal{S}\log N)$ time using only $\mathcal{O}(\mathcal{S},\mathrm{Poly}\log N)$ logical qubits, a significant improvement over the $\mathcal{O}(N^2)$ qubits required in previous approaches for $N$-dimensional linear systems~\cite{dalzell2023end}.

In the circuit deployment phase (Fig.~\ref{fig:schematic}b), we develop a match-mask-and-merge circuit synthesis strategy to reduce logical resources. The symmetry inherent in our I/O method's subcircuits enables us to mask, match, and merge components, significantly reducing gate counts and circuit depth. We further introduce a hybrid QEC method with a refined magic state factory to reduce physical resource overhead.

In the resource analysis phase (Fig.~\ref{fig:schematic}c), we numerically validate our method, characterizing hyperparameters and modeling errors to demonstrate scalability. Our full-stack software framework automates these optimizations, enabling rigorous resource estimation. For a $2^{40} \times 2^{40}$ grid (corresponding to a $\mathbf{2^{82}}$-dimensional linear system), our approach solves the NSE using \textbf{8.71 million} physical qubits (at a physical error rate of $5\times10^{-4}$) within \textbf{42.6 days}. This promises a $\mathbf{1100\times}$ speedup compared to a state-of-the-art supercomputer, which would require an estimated \textbf{130 years}.

Our work provides an optimistic and concrete step toward practical quantum advantage in the FTQC era, significantly bringing forward the quantum-classical crossover for CFD. The novel techniques developed here—including block encoding, state tomography, circuit synthesis, and error correction—are of independent interest. Furthermore, the resource bottlenecks identified in this study highlight key areas requiring future breakthroughs.

\section{Results}

\subsection{End-to-End Quantum Navier-Stokes Solver}
We study the $D$-dimensional compressible NSE, which comprise the conservation laws for mass, momentum, and energy. 
As a representative, the momentum balance can be written as
\begin{equation}\nonumber
    \frac{\partial}{\partial t} (\rho \Vec{v}) +\nabla\cdot(\rho\Vec{v}\otimes\Vec{v}) = \nabla\cdot\boldsymbol{\sigma},
\end{equation}
where $\rho$ is the density, $\Vec{v}=(u, v, ...)$ is the $D$-dimensional velocity, and $\boldsymbol{\sigma}$ is the viscosity tensor.
Due to the underlying nonlinear dynamics, the NSE is often transformed into an iterative $(D+2)N$-dimensional linear system, by FVM spatial discretization on $N$ grids and implicit Euler  temporal discretization~\cite{blazek2015computational}.
While Harrow-Hassidim-Lloyd-type quantum linear system solvers (HHL-QLSS) promise an exponential quantum speed-up on the linear system's size~\cite{harrow2009quantum, dalzell2024shortcut}, the application to an iterative quantum linear system solver (QLSS) is heavily limited for two primary reasons:
Firstly, the encoding and read-out of classical information from real-world scenarios are both proven hard as a direct result of the state-preparation complexity lower bound and the Holevo's bound, respectively~\cite{holevo1973bounds}.
Secondly, and even more seriously, as an economic quantum-only iteration is forbidden by the No-Cloning Theorem~\cite{wootters1982single}, a classical iteration for nonlinear dynamics is necessary, wherein the costly oracles for the large-scale matrix and vector are repeated an enormous times in each single iteration, as illustrated in Fig.~\ref{fig:schematic} (a). 

We propose a general iterative QLSS model to quantitatively study these limitations and find that the efficient spectral sparsity plays a key role in its complexity lower bound.
Assume an iterative linear system $A^{(k)}W^{(k)}=b^{(k)} (0\leq k<\tau)$, where the $(k+1)$-th iteration's left-hand-side matrix $A$ and right-hand-side vector $b$ are determined by the $k$-th solution, and $\tau$ is the total number of iterations.
The corresponding (generalized) iterative QLSS model can be formulated as
\begin{equation}
    {\Tilde{W}}^{(k)}\xrightarrow{\mathcal{E}}\mathcal{O}_A^{(k+1)}, \mathcal{O}_b^{(k+1)}
    \xrightarrow{\mathrm{QLSS}}{W^{(k+1)}} = \mathcal{U}_\mathrm{QLSS}^{(k+1)}\ket{0}\bra{0}\mathcal{U}_\mathrm{QLSS}^{\dagger(k+1)}
    \xrightarrow{\mathcal{D}}
    {\Tilde{W}}^{(k+1)}=\mathrm{Tomo}\left[U_\mathcal{D}W^{(k+1)}U_\mathcal{D}^\dagger\right],
\end{equation}
wherein $\mathcal{E}$ is an encoding process for the $(k+1)$-th iteration's block-encoding and state preparation oracles $\mathcal{O}_A^{(k+1)} , \mathcal{O}_b^{(k+1)}$ given the $k$-th iteration's output ${\Tilde{W}}^{(k)}$, 
$\mathcal{U}_\mathrm{QLSS}$ is an HHL-QLSS to solve large-scale linear systems with $\Tilde{\mathcal{O}}(\kappa)$ queries of oracles $\mathcal{O}_A, \mathcal{O}_b$,
and $\mathcal{D}$ is a decoding process transforming the $(k+1)$-th QLSS's output state $W^{(k+1)}$ into the $(k+1)$-th classical tomoraphy ${\Tilde{W}}^{(k+1)}$.
For conventional QLSS, we have that $U_\mathcal{D}= \mathcal{I}$ and ${\Tilde{W}}^{(k)}$ is exactly the classical (normalized) solution  $W^{(k)}$. 
Formally, we define the iteration-tolerant bandwidth
\begin{equation}
        \mathcal{S}(A, b, \tau) \vcentcolon = \min_{\mathcal{E},\mathcal{D}}\sup_{0\leq k<\tau}\left\{ \mathrm{dim}_\mathbb{C}(\Tilde{W}^{(k)}_{\mathcal{E},\mathcal{D}}) \middle| \lvert W^{(\tau)}_{\mathcal{E},\mathcal{D}} - W_\mathrm{ideal} \rvert<\epsilon \right\}
    \end{equation}
to be the minimum number of independent components within ${\Tilde{W}}^{(k)}$ so that the end-to-end iterative QLSS converges with error bound $\epsilon$.
Further, we prove that the iteration-tolerant bandwidth is one critical parameter to the iterative QLSS's complexity:
\begin{theorem}[Iterative QLSS complexity lower bound]\label{thm:1}
    Given the iteration-tolerant bandwidth $\mathcal{S}$, the iterative quantum linear system solver's time complexity is lower-bounded by $\Tilde{\mathcal{O}}(\tau\kappa\mathrm{Poly}(\mathcal{S}, \frac{1}{\epsilon}))$.
    In particular, an exponential speed-up can not be achieved unless $\mathcal{S}\in\mathcal{O}(\log N)$.
\end{theorem}
\noindent Note that the iteration-tolerant bandwidth is equivalent to the efficient spectral sparsity of the working space, defined as
\begin{equation}
        \mathcal{S}_\mathrm{spectral}(A, b, \tau) \vcentcolon= \min_{\mathcal{B}}\sup_{0\leq k<\tau}\left\{ \left\lVert\mathrm{Spec}_{\mathcal{B}}({W}^{(k)})\right\rVert_0 \middle| \lvert W^{(\tau)}_{\mathcal{B}} - W_\mathrm{ideal} \rvert<\epsilon \right\},
\end{equation}
and hence is independent of the quantum algorithm.
Consequently, the expected exponential speed-up is achieved if and only if the encoding-and-decoding pair $(\mathcal{E}, \mathcal{D})$ is well-designed with low iteration-tolerant bandwidth $\mathcal{S}$.

On the other hand, we show that the NSE can be solved with an exponential speed-up on asymptotic complexity can be achieved for practical problems.
Herein, we upper bound the iteration-tolerant bandwidth by the underlying variable's spectral sparsity and propose a framework for designing $(\mathcal{E}, \mathcal{D})$ by exploitation of the spectral and structural information to reduce the end-to-end complexity:
\begin{theorem}[Quantum Navier-Stokes Solver, Informal]\label{thm:2}
    The Navier-Stokes equation discretized by the implicit Euler and finite volume methods on $N$ cells can be solved within $\Tilde{\mathcal{O}}\left( \left(\kappa D\mathcal{S}+\log^2 N\right)\frac{\tau\mathcal{S}}{\epsilon^2}\right)$ if these variables have an $\mathcal{S}$ time on $\mathcal{O}(\mathcal{S}D+\log N)$ qubits, where $\epsilon$ is the precision and the density $\rho$, velocities $u, v$, and the internal energy per unit mass $e$ have an $\mathcal{S}$ Fourier spectral sparsity.
    Alternatively, the circuit depth turns out to be $\Tilde{\mathcal{O}}\left(\frac{\tau\kappa D \mathcal{S}^2}{\epsilon^2}\mathrm{Poly}\log N\right)$ polynomial spectral sparsity.
\end{theorem}
\noindent In the input phase, our hierarchy spectral encoding method can blow up the Hilbert space of interest, from a low $\mathcal{O}(\mathcal{S})$-dimensional spectral space into a high $\mathcal{O}(N)$-dimensional working space for the quantum linear solver.
In the output phase, we design a corresponding sparse spectral tomography method to read out classical information for the next loop, as depicted in Fig.~\ref{fig:schematic}~(a).

    \begin{figure}
        \centering
        \begin{tikzpicture}
            \node{\includegraphics[width=\textwidth]{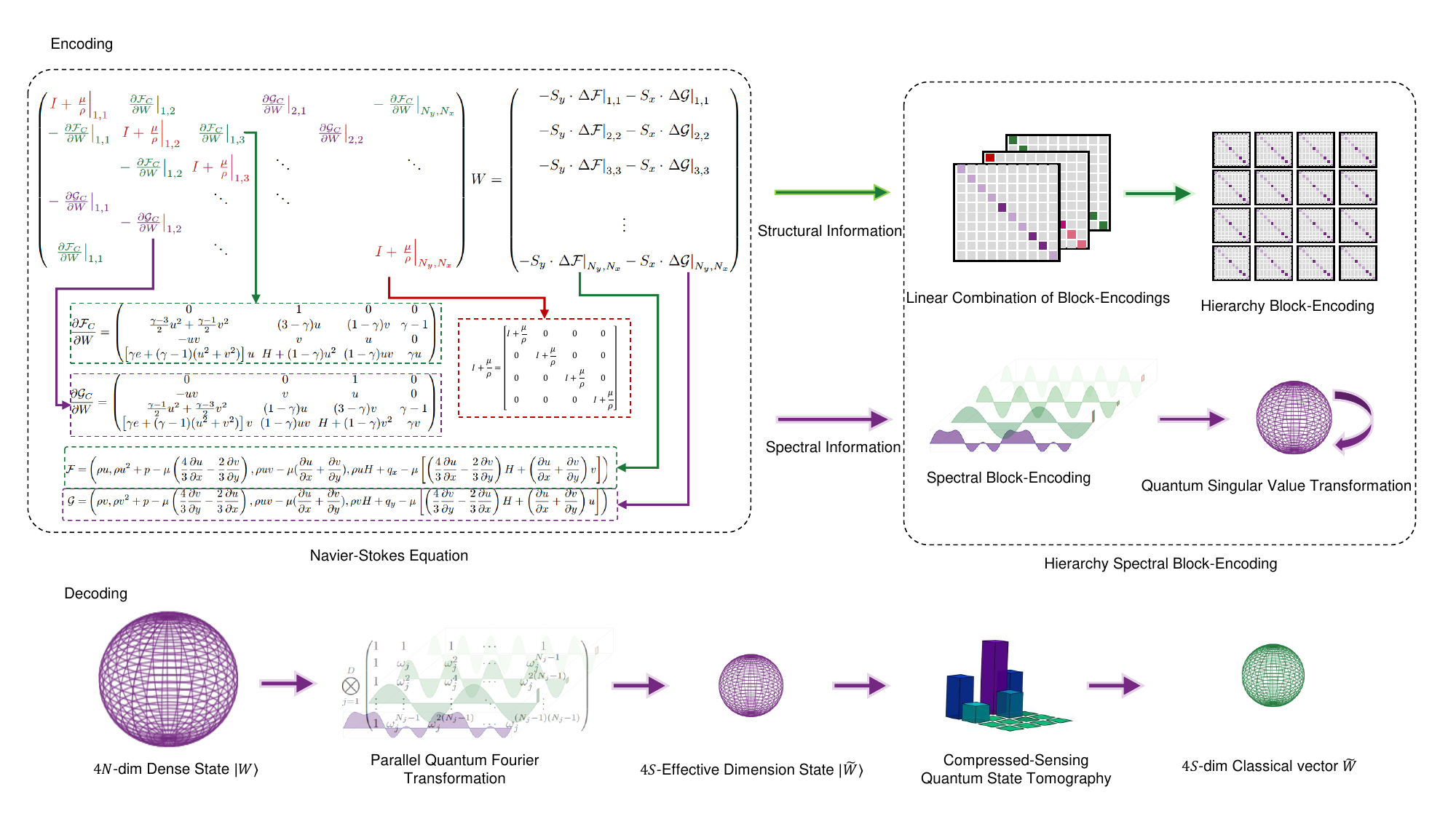}};
            \node[] at (-8.5, 4.51) {\textbf{a}};
            \node[] at (-8.4, -2.21) {\textbf{b}};
        \end{tikzpicture}
        \caption{Input and Output Schemes in Quantum Navier Stokes Solver: (a) Hierarchy Spectral Encoding. Left: the linear system $AW=b$ in $2$-dim NSE. The left-hand-side Jacobian matrix consists of five diagonals of $4\times4$ matrices. The red primary diagonal represents the numerical dispassion, and the green and purple sub- and super-diagonal correspond to the horizontal and vertical convective fluxes, respectively. Each block is evaluated at the cell labeled by the subscripts and consequently is a matrix-valued function on $N$ cells. The right-hand-side residual vector is also a weighted summation of four fluxes. Right: the matrix and vector's structural information is encoded by the hierarchy block-encoding. Both the matrix and vector are linear combinations of diagonals of sub-matrices. Each diagonal can be converted into $16$ submatrices of size $N\times N$ by \textit{conceptually} interchanging the cell index and the inter-block index. Each submatrix corresponds to a variable's value at $N$ cells. This spectral information can be encoded by the spectral block-encoding. (b) Sparse Spectral Decoding. A $4N$-dim dense state can be translated, by a parallel quantum Fourier transformation according to axis, into a state of low effective dimensions so that the state tomography can be applied to derive a classical vector for the next iteration.}
        \label{fig:algorithm}
    \end{figure}

We sketch the Hilbert space blow-up of the $4N$-dimensional linear system  in  Fig.~\ref{fig:algorithm}(a).
Given a linear system $AW=b$, HHL-QLSS usually utilizes an $(\alpha, a, 0)$-block-encoding of a non-unitary matrix $A$ as its upper left corner of an $(a+n)$-qubit extended unitary operator as that $A/\alpha=\bra{0}^{\otimes a}U_A\ket{0}^{\otimes a}$.
One main obstacle to block-encoding this Jacobian matrix $A$ is its complicated structure induced by the numerical method, comprising a blockwise multi-diagonal structure of diagonals $A^{(K)}  (-D\leq K\leq D)$ determined by topological connectivity and an inter-block structure dictated by the governing equation, as depicted in the upper left corner of Fig.~\ref{fig:algorithm}(a).
We address this issue by developing a hierarchy block-encoding
technique containing three levels as follows:
At the highest \textbf{diagonal-level}, we implement  a linear combination of $(2D+1)$ block-encodings followed by an arithmetic on the cell index $I, I'$  to shift each sub- and super-diagonal
\begin{equation}
    U_A = \sum_{K=-D}^{D} \left(\sum_I\ket{J+N_K}\bra{J}\right)U_A^{(K)}, 
\end{equation}
wherein $U_A^{(K)} $ is the block-encoding of the $K^{th}$ blockwise diagonal matrix $A^{(K)}$ and $N_K=\pm1, \pm N_x, ...$ is the corresponding shift, as depicted in the upper middle of Fig.~\ref{fig:algorithm}(a)
By conceptually interchanging the inter-block indices $(I, I')$ and the cell indices $(J, J')$, each $A^{(K)}$ can be regarded as $(D+2)^2$ sub-matrices $A^{(K)}_{I, I'} \in \mathbb{R}^{N\times N}$, wherein each sub-matrix is a function of cell index $J, J'$ on $N\times N$ sites so that $A^{(K)}_{I, I'}(J, J') = A^{(K)}_{I, I'; J, J'}$, as depicted in the upper right corner of Fig.~\ref{fig:algorithm}(a).
Consequently, we can implement a block-encoding unitary pair $U_L^{(K)}$ and $U_R^{\dagger(K)}$ of $\alpha^{(K)}$ at the \textbf{inter-block level}, and then insert a sequence of controlled block-encoding unitaries $C_{I, I'}-U^{(K)}_{I, I'}$ at the lowest \textbf{element-level} as
\begin{equation}
    U_A^{(K)} = U_R^{\dagger(K)} \left( \prod_{I, I' = 1}^{D+2}  C_{I, I'}-U^{(K)}_{I, I'} \right) U_L^{(K)},
\end{equation}
wherein $U^{(K)}_{I, I'}$ is some $(\alpha^{(K)}_{I, I'}, a, 0)$-block-encoding of $A^{(K)}_{I, I'}$ and
\begin{equation}
    \alpha^{(K)} = \begin{pmatrix}
   \alpha^{(K)}_{I, I'} 
\end{pmatrix}_{(D+2)\times(D+2)}
\end{equation}
is the normalization constant matrix.

The remaining obstacle to block-encoding $A$ is to encode those input and (intermediate) variables appearing in each $A^{(K)}_{I, I'} \in \mathbb{R}^{N\times N}$.
We apply sparse spectral block-encoding (SS-BE) to overcome this by exploiting the problem-originated spectral structure so that our concerned variables can be efficiently approximated in the subspace spanned by at most $\mathcal{S}\in \mathcal{O}(\mathrm{Poly}(\log{N}))$ basis, as depicted in the lower right corner of Fig.~\ref{fig:algorithm}(a).
By assumption, the independent input variables $\rho, u, v, e$ can be approximated by a linear combination of at most $\mathcal{S}$ spectral basis.
Consequently, those dependent variables derived by compactness-preserving operations (e.g., polynomials and partial derivatives), including total enthalpy $H = \gamma e + \frac{1}{2}u^2 +\frac{1}{2}v^2$, temperature $T=\frac{\gamma-1}{R}e$, etc., preserve the spectral sparsity upper bounded by $\mathcal{S}^d$ with polynomial degree $d\leq3$ (see Supplementary Information).
Such a linear combination of spectral bases can be block-encoded by
\begin{equation}
    U_\mathrm{SS-BE}^\mathrm{C} = (P_R^\dagger\otimes I_n)\left(\prod_{k=1}^\mathcal{S} C_k-U_{\mathrm{BE}(F_k)}\right)(P_L\otimes I_n),
\end{equation}
wherein $U_{\mathrm{BE}(F_k)}$ are \textit{static} block-encoding unitaries of basis function $F_k$  loading the public information and $\left( P_L, P_R^\dagger\right)$ is a state preparation pair to inject the \emph{varying} spectral information.
We prove in Supplementary Information that, for the common Fourier or monomial spectra, the SS-BE can be accomplished within $\mathcal{O}(\mathcal{S}\log N)$ time.
As for other variables derived by non-compact operators, such as the viscosity $\mu=\frac{T^{1.5}}{T+T_\mathrm{S}}$, the finite spectral sparsity is no longer guaranteed.
Alternatively, we approximate such a non-compact operator by a $d'$-degree polynomial $P'$ and utilize quantum singular value transformation 
\begin{equation}
    U_\mathrm{SS-BE}^\mathrm{NC} = QSVT\left(U_\mathrm{SS-BE}^\mathrm{C}, \phi_{P'}\right)
\end{equation}
to block-encode these variables in $\mathcal{O}(\mathcal{S}d')$ depth.
We highlight that the circuit depth and ancillary qubits depend mainly on spectral sparsity $\mathcal{S}$, removing the linear dependency on problem size $N$.
Further, non-compact operators only appear in the non-linear (numerical) viscosity terms of the governing equation, releasing an interesting theoretical relation between the quantum algorithm complexity theory and the classical PDE theory.
We also emphasize that the sparse spectra can be updated classically in $\mathcal{O}(\mathcal{S})$ time, and the polynomial approximation of non-compact operators can be pre-computed and is irrelevant to the varying spectra.
A direct consequence is the classical computational efficiency.
Equipped with these two block-encoding methods, the Jacobian matrix $A$ can be directly block-encoded, and the residual vector $b$ can be prepared by conducting a similar diagonal block-encoding circuit on the superposition of the computational basis.
The circuit depth for both these two subroutines are $\mathcal{O}( D \mathcal{S}\log\frac{1}{\epsilon}\log\log N)$.
To also consider the query complexity in quantum linear solvers, the total circuit depth turns out to be $\mathcal{O}(\kappa D \mathcal{S}\log\frac{1}{\epsilon}\log\log N)$.

Since the solution $W$ is encoded in a dense quantum state, a direct state tomography turns out to be too expensive.
Unlike previous works, our sampling complexities can also be reduced under the spectral sparsity condition as illustrated in Fig.~\ref{fig:algorithm}(b).
Indeed, since the efficient spectral sparsity of elements in $W$ is upper bounded by $\mathcal{O}(\mathcal{S})$,  the $4N$-dimensional dense vector $W$ can be translated into a vector $\Tilde{W}$ of $4\mathcal{S}$ sparsity by a parallel quantum Fourier transformation
\begin{equation}
    \left(\bigotimes_{j=1}^D QFT_{n_j}\right)\ket{W}=\ket{\Tilde{W}},
\end{equation}
where $QFT_{n_j}$ acts on the quantum register corresponding to the cell's $j$-th index.
Consequently, the sampling complexity of sparse spectral state tomography scales as $\Tilde{\mathcal{O}}(\mathcal{S}/\epsilon^2)$ instead of $\Tilde{\mathcal{O}}(N/\epsilon^2)$.
To consider both circuit depth and sampling complexity, we accomplish the end-to-end quantum Navier-Stokes solver with time complexity bounded by $\Tilde{\mathcal{O}}\left( \left(\kappa D\mathcal{S}+\log^2 N\right)\frac{\tau\mathcal{S}}{\epsilon^2}\right)$.
Notably, in some practical CFD problems, the boundary condition can be more complicated, and the residual vector $b$ can be violent with  discontinuities, also known as the Riemann problem. 
Our block-encoding method still applies to these scenarios through a quantum boundary condition mechanism or a quantum Riemann solver, as detailed in Supplementary Information.

\subsection{Quantum Resource Optimization}

    \begin{figure}
        \centering
        \begin{tikzpicture}
            \node{\includegraphics[width=\textwidth]{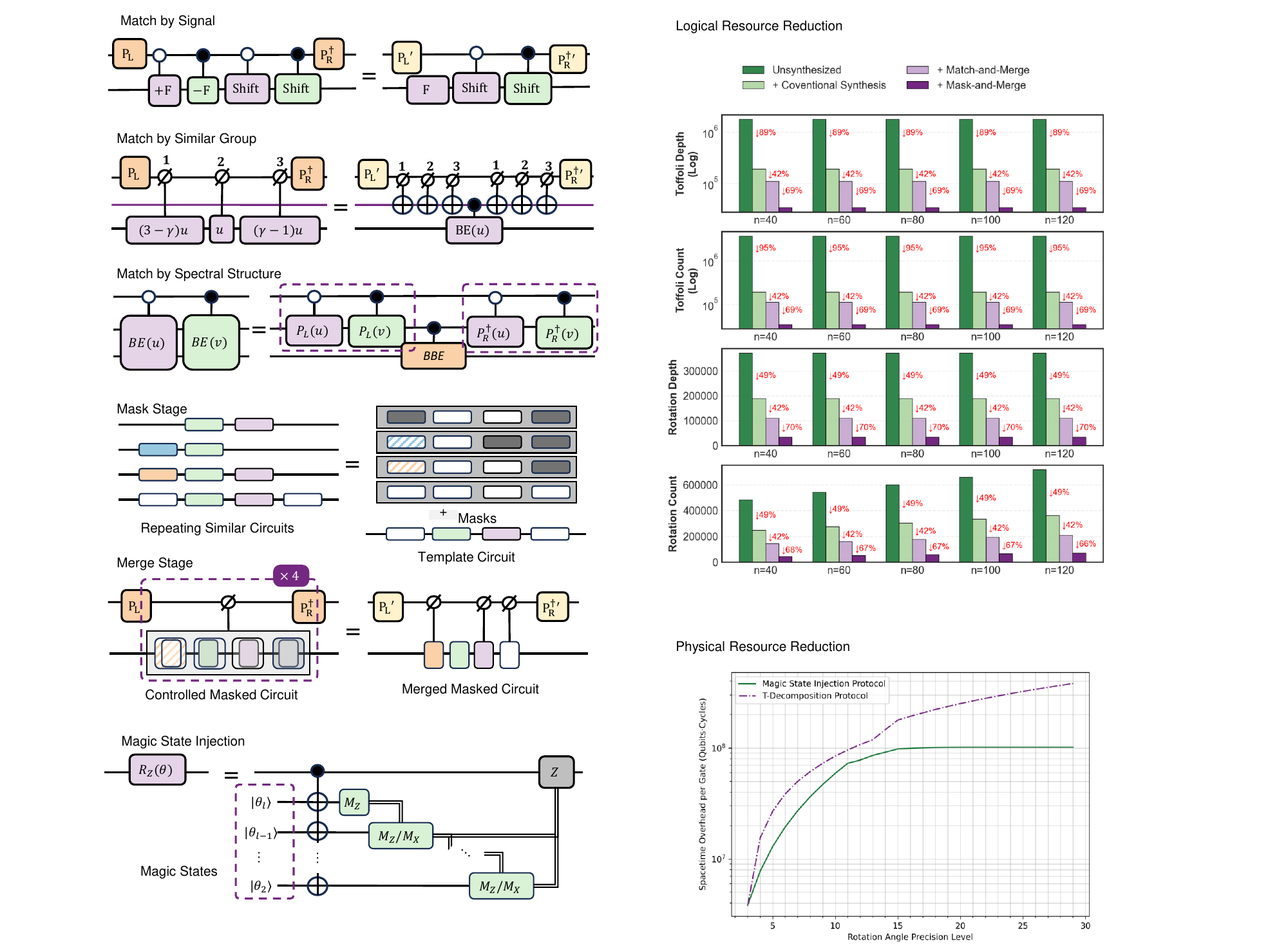}};
            \node[] at (-7.5, 6.35) {\textbf{a}};
            \node[] at (-7.5, 4.72) {\textbf{b}};
            \node[] at (-7.5, 2.83) {\textbf{c}};
            \node[] at (-7.5, 0.95) {\textbf{d}};
            \node[] at (-7.5, -1.2) {\textbf{e}};
            \node[] at (-7.5, -3.69) {\textbf{f}};
            \node[] at (0.15, 6.25) {\textbf{g}};
            \node[] at (0.15, -2.35) {\textbf{h}};
        \end{tikzpicture}
        \caption{Logical and Physical Resource Reduction: (a-c) Match-and-Merge synthesis by signal, similar group, and spectral structure patterns.  (d-e) Mask-and-Merge synthesis: (d) In the mask stage, repeating circuits with similar structure are regarded as a public template circuit plus varying masks; (e) In the merge stage, $4$ controlled masked circuits are merged, wherein redundant sub-circuits and controls are removed. 
        (f) The rotation gate with magic states injection utilized in this work. (g) Logical resource reduction for Toffoli depth, Toffoli count, rotation depth, and rotation count. (h) Physical resource reduction of the magic state injection protocol in comparison to the conventional T-decomposition protocol.}
        \label{fig:resource}
    \end{figure}
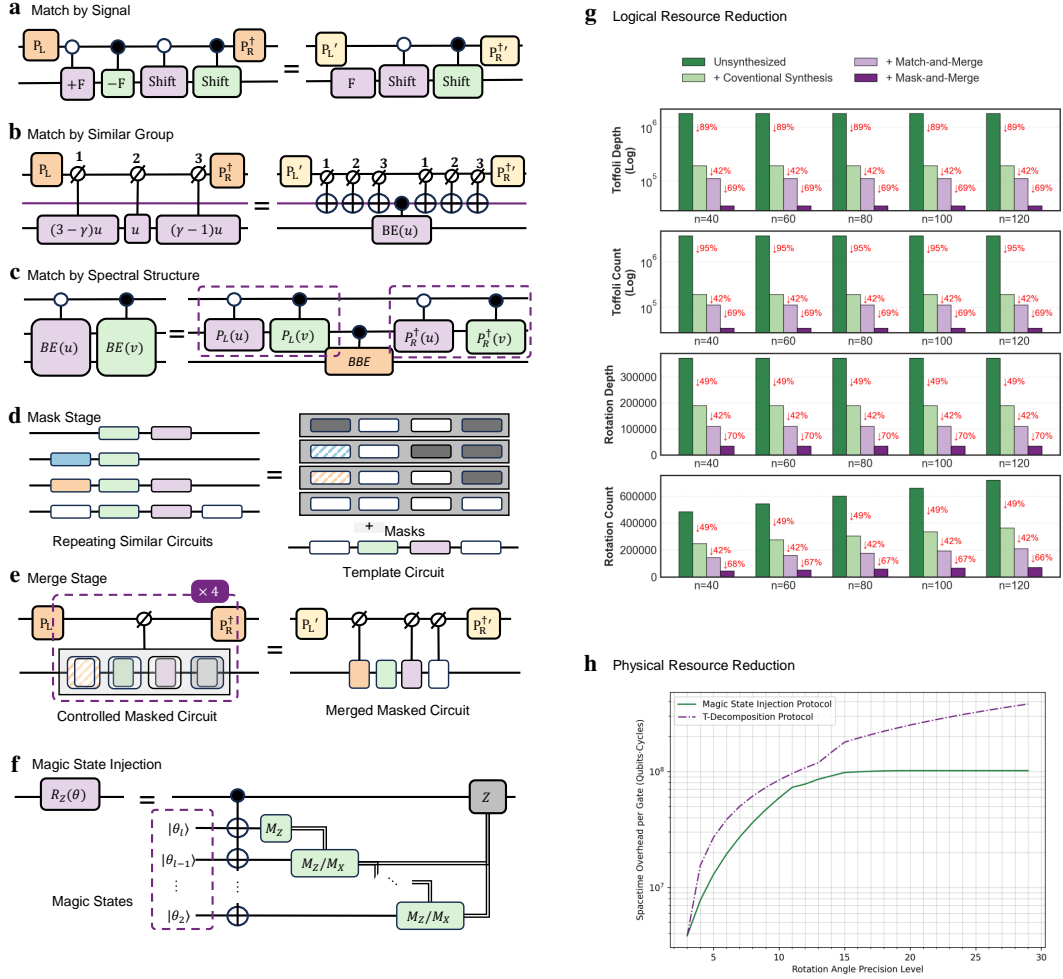
While our end-to-end algorithm is asymptotically efficient, reaching the crosspoint still needs more careful quantum circuit synthesis and QEC protocol to generate a logical-resource-efficient circuit and to implement a physical-resource-efficient deployment, respectively. 
In the conventional FTQC context, the logical quantum resource is characterized by the number of logical qubits and the count and depth of T gates.
However, recent works suggest that rotation gates can be implemented more efficiently through an ancilla state injection, instead of a T-based gate compilation~\cite{campbell2016efficient}.
Henceforth, we first consider the circuit synthesis targeted to reduce the logical resource of rotation and Toffoli counts/depths, and then consider the QEC implementation to minimize the physical resources in terms of cycles and physical qubit number.

For logical resource reduction, conventional circuit synthesis techniques often utilize \textbf{local symmetries} among a group of gates to remove redundant gates and unnecessary multi-control behaviour.
For example, we use uniformly controlled rotation (UCR) to reduce the Toffoli gate count in state preparation subroutines and use the fan-out $X$ gate, which is native for surface code lattice surgery, to in Fourier basis block-encoding subroutines~\cite{mottonen2004transformation, fowler2018low}.
Besides conventional problem-agnostic synthesis techniques, we highlight that the utilization of global symmetries among a group of subcircuits can further remove many costly subcircuits or their control qubits by exploiting public information in the algorithm.
More specifically, we develop two new circuit synthesis techniques, including the match-and-merge to directly use \textbf{intrinsic symmetries} between subroutines and the mask-and-merge to adaptively modify the subroutines for \textbf{artificial symmetries}.

The match-and-merge method matches subroutines by symmetry patterns and then merges them accordingly. 
The first pattern is that two or more subroutines differ only by a global sign and are organized by a linear combination of unitaries, as depicted in Fig.~\ref{fig:resource}(a).
For this pattern, we can interchange the subroutine order and merge the symmetrical pair into a single unitary, and then alternate the corresponding amplitude's sign in the state preparation pair.
The second pattern is that two or more subroutines differ only by a  strength coefficient and are organized by a linear combination of unitaries, as depicted in Fig.~\ref{fig:resource}(b).
For these unitaries within a common similar group, we can introduce group flag ancillaries to match their corresponding block-encoding circuits into groups, and then merge the subroutines within the same group to remove redundant controls and subroutines.
For example, the elements in each horizontal or vertical flux Jacobian can be classified into seven groups of $\Tilde{u}=\{u, \gamma u, (3-\gamma)u\}$, $\Tilde{v}= \{v, (1-\gamma)v\}$, etc.
For each group, the $U_\mathrm{BBE}$ is called only once, and the $U_\mathrm{SP/CSP}(\mathcal{S}(\cdot))$ are merged to remove one control qubit for each Toffoli gate.
The third pattern is more specific to our hierarchy spectral block-encoding.
Note that our block-encoding circuit can be decomposed into two classes of subroutines: those (controlled) state preparation subroutines $\left(P_L,P_R^\dagger\right)$ to load the variable spectrum, and the problem-agnostic basis function block-encoding subroutines $U_\mathrm{BBE}$.
Particularly, for those elements with common spectral structure yet different spectra, as illustrated in Fig.~\ref{fig:resource}(c).
For this pattern, we can remove a redundant $U_\mathrm{BBE}$ as well as its control qubits, and the remaining $\left(P_L,P_R^\dagger\right)$ unitaries can also be merged by further UCR technique, as circled in Fig.~\ref{fig:resource}(c).
In the quantum Navier-Stokes solver, these patterns appear when block-encoding flux components of symmetrical interfaces and variables of similar spectrums. 
For potential applications, they can also be found in other partial derivative equation numerical solution problems with spatial discretization and Hamiltonian simulation problems with molecular orbital integrals of different spins~\cite{helgaker2013molecular}.

The mask-and-merge method is a more complicated yet powerful generalization that can synthesize quite different subroutines.
The basic idea is to embed these subroutines into a least common circuit named \textit{template circuit} so that we can mask or modify several parts therein to adaptively generate the desired subroutine, and then merge these masked template circuits into an integration.
In the mask stage depicted in Fig.~\ref{fig:resource}(d), prior information is exploited to loosen the former symmetry requirement:
Those gates that only appear in several subroutines (purple-colored) can be turned on or off by resetting the rotation angles in the state preparation pair.
Those subroutines with identical Toffoli and rotation structure, yet different rotation angles configuration (blue- and orange- colored), can be masked by additional controlled rotation gates afterwards.
In this way, we construct artificial symmetry among different subroutines and translate them into a common template circuit that only differs by their rotation angles.
In the merge stage depicted in Fig.~\ref{fig:resource} (e), since these subroutines are symmetrical, we can apply the aforementioned techniques to remove redundant subroutines and costly control behaviors.
In extensive numerical experiments, we make $22.72\times$ and $49.98\times$ reductions on Toffoli gate depth and count, as well as $5.00\times$ reductions on both rotation gate depth and count, as depicted in Fig.~\ref{fig:resource}(g).
Consequently, the logical quantum resources and time are reduced by more than an order of magnitude.

For physical resource reduction, we apply lattice surgery to implement Clifford operations and gate teleportation with ancillary magic states for non-Clifford gates~\cite{litinski2019magic}.
In particular, the rotation gates with ancillary rotation states can be implemented in one non-Clifford layer instead of more than one hundred layers of $T$ gates, as illustrated in Fig.~\ref{fig:resource}(f), reducing the time by two orders of magnitude.
A direct result is that the accumulative logical error in the circuit is also suppressed by two orders of magnitude, inducing a smaller code distance and much fewer physical qubits for the circuit. 
Further, we refine the rotation state distillation protocol with a further $3.5\times$ reduction on the magic state factory's physical qubit number when compared to the conventional T-decomposition protocol, as depicted in Fig.~\ref{fig:resource}(h)~\cite{mlti}.

\subsection{Characterizing practical quantum advantage}

In principle, the optimal quantum resources to solve an NSE are determined by a multi-objective mixed-integer constrained programming problem
\begin{equation}
    \min_{\theta\in\Theta}\left\{(\mathcal{QN}_\mathrm{physics}(\theta), \mathcal{T}_\mathrm{quantum}(\theta))|\epsilon_\mathrm{quantum}(\theta, N)<\epsilon_\mathrm{threshold}(N)\right\},
\end{equation}
wherein $\theta$ are the hyperparameters involved in the framework, $\epsilon_\mathrm{quantum}$ is the end-to-end error in each iteration of quantum N-S solver and $\epsilon_\mathrm{threshold}$ is the iteration-tolerant error threshold for our problem. 

To validate the quantum N-S solver and to characterize problem-specified hyperparameters, we implement extensive numerical experiments of varying problem size $N$ and efficient spectral sparsity $\mathcal{S}$. We consider the two-dimensional Taylor–Green vortex in our numerical experiments. The flow field is initialized with a constant temperature initial condition, and the numerical solutions obtained from our quantum N-S solver are systematically compared against the available incompressible analytical solution.
The numerical experiment suggests that our algorithm's solution succeeds in converging to the benchmarking solution quickly as the error decays:
The fluid field information can be faithfully encoded and processed by our algorithm as depicted in Fig.~\ref{fig:numerical}(a).
The final solution's relative error rate converges to zero quickly as the output quantum state's error rate decays, as illustrated in Fig.~\ref{fig:numerical}(b).
In particular, given an end-to-end error bound, the single-iteration error threshold turns out to be loose as the problem size grows, as a partial result of that a larger problem size has less discretization error, higher spatial resolution, and mitigated numerical diffusion.
Further, this law promises a robust (and even conservative) estimation of the iteration-tolerant error threshold $\epsilon_\mathrm{threshold}$ for large $N$.

In the numerical simulation, we monitor those problem-specific hyperparameters $\theta_\mathrm{problem}$ for a reasonable extrapolation in large-scale cases.
The condition number is a key factor in circuit depth,  and accumulative block-encoding and logical error.
The numerical experiment shows that the condition number $\kappa$ saturates as the problem size $N$ grows once the spectral sparsity $\mathcal{S}$ is fixed, as illustrated in Fig.~\ref{fig:numerical}(c). This \textit{$\kappa$-saturation} phenomenon not only highlights our method's scalability for large-scale cases but also  reveals a scaling law between the variable's spectral sparsity and the Jacobian matrix's condition number that may have its own interest in CFD from a quantum perspective.
We also track those normalization constants in the quantum N-S solver, as depicted in Fig.~\ref{fig:numerical}(d), which plays a key role in block-encoding and state preparation error bounds.
The numerical stability suggests a reasonable estimation of these normalization constants for large $N$ cases.

We consider the end-to-end error analysis and allocation of three fundamental sources: 
The algorithmic error $\epsilon_\mathrm{alg}$ from the polynomial approximation of nonlinear behavior, including the imperfect kernel-reflection/projection (eigen-filtering) in QLSS and the accumulative block-encoding and state preparation error of the viscosity term.
The deployment error $\epsilon_\mathrm{deploy}$ is caused mainly by the accumulative logical error in the circuit and the ancillary magic state preparation error.
And the tomographic error $\epsilon_\mathrm{tomography}$ is a direct result of the finite sampling number to reconstruct the quantum state.
Due to the discontinuous and multiple dependencies of these errors  on the remaining quantum N-S solver-relevant hyperparameters $\theta_\mathrm{solver}$, the evaluation process is computationally costly, and we apply a heuristic error allocation to minimize the search and optimization overhead, as detailed in Supplementary Information.

Based on problem size $N$, problem-specified hyperparameters $\epsilon_\mathrm{threshold}$ and $\theta_\mathrm{problem}$ estimated by numerical experiments, and free hyperparameters $\theta_\mathrm{solver}$ optimized by error allocation process, we get a precise estimation of the end-to-end time and quantum resources of the quantum Navier-Stokes solver that a $N=2^{80}$-grid $2$-dimensional discretized Navier-Stokes equation can be solved within
\begin{equation}
    \mathcal{T}_\mathrm{quantum} = \mathcal{N}_\mathrm{sample}\times\mathcal{D}\times d\times\mathcal{T}_\mathrm{cycle}\leq3.68\times10^6 \text{ s}\simeq42.6\text{ days}
\end{equation}
time, utilizing  less than
\begin{equation}
    \mathcal{QN}_\mathrm{physics} = \mathcal{QN}_\mathrm{circuit} + \mathcal{QN}_\mathrm{routing} + \mathcal{QN}_\mathrm{factory} \leq 8.71\times10^6
\end{equation}
physical qubits of $5\times 10^{-4}$-error-rate,
wherein $\mathcal{N}$, $\mathcal{D}$, $d$, and $\mathcal{T}_\mathrm{cycle}$ are the sampling number, non-Clifford depth, surface code distance, and the standard QEC per-cycle time on a superconducting quantum processor, respectively. The classical counterpart time is estimated in Supplementary Information to be
\begin{equation}
    \mathcal{T}_\mathrm{classical}=\mathcal{C}/R_\mathrm{max}\simeq129.85\text{ years},
\end{equation}
where $\mathcal{C}$ is the total floating-point operations, and $R_\mathrm{max}=1742.00\text{ PFlop}\cdot\text{s}^{-1}$ is the maximum LINPACK performance achieved for current supercomputer~\cite{top500}.
Suppose a perfect parallel execution, achieving theoretical peak performance $R_\mathrm{peak}=2746.38\text{ PFlop}\cdot\text{s}^{-1}$ with no diminishing marginal returns effect, we still require $720$ copies of the most powerful supercomputer available today to complete the same task within an equal time.

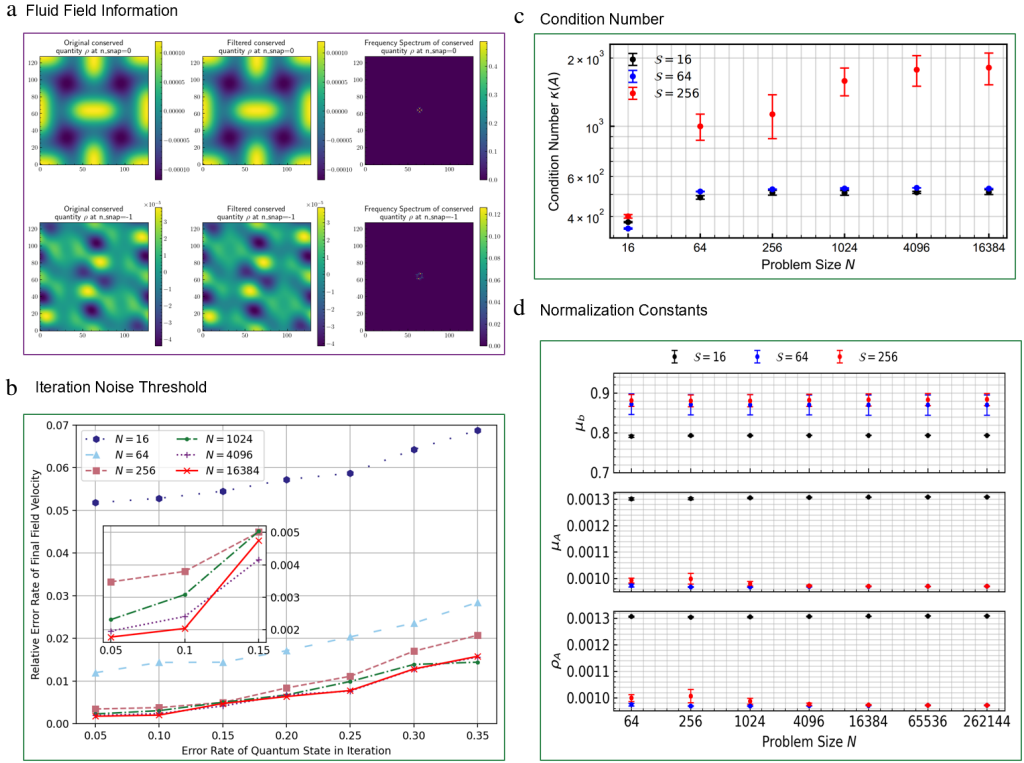
\begin{figure}
    \centering
    \begin{tikzpicture}
        \node[]{\includegraphics[width=0.75\textwidth]{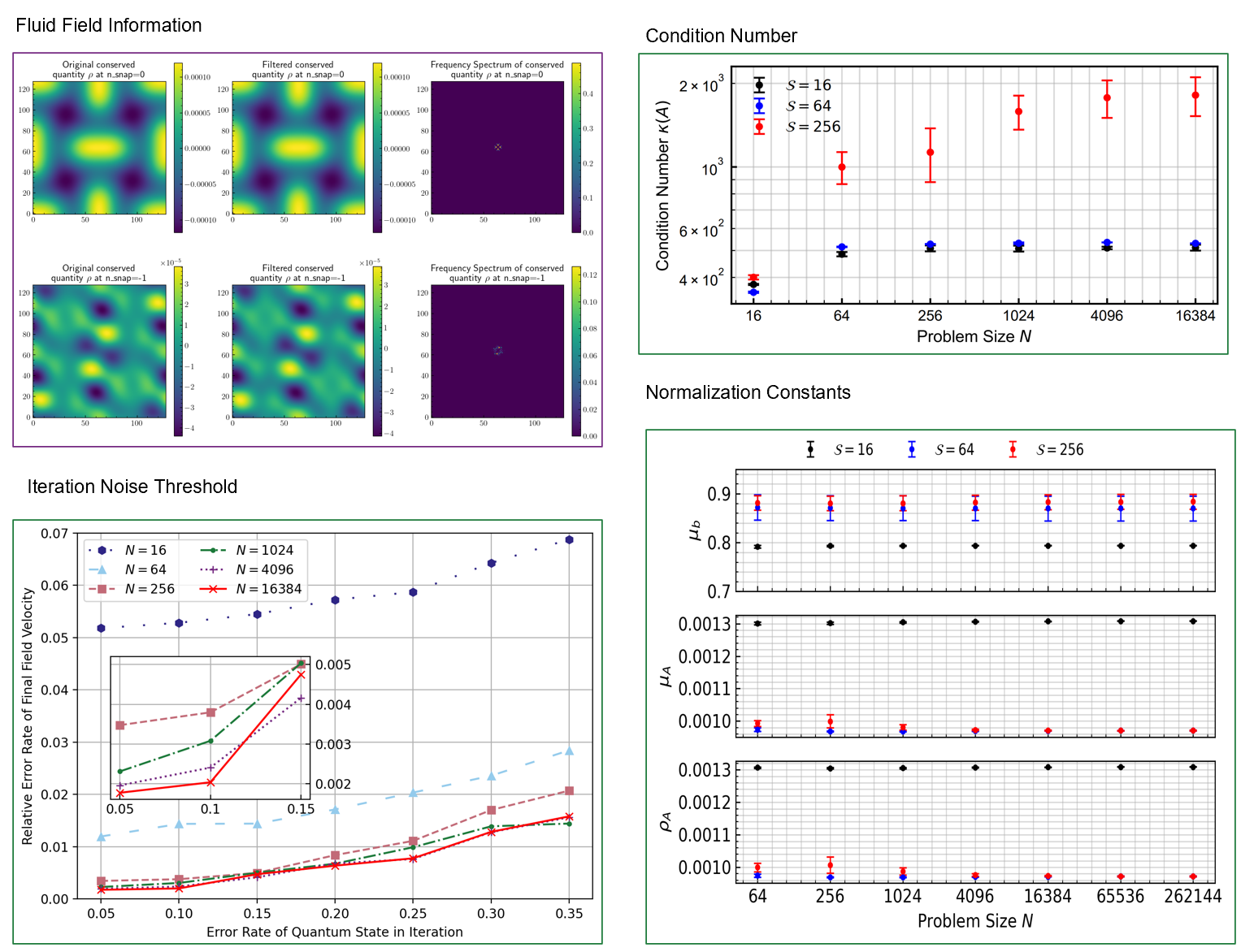}};
        \node[] at (-6.75, 4.9) {a};
        \node[] at (-6.75, -0.11) {b};
        \node[] at (0, 4.8) {c};
        \node[] at (0, 0.95) {d};
    \end{tikzpicture}
    
    \caption{
    Numerical experiments for the quantum Navier-Stokes solver. (a) The fluid field information at the initial (upside) and final iteration (downside). From left to right are the original field, the approximated field of given spectral sparsity, and the spectra in the frequency domain, respectively.
    (b) The relative error rate of the final iteration's field velocity converges to zero quickly as the quantum state error decays.
    (c) The condition number saturates when $N$ grows, and the saturation level is determined by the efficient spectral sparsity.
    (d) The normalization constants are stably bounded when the problem size $N$ grows.
    }
    \label{fig:numerical}
\end{figure}

\section{Discussion}

In this work, we have proposed a comprehensive framework for solving practical computational fluid dynamics (CFD) problems on a fault-tolerant quantum computer, supported by a rigorous end-to-end resource analysis. We derived a general complexity lower bound for iterative quantum linear system solvers, governed by the efficient spectral sparsity of the problem, and designed an algorithm that achieves an exponential quantum speedup for the Navier-Stokes equations (NSE) through a hierarchical spectral input method and a sparse spectral output protocol. Furthermore, we developed global symmetry-based circuit synthesis techniques (match-mask-and-merge) and a refined quantum error correction (QEC) protocol to drastically reduce logical and physical resource demands. Our numerical validation characterized critical hyperparameters for large-scale instances and incorporated a full error analysis. The results demonstrate that solving NSE on a $2^{80}$-grid is feasible within 42.6 days using 8.71 million physical qubits at an error rate of $5\times 10^{-4}$. We present this finding as compelling evidence for a practical quantum advantage in CFD and high-performance scientific computing.

It is important to emphasize that our conclusion of a practical quantum advantage rests on several key assumptions, which also delineate the limitations of this study. First, we assume an efficient spectral sparsity that grows substantially slower than the problem size. Although this sparsity is fixed in our specific case, its prevalence in other practical scenarios requires further numerical validation. The complexity lower bound established here implies that a quantum advantage may be unattainable for problems lacking such an efficient spectral representation. Second, our efficient quantum state tomography protocol relies on numerically validated sample numbers rather than a theoretical upper bound, which is too loose for practical purposes. While we have robustly tested this numerically for our case, establishing a general theoretical bound with a tighter prefactor for large spectral sparsity $\mathcal{S}$ remains an open problem for future work.

From a quantum algorithms perspective, our block-encoding and circuit synthesis methods are of independent interest. They offer a promising approach for loading real-world data into quantum states, with potential applications in quantum machine learning, financial modeling, and beyond. Notably, our algorithm's end-to-end complexity is dictated by spectral sparsity—a distinct and often more favorable scaling than the conventional matrix sparsity that governs the query complexity of HHL-style linear solvers.

Our resource analysis underscores that circuit synthesis is critical not only for reducing logical resources but also for minimizing physical overheads as a key insight for achieving practical quantum advantage. The symmetry-based techniques introduced here could be adapted to block-encoding circuits in quantum computational chemistry and Hamiltonian simulation, though their generalization to other domains warrants further rigorous analysis and numerical support. Furthermore, our results highlight the rotation gate count as a major bottleneck determining physical qubit requirements. Since rotation gates implemented via ancillary state injection hold great potential for resource reduction, developing more efficient magic state distillation factories is essential for broader practicality.

From a CFD perspective, our numerical experiments reveal a suggestive correlation between the spectral sparsity of flow variables and the condition number of the Jacobian matrix, which is a known determinant of computational cost in both classical and quantum solvers. While the role of spectral sparsity in CFD has been underexplored, our work highlights its importance. We note that realistic CFD scenarios often involve greater complexity, such as turbulent boundary conditions, unstructured grids, and high-order numerical schemes, which may introduce additional costs but also present exploitable structures (see Supplementary Information). We posit that the identification and utilization of such problem-specific intrinsic structures will be a central task for achieving practical quantum advantage in the fault-tolerant era.

\section*{acknowledgement}
This work has been supported by the National Key Research and Development Program of China (Grant No. 2023YFB4502500).

\appendix

\section{Related Works}

\subsection{Related Works on Block-Encoding}

\subparagraph{On block-encoding of structured matrix.} 

Ref.~\cite{clader2023quantum} studies the block-encoding resources of an arbitrary matrix based on a qRAM architecture.
Ref.~\cite{wan2021block} considers an efficient block-encoding of matrices with the displacement structures of circulant, Toeplitz, and Hankel types based on either a black-box model or a QRAM data structure model.
However, as we have discussed in the main text, qRAM-based block-encoding consumes lots of logical qubits and is often too costly for large-scale classical data.

Ref.~\cite{nguyen2022block} utilizes hierarchical matrices to block-encode displacement kernel matrices. Both their work and ours use a divide-and-conquer approach to block-encode a structured matrix, but the suitable structure is quite different. 
The block-encoding in Ref.~\cite{nguyen2022block} divides a dense and full-rank kernel matrix into hierarchically off-diagonal blocks and approximates each block by low-rank approximation, and this method is natural for displacement kernel matrices such as N-body gravitational force problems.
Our hierarchy block-encoding divides the large matrix into a sequence of equal-sized sub-matrices and block-encodes each block precisely, and our method is general for numerical PDE problems with local topological connectivity, such as fluid dynamics problems.
It will be very interesting  to combine their fine block-encoding techniques with ours for more complicated and fancy structured matrices, and we leave this as future work.

Ref.~\cite{camps2024explicit} shows an explicit quantum circuit to block encode a banded circulant matrix, an extended binary tree, and a symmetric stochastic matrix.
There are mainly two differences between this article and our hierarchy block-encoding: Firstly, this work shows how to block encode a multi-diagonal structure matrix, wherein the elements lying on the same diagonal are exactly the same. However, this method does not apply to our case since in our case the elements are not constant but a function evaluated on $N$ sites for large $N$. Secondly, while both this work and our hierarchy block-encoding utilize an arithmetic approach to implement the diagonal's shift, we reduce its quantum resource under control by a refined circuit synthesis.

Ref.~\cite{sunderhauf2024block} designs an efficient block-encoding algorithm for matrices of structure computable by arithmetic, including Checkerboard matrix, Toeplitz matrix, tridiagonal matrix, extended binary tree matrix, and 2-dimensional Laplacian matrix. Therein, they apply multiplexed rotation gates to block-encode $D$ different elements within $\mathcal{O}(D)$ circuit depth.
Nevertheless, for CFD applications, there is usually no guarantee that the flux Jacobian matrix only contains a very few number of different elements. Indeed, an $N$-grid discretization of the fluid field often has $\Theta(N)$ different elements, and hence the block-encoding method in Ref.~\cite{sunderhauf2024block} turns out to be inefficient for a QLSS algorithm.

\subparagraph{On spectral-based block-encoding.}
In our previous work~\cite{zhuang2024statistics}, we developed a general state preparation framework based on the maximum-entropy principle.
Therein, we consider the preparation of the weighted mixture of a family of statistical distributions with a latent space to encode their coefficients and a work space to encode those \textit{static} components.
Indeed, this weighted distribution mixture is an efficient block encoding in the spectral space, where the distribution family is exactly the underlying basis, and the latent space therein corresponds to the spectral space.
Also, Ref.~\cite{rosenkranz2025quantum} considers the multivariable function preparation problem, wherein they also construct an efficient block-encoding of the weighted sum of Fourier-basis and Chebyshev-basis.
In particular, our spectral block-encoding can be regarded as a reformulation and combination of these two works on state preparation from an end-to-end and spectral method perspective.

It should be mentioned that, although these two works are powerful ingredients in many quantum algorithms, they can not be directly applied to CFD and many other scientific numerical computing problems. This is mainly because the matrix in these problems often contains numerous independent variables, and there is usually no common spectral structure or statistical distribution family among these variables. A structured block-encoding architecture that can flexibly assemble varying spectra and basis block-encoding is still required. Moreover, for an end-to-end application, it is also fundamental to design a corresponding state tomography protocol as we did in this work.

\subsection{Related Works on Practical Quantum Advantage in (Early) Fault-Tolerant Quantum Computing Era}
\subparagraph{On (early) fault-tolerant quantum computing.}
We are inspired by and benefited from a sequence of works to characterize practical quantum advantage in the (early) fault-tolerant quantum computing era~\cite{katabarwa2024early}.
Many works have been proposed to give a rigorous resource analysis for Shor's algorithm, one of the most celebrated quantum algorithms with potential exponential advantage to factor large  RSA integers  (see Refs.~\cite{gidney2021factor, gidney2025factor} and the references therein). Refs.~\cite{gidney2021factor, gidney2025factor} demonstrate an end-to-end quantum resource analysis framework including quantum error correction overhead. More efficient quantum error correction protocols like magic state distillation and cultivation have been developed to significantly reduce the quantum resources for non-Clifford quantum gates~\cite{fowler2012time, gidney2019efficient, gidney2024magic}.
In the deployment and quantum resource optimization phase of this work, we utilize these improved T-factory designs.
In the resource analysis phase of this work, we follow the quantum error correction resources analysis framework therein for a rigorous physical qubit number analysis.
In fact, we regard these efforts as bringing the practical application of fault-tolerant quantum applications much closer to reality.

\subparagraph{On end-to-end resource analysis of HHL-based quantum interior point method portfolio optimization ~\cite{dalzell2023end}.}
The previous work on the end-to-end quantum resource analysis for quantum interior point methods with applications to portfolio optimization (QIPM-PO) is one of the earliest end-to-end resource analysis for practical HHL-like applications, establishing a standard analysis framework for logical quantum resource and providing rigorous theoretical and numerical analysis, including the quantum resource to block-encode a matrix by qRAM, to prepare quantum states, to solve the quantum linear system, to tomograph the quantum state, and to synthesize the non-Clifford gates~\cite{dalzell2023end}. Based on numerical extrapolation of historical stock data, it is reported by Alexander M. Dalzell et al. that $8\times10^6$ logical qubits and $7\times 10^{29}$ $T$ gates are required to solve the portfolio optimization problem for $100$ stocks. Their analysis suggests that fundamental improvements on large prefactors, costly quantum state tomography, etc., are required for practical quantum advantage.

While we are strongly inspired, in the resource analysis phase, by Ref.~\cite{dalzell2023end}, our result seems to be more positive. We believe that this significant difference is a consequence of the recent progress of the whole quantum computing community, and can be summarized as follows:
Firstly, in the algorithm phase, we design a qRAM-free input model and a corresponding low effective-dimension output model. In QIPM-PO, most of the logical quantum resources, including the logical qubits and the $T$ gates, are consumed by the qRAM-based block-encoding architecture to load an arbitrary and large-scale matrix. Indeed, the quantum resources characterized by $T$ gate count and logical qubit number both grow linearly with the problem's size and are multiplied by a sequence of large prefactors. Besides, we apply a more efficient QLSS that significantly reduces the large prefactor in QLSS~\cite{dalzell2024shortcut}. Moreover, since the QIPM-PO assumes zero prior knowledge of the problem, the output state is always dense and requires a sampling number that seems unaffordable for practical quantum advantage. This can be viewed as a direct result of the Holevo bound, as we analyzed earlier. By the hierarchy spectral block-encoding techniques developed in this work, the logical qubit number and non-Clifford gate count are both exponentially reduced. Further, by the sparse spectral tomography, the amount of information that needs to be extracted from the quantum states is also significantly reduced. We believe this improved encoding/decoding pair is one dominating contribution to the reduction of quantum resources. In particular, an asymptotic exponential quantum speed-up is excluded by QIPM-PO due to the linear-growing sampling complexity to determine the signatures, but is promised by our spectral-based quantum N-S solver.
Secondly, in the deployment phase, we implement an in-depth circuit synthesis that significantly reduces the logical quantum resources by more than two orders of magnitude. This result supports the previous work's opinion that a large constant prefactor is one critical obstacle to practical quantum advantage. And this also highlights the importance of an efficient circuit synthesis protocol for FTQC applications.
Thirdly, our quantum error correction protocol admits a further time and quantum resources reduction by more than two orders of magnitude. We utilize a gate teleportation with ancillary rotation states, rather than a more costly $T$ gate decomposition protocol, to implement the non-Clifford rotation gates. Our numerical analysis shows that this will reduce the non-Clifford layer number by approximately two orders of magnitude. Besides, 

Based on these improvements, we establish a more complete resource analysis framework to also consider the physical quantum resources. In our analysis, not only algorithmic and tomographic errors but also accumulative logical errors and imperfect magic states' errors are considered. We have a more precise estimation of the surface code distance and magic state factory size to characterize the condition for  practical and fault-tolerant quantum advantage.

\subparagraph{On sparse oracle-based HHL resource analysis~\cite{tu2025towards}.}
In the preparation of this work, we also noticed this sparse oracle-based HHL resource analysis that gives an estimation of the potential Fault-tolerant advantage on space, time, and energy. 
In this work, it is reported that a potential quantum advantage on the HHL solver can be achieved in $10^6$ seconds on $10^5$ physical qubits with error bounded by $P_\mathrm{physics}<1\times10^{-5}$ to solve a linear system of condition number $\kappa\simeq10\sim100$ and dimension $N\simeq2^{33}\sim2^{48}$.
In comparison, our resource analysis suggests a practical quantum advantage within $4\times10^6$ seconds on $4\times10^6$ physical qubits with error bounded by $P_\mathrm{physics}<5\times10^{-4}$ on $N\simeq2^{50}$.
The main difference in logical resources between our work and this one is the input model. In Ref.~\cite{tu2025towards}, assume a sparse oracle $M$ that can output the digital encoding of the value and position as
\begin{equation}\nonumber
    \ket{x}\ket{0}\ket{0}\ket{0}\xrightarrow{M}\ket{x}\ket{m(x)}\ket{0}\ket{w(x)}.
\end{equation}
In this work, we give an explicit quantum circuit construction of such an oracle with reduced quantum resources.
Besides, the difference in physical resources estimation is partially because we pose a looser limitation on the physical error rate of qubits. The physical qubit number can be significantly reduced in our work if the physical error rate is bounded by $P_\mathrm{physics}<1\times10^{-5}$.
We follow the classical computing time estimation in Ref.~\cite{tu2025towards}.
The energy estimation therein is also of fundamental interest from an environmental and economic perspective. We are pleased to consider the possible energy reduction in the context of quantum computational fluid dynamics.

\subsection{Related Works on Quantum Computational Fluid Dynamics.}

\subparagraph{On digital quantum Navier-Stokes solver~\cite{gaitan2020finding}.}
Frank Gaitan proposed a quantum algorithm to solve the Navier-Stokes equation with application to find the steady-state inviscid, compressible flow through a convergent-divergent nozzle~\cite{gaitan2020finding}. Therein, Ref.~\cite{gaitan2020finding} transforms the original NSE into a set of ordinary derivative equations and then solves it by  a quantum amplitude estimation algorithm to evaluate the numerical integration average.
In comparison to Frank's early yet excellent work on a quantum Navier-Stokes solver, there are two main differences in this work from a quantum algorithm perspective:
Firstly, Ref.~\cite{gaitan2020finding} considers a quantum oracle as the input model, as detailed in Eq.~(50) of its supplementary information, where the additional and usually costly quantum resources overhead is not considered.
In this work, we consider an end-to-end quantum algorithm with an explicit quantum circuit input model instead. Consequently, it is plausible to give a concrete quantum resource estimation.
Secondly, Ref.~\cite{gaitan2020finding} admits an at-most quadratic speedup with
\begin{equation}
    \begin{split}
        \mathrm{comp}^\mathrm{random}&=\mathcal{O}\left(\left(\frac{1}{\epsilon}\right)^{\frac{1}{1/2+q-\gamma}}\right),\\
        \mathrm{comp}^\mathrm{quantum}&=\mathcal{O}\left(\left(\frac{1}{\epsilon}\right)^{\frac{1}{1+q-\gamma}}\right),
    \end{split}
\end{equation}
at the limit case of $q,\gamma\ll1$,
while our algorithm ensures an exponential speedup.
One possible reason is that we apply the more powerful HHL-QLSS instead of quantum amplitude estimation only.

\subparagraph{On quantum simulation of Navier-Stokes flows}

There are several elegant works on transforming the original Navier-Stokes equation and many other partial derivative equations into a Hamiltonian simulation problem.
In Ref.~\cite{itani2024quantum}, Wael Itani, Katepalli R. Sreenivasan, and Sauro Succi design a framework to bridge the nonlinear dynamics to the bosonic modes evolution, as well as a Carleman linearization truncation in the bosonic Fock space.
In Ref.~\cite{meng2024quantum}, Zhaoyuan Meng and Yue Yang derive a quantum spin representation for the Navier-Stokes equation.
Therein, they derive the mapping to Schr\"{o}dinger-Pauli equation to facilitate quantum computing of fluid dynamics through Hamiltonian simulation.
Recently, Shi Jin, Nana Liu, and Yue Yu established a novel, powerful framework named Schr\"{o}dingerization to transform linear ordinary derivative equations and partial derivative equations into a system of Schr\"{o}dinger’s equations by the warped phase transformation~\cite{jin2023quantum, jin2024quantum}. As noted in Ref.~\cite{jin2024quantum}, it has no explicit complexity dependency on the condition number. While the Navier-Stokes equation is nonlinear, it is interesting to consider whether this framework can be extended to the nonlinear case.
Since block-encoding of the Hamiltonian is one key ingredient in quantum computational Hamiltonian simulation, we have confidence that these quantum algorithms can be enhanced by our hierarchy spectral block-encoding method.
It will be very interesting to combine these transformations and our block-encoding techniques for an improved or even optimal quantum simulation on the NSE and many other partial derivative equations.

The Navier-Stokes equation can also be simulated on an annealing quantum device.
In Ref.~\cite{ray2019towards}, N.Ray et al.  translate the Navier-Stokes equation into a quadratic unconstrained binary optimization problem and solve this problem on an adiabatic annealing-based quantum device.
Besides, Ref.~\cite{meena2024towards} studies the digital quantum simulation of the Navier-Stokes equation by a matrix-product-state (MPS) method. Therein, they apply MPS states to encode the linear and non-linear terms, as well as a sequence of MPS operators with a Krylov subspace method to simulate the non-linear dynamics. We believe that the MPS method is a powerful tool to solve complicated CFD problems, including the Navier-Stokes equation. However, since there is a sequence of complicated optimization procedures in both Ref.~\cite{ray2019towards} and Ref.~\cite{meena2024towards}, it is hard to derive a theoretical complexity upper bound for comparison.

\section{Iterative QLSS: Generalized Definition, Iteration-Tolerant Bandwidth and Complexity Lower Bound}
Solving linear systems of equations is of fundamental importance and general interest to both the science and engineering communities.
Given a matrix $A\in\mathbb{C}^{M\times N}$ and a vector $b\in\mathbb{C}^M$, the solution (if it exists) is supposed to be the vector $x\in\mathbb{C}^N$ such that $Ax=b$.
While the fastest classical linear solver can numerically solve this equation within $\mathcal{O}(\mathrm{Poly}(N))$, quantum linear systems solvers (QLSS) have been proposed to solve large linear systems within $\mathcal{O}(n\kappa\log{(1/\epsilon)})$, promising a potentially exponential speedup given efficient access to the matrix $A$ and vector $b$~\cite{morales2024quantum,harrow2009quantum,ambainis2010variable,childs2017quantum,subacsi2019quantum,costa2022optimal,an2022quantum,jennings2023efficient,costa2023discrete,dalzell2024shortcut,low2024quantum}.
Herein, $n=\log_2{(\max{M, N})}$ is the logarithm of the system's variable number, $\kappa$ is the condition number, and $\epsilon$ is the precision of the solution.
Aram Harrow, Avinatan Hassidim, and Seth Lloyd propose the first QLSS, by combining the Hamiltonian simulation algorithm and quantum phase estimation subroutine to achieve an exponential speedup over its classical counterpart ~\cite{harrow2009quantum}.
Then the query complexity dependence on the condition number $\kappa$ and the precision $\epsilon$ is improved by variable-time amplitude amplification and the linear combination of unitaries ~\cite{ambainis2010variable,childs2017quantum}.
Based on adiabatic quantum computing, the complexity is further improved to be $\mathcal{O}(\kappa\log{(\frac{1}{\epsilon})})$,  and is proven to be optimal ~\cite{subacsi2019quantum,costa2022optimal}.
Following this, subsequent works are proposed to reduce the constant prefactor from an order of $10^5$ to less than $10^2$~\cite{an2022quantum,jennings2023efficient,costa2023discrete}.
Specifically, an augmented linear system is introduced to avoid the difficult-to-analyze adiabatic path-following and to realize a low query complexity $\sim 80\kappa$ for $(\kappa, \epsilon) = (10^5, 10^{-10})$~\cite{dalzell2024shortcut}.
Since this shortcut QLSS is easy to analyze with a small prefactor, we will utilize it in this work.
\begin{atheorem}[{Shortcut QLSS, Theorem 4 in Ref.~\cite{dalzell2024shortcut}}]\label{theorem: qlss}
    
    Suppose that $\left\lVert b \right\rVert = 1$ and that all nonzero singular values of $A\in\mathbb{R}^{N \times N}$ lie in the interval $[\kappa^{-1}, 1]$, as that $x$ is the unique solution satisfying $Ax=b$. Further suppose an $(\alpha, a, 0)$-block-encoding of $A$ to be
    \begin{equation}\label{eq:block_encoding}
        \frac{A}{\alpha} = (\bra{0}^{\otimes a}\otimes I_{2^n}) U_A (\ket{0}^{\otimes a}\otimes I_{2^n})
    \end{equation}
    and a state preparation unitary for $b$ to be
    \begin{equation}\label{eq:state_preparation}
        U_b (\ket{0}^{\otimes n}) = \sum_{j=0}^{2^n-1}b_j\ket{j},
    \end{equation}
    with $n=\log_2{N}$.
    Then the numerical solution $\Tilde{x}$ with error bounded by $\frac{1}{2}\left\lVert \ket{\Tilde{x}}\bra{\Tilde{x}} - \ket{x}\bra{x} \right\rVert_1\leq\epsilon$ can be numerically solved by $\mathcal{Q}$ queries to (controlled) $U_A$ or $U_A^\dagger$ and $2\mathcal{Q}$ queries to (controlled) $U_b$ or $U_b^\dagger$ with
    \begin{equation}
        \mathcal{Q} \leq 56.0\kappa + 1.05\kappa\ln{\left( \frac{\sqrt{1-\epsilon^2}}{\epsilon} \right)} + 2.78{\ln{(\kappa)}}^3 + 3.17.
    \end{equation}
\end{atheorem}
\noindent In a typical quantum linear solver,  there are usually two kinds of input models to translate classical data into quantum information. The first input model, state preparation, prepares an amplitude-encoded quantum state for a real-valued vector.
Formally speaking, we have the following definition of state preparation:
\begin{adefinition}[{State Preparation}]
     Given an $2^n$-length vector b, the state preparation subroutine $U_b$ is an $(n+a)$-qubit unitary operator acting on initial state $\ket{0}^{\otimes (n+a)}$ to prepare the amplitude-encoded state
    \begin{equation}
        U_b\ket{0}^{\otimes (n+a)}=\frac{1}{\mathcal{N}_b}\ket{0}^{\otimes a}_\mathrm{anc}\otimes(\sum_{j=0}^{2^n-1}b_j\ket{j}),
    \end{equation}
    where $\mathcal{N}_b=\sqrt{\sum_{j=0}^{2^n-1}b_j^2}$ is the normalization constant.
\end{adefinition}
\noindent Preparing a general quantum state is proven hard as that exponentially growing quantum resource is consumed~\cite{zhang2022quantum,sun2023asymptotically,gui2024spacetime}.
It has been proven that an efficient implementation of state preparation in $\mathcal{O}(n)$ time preserves the potential for exponential speedup, yet necessitates $\Omega{(2^n)}$ qubits~\cite{zhang2022quantum, gui2024spacetime}.
Araujo et. al. introduce $\mathcal{O}(2^n)$ ancillary qubits to store the 'pre-rotated' angles for state preparation~\cite{araujo2021divide}. 
Clader et. al. utilize a flag mechanism to uncompute the garbage~\cite{clader2022quantum}.
A parallel copy and swap mechanism is proposed to also reduce the Clifford depth to $\mathcal{O}(n)$,
and the overall consumption is quantitatively characterized by spacetime allocation, the sum of the active-qubit number in each layer~\cite{gui2024spacetime}.
The second input model, block-encoding, implements a unitary operator to encode an arbitrary matrix as its upper left block.
Formally speaking, we have the following definition of block-encoding of a matrix:
\begin{adefinition}[{Block-Encoding}]
     Given an $2^n\times2^n$ matrix $A$, the $(n+a)$-qubit unitary operator $U_A$ is an $(\alpha, a, 0)$-block-encoding of $A$ if
    \begin{equation}
        U_A = \begin{pmatrix}
            \frac{A}{\alpha} & *\\
            * & * 
        \end{pmatrix}_{2^{n+a}\times2^{n+a}},
    \end{equation}
    where $\alpha$ is the subnormalization constant and $a$ is the ancillary qubit number.
\end{adefinition}
\noindent To the best of our knowledge, to block-encode a general matrix is costly as $\mathcal{O}(2^{2n})$ qubits and gates are required to block-encode a $2^{n}\times 2^{n}$ matrix ~\cite{clader2022quantum}.
Further, a simple calculation shows that the quantum resource can not be significantly improved:
preparing an arbitrary $n$-qubit quantum state with $a$ ancillary qubits by the operator $U_\mathrm{prepare}\ket{0}^{\otimes (n+a)} = \ket{\psi}\ket{0}^{\otimes a}$ can be viewed as a $(1, a, 0)$ block-encoding of some matrix 
\begin{equation*}
    \begin{pmatrix}
    \psi_0 & \psi_1 & \cdots & \psi_{2^n-1}\\
    * & * & \cdots & * \\
    \vdots & \vdots & \ddots & \vdots\\
    * & * & \cdots & *
\end{pmatrix}_{2^{n}\times 2^{n}},
\end{equation*}
and hence the block-encoding spacetime allocation is also lower bounded by $\Omega(2^n)$.
Besides input models, a typical QLSS also utilizes a quantum state tomography to extract the classical information from the QLSS's output states. It is proven that at least $\Tilde{\mathcal{O}}(d/\epsilon^2)$ copies of states or $\Tilde{\mathcal{O}}(d/\epsilon)$ queries of oracles have to be utilized to derive an estimation of a $D$-dimensional state with $l_2$-error bounded by $\epsilon$. Consequently, even with exponentially many qubits, the time complexity still grows linearly with system size $2^n$ when the solution is dense~\cite{dalzell2023end}.

The quantum resource and time overhead can be even worse for practical scenarios, such as computational fluid dynamics and structural mechanics, wherein the underlying dynamics can be non-linear or ill-conditioned, and an iterative method is introduced.
Formally, we study 
\begin{adefinition}[Iterative linear system]
    An iterative linear system $(\overrightarrow{A}, \overrightarrow{b}, \tau)$ is a sequence of linear systems $A^{(k)}W^{(k)}=b^{(k)} (0\leq k<\tau)$, where $A^{(0)}$ and $b^{(0)}$ are initial value and $A^{(k+1)}$ and $b^{(k+1)}$ are uniquely determined by the $k$-th solution $W^{(k)}$.
\end{adefinition}
\noindent This iterative system can be solved by an end-to-end hybrid quantum algorithm defined as 
\begin{adefinition}[Iterative quantum linear system model]
    Suppose a $\tau$-cycle iterative linear system $(\overrightarrow{A}, \overrightarrow{b}, \tau)$,
    then an iterative quantum linear system model consists of, for each iteration, (1) an encoding process to block encode $A^{(k+1)}$ and prepare $b^{(k+1)}$ $\mathcal{E}:{\Tilde{W}}^{(k)}\xrightarrow{}\mathcal{O}_A^{(k+1)}, \mathcal{O}_A^{(k+1)}$, (2) a quantum linear system solver $U_\mathrm{QLSS}: \mathcal{O}_A^{(k+1)}, \mathcal{O}_b^{(k+1)}\xrightarrow{} {W^{(k+1)}} =\mathcal{U}_\mathrm{QLSS}^{(k+1)}\ket{0}\bra{0}\mathcal{U}_\mathrm{QLSS}^{\dagger(k+1)}$, and (3) a decoding process to $\mathcal{D}:{W^{(k+1)}}\xrightarrow{}{\Tilde{W}}^{(k+1)}=\mathrm{Tomo}\left[U_\mathcal{D}W^{(k+1)}U_\mathcal{D}^\dagger\right]$.
\end{adefinition}
\noindent Note that we consider a more general model rather than the typical iterative quantum linear system, where $\mathcal{D}$ is a quantum state tomography protocol and ${\Tilde{W}}^{(k)}$ is indeed a classical vector estimation of $W^{(k+1)}$.
Herein, we allow an additional decoding process to also take any potential efficient \textit{quantum post-process} method into consideration.
In this model, we quantitatively characterize the restriction between quantum and classical information, including both input and output processes, by a uniform definition
\begin{adefinition}[Iteration-tolerant bandwidth]\label{def:bandwidth}
    Given a $\tau$-iteration linear system (A, b), we define its iteration-tolerant bandwidth to be the minimal supremum of complex dimensionality of $\Tilde{W}^{(k)}_{\mathcal{E},\mathcal{D}}$ over all possible encoding and decoding pairs $(\mathcal{E},\mathcal{D})$ so that the final solution $W^{(\tau)}_{\mathcal{E},\mathcal{D}}$'s $l_2$-error is bounded by $\epsilon$:
    \begin{equation}
        \mathcal{S}(A, b, \tau) = \min_{\mathcal{E},\mathcal{D}}\sup_{0\leq k<\tau}\left\{ \mathrm{dim}_\mathbb{C}(\Tilde{W}^{(k)}_{\mathcal{E},\mathcal{D}}) \middle| \lvert W^{(\tau)}_{\mathcal{E},\mathcal{D}} - W_\mathrm{ideal} \rvert<\epsilon \right\}
    \end{equation}
    wherein the subscripts of $\Tilde{W}$ is for the specified encoding-and-decoding pair and $W_\mathrm{ideal}$ is the normalized target solution. 
\end{adefinition}
\noindent We emphasize that the iteration-tolerant bandwidth is indeed defined for an iterative linear system $(A, b, \tau)$ rather than for a specified iterative QLSS algorithm $(\mathcal{E}, U_\mathrm{QLSS}, \mathcal{D})$. Moreover, while Definition~\ref{def:bandwidth} is defined in the quantum computing context, it is indeed equivalent to the following \textit{classical-only} definition:
\begin{adefinition}[Efficient Spectral Sparsity]
    Given a $\tau$-iteration linear system (A, b), we define its efficient spectral sparsity to be the minimal supremum of spectral sparsity of ${W}^{(k)}$ over all possible spectral basis $\mathcal{B}$ so that the final solution $W^{(\tau)}$'s $l_2$-error is bounded by $\epsilon$:
    \begin{equation}
        \mathcal{S}_\mathrm{spectral}(A, b, \tau) = \min_{\mathcal{B}}\sup_{0\leq k<\tau}\left\{ \left\lVert\mathrm{Spec}_{\mathcal{B}}({W}^{(k)})\right\rVert_0 \middle| \lvert W^{(\tau)}_{\mathcal{B}} - W_\mathrm{ideal} \rvert<\epsilon \right\}
    \end{equation}
    wherein $\mathrm{Spec}_{\mathcal{B}}(\cdot)$ is the $\mathcal{B}$-basis spectrum, $\left\lVert\cdot\right\rVert_0$ is the $l_0$-norm. 
\end{adefinition}
\noindent The equivalence is formulated as:
\begin{alemma}[Equivalance of efficient iteration-tolerant bandwidth and efficient spectral sparsity]
    Given a $\tau$-iteration linear system (A, b), we have:
    \begin{equation}
        \mathcal{S}(A, b, \tau) = \mathcal{S}_\mathrm{spectral}(A, b, \tau).
    \end{equation}
\end{alemma}
\begin{proof}
    On the one hand, for an arbitrary encoding and decoding pair $(\mathcal{E},\mathcal{D})$, we consider the corresponding supereme
    \begin{equation}
        \mathcal{S}(\mathcal{E},\mathcal{D}) = \sup_{0\leq k<\tau}\left\{ \mathrm{dim}_\mathbb{C}(\Tilde{W}^{(k)}_{\mathcal{E},\mathcal{D}}) \middle| \lvert W^{(\tau)}_{\mathcal{E},\mathcal{D}} - W_\mathrm{ideal} \rvert<\epsilon \right\}.
    \end{equation}
    By definition, we have that 
    \begin{equation}
        \mathrm{dim}_\mathbb{C}(\Tilde{W}^{(k)}_{\mathcal{E},\mathcal{D}}) \leq \mathcal{S}(\mathcal{E},\mathcal{D})
    \end{equation}
    holds for all $0\leq k<\tau$.
    Consequently, there are orthogonal vectors $\{\Tilde{w}_1, \cdots, \Tilde{w}_{\mathcal{S}(\mathcal{E},\mathcal{D})}\}$ so that  $\Tilde{W}^{(k)}$ lies in the subspace spanned by them.
    Since $U_\mathcal{D}$ is a unitary, ${W}^{(k)}$ lies in the subspace spanned by linear-independent vectors 
    \begin{equation}
        \{U_\mathcal{D}^\dagger \Tilde{w}_1 U_\mathcal{D}, \cdots, U_\mathcal{D}^\dagger \Tilde{w}_{\mathcal{S}(\mathcal{E},\mathcal{D})} U_\mathcal{D}  \}.
    \end{equation}
    We can extend this to a complete basis $\mathcal{B}$ so that
    \begin{equation}
        \left\lVert\mathrm{Spec}_{\mathcal{B}}({W}^{(k)})\right\rVert_0 \leq \mathcal{S}(\mathcal{E},\mathcal{D})
    \end{equation}
    also holds or all $0\leq k<\tau$, and so does their supereme.
    By the definition of efficient spectral sparsity, we have that
    \begin{equation}
        \mathcal{S}_\mathrm{spectral}(A, b, \tau) \leq \sup_{0\leq k<\tau}\left\{ \left\lVert\mathrm{Spec}_{\mathcal{B}}({W}^{(k)})\right\rVert_0 \middle| \lvert W^{(\tau)}_{\mathcal{B}} - W_\mathrm{ideal} \rvert<\epsilon \right\} \leq \mathcal{S}(\mathcal{E},\mathcal{D}).
    \end{equation}
    Due to the arbitrariness of the choice of $(\mathcal{E},\mathcal{D})$, we have that
    \begin{equation}
        \mathcal{S}_\mathrm{spectral}(A, b, \tau) \leq \min_{\mathcal{E},\mathcal{D}} \mathcal{S}(\mathcal{E},\mathcal{D}) = \mathcal{S}(A, b, \tau).
    \end{equation}
    On the other hand, for an arbitray basis $\mathcal{B}$, we consider the following supereme
    \begin{equation}
        \mathcal{S}_\mathrm{Spec}(\mathcal{B})=\sup_{0\leq k<\tau}\left\{ \left\lVert\mathrm{Spec}_{\mathcal{B}}({W}^{(k)})\right\rVert_0 \middle| \lvert W^{(\tau)}_{\mathcal{B}} - W_\mathrm{ideal} \rvert<\epsilon \right\}.
    \end{equation}
    Henceforth, there are orthogonal states in $\mathcal{B} = \mathrm{span}\{w_1, \cdots, w_{\mathcal{S}_\mathrm{spectral}(\mathcal{B})}, \cdots\}$ so that  ${W}^{(k)}$ lies in the subspace spanned by the first $\mathcal{S}_\mathrm{spectral}(\mathcal{B})$ basis. 
    Then the encoding process $\mathcal{E}$ can be constructed by a block-encoding unitary and a state preparation unitary for $\mathcal{B}$, for example, by a polynomial or Chebyshev expansion as in the later sections.
    By applying the basis transformation unitary
    \begin{equation}
        U_\mathcal{D} = \sum_{j} \ket{j}\bra{w_j},
    \end{equation}
    we can also derive the corresponding output states $\Tilde{W}^{(k)}$. It is easy to verify that their dimensionality and their supremum are also bounded by 
    \begin{equation}
        \mathrm{dim}_\mathbb{C}(\Tilde{W}^{(k)}) \leq \mathcal{S}_\mathrm{spectral}(\mathcal{B}) (0\leq k< \tau).
    \end{equation}
    By definition of efficient iteration-tolerant bandwidth, we have that
    \begin{equation}
        \mathcal{S}(A, b, \tau) \leq \sup_{0\leq k<\tau}\left\{ \mathrm{dim}_\mathbb{C}(\Tilde{W}^{(k)}_{\mathcal{E},\mathcal{D}}) \middle| \lvert W^{(\tau)}_{\mathcal{E},\mathcal{D}} - W_\mathrm{ideal} \rvert<\epsilon \right\} \leq \mathcal{S}_\mathrm{spectral}(\mathcal{B}).
    \end{equation}
    Since the choice of basis $\mathcal{B}$ is arbitrary, we have that 
    \begin{equation}
        \mathcal{S}(A, b, \tau) \leq \min_{\mathcal{B}} \mathcal{S}_\mathrm{spectral}(\mathcal{B}) = \mathcal{S}_\mathrm{spectral}(A, b, \tau).
    \end{equation}
\end{proof}
\noindent Given these definitions, we can now prove:
\begin{proof}[Proof of Theorem~\ref{thm:1}]
    We prove the first claim by contradiction and suppose that we have an iterative QLSS $(\mathcal{E}, U_\mathrm{QLSS}, \mathcal{D})$ with end-to-end time complexity significantly lower than $\Tilde{\mathcal{O}}(\tau\kappa\mathrm{Poly}(\mathcal{S}, \frac{1}{\epsilon}))$. Since in each iteration the linear system is explicitly dependent on the previous iteration's solution, the time complexity of each iteration should be significantly lower than $\Tilde{\mathcal{O}}(\kappa\mathrm{Poly}(\mathcal{S}, \frac{1}{\epsilon}))$. 
    Further, it is reported that a quantum algorithm must make at least $\Omega(\kappa\log{(1/\epsilon)})$ queries in general to solve a quantum linear system problem~\cite{costa2022optimal, dalzell2024shortcut}.
    So the sampling times for each iteration should be significantly less than
    \begin{equation}\label{aeq:contradiction}
        \mathcal{N}_\mathrm{sampling}\ll \Tilde{\mathcal{O}}(\mathrm{Poly}(\mathcal{S}, \frac{1}{\epsilon})).
    \end{equation}
    On the other hand, by the definition of iteration-tolerant bandwidth, the supremum of complex dimensionality satisfies 
    \begin{equation}
        \sup_{0\leq k<\tau}\left\{ \mathrm{dim}_\mathbb{C}(\Tilde{W}^{(k)}_{\mathcal{E},\mathcal{D}}) \middle| \lvert W^{(\tau)}_{\mathcal{E},\mathcal{D}} - W_\mathrm{ideal} \rvert<\epsilon \right\} \geq \mathcal{S}.
    \end{equation}
    Consequently, we have at least one iteration $k_0$ whose output has complex dimensionality satisfying
    \begin{equation}
        \mathrm{dim}_\mathbb{C}(\Tilde{W}^{(k_0)}_{\mathcal{E},\mathcal{D}})\geq \mathcal{S}.
    \end{equation}
    In other words, there are $\mathcal{O}(\mathcal{S})$ non-zero complex components within the random vector $\Tilde{W}^{(k_0)}$ from the tomography of the corresponding $\mathcal{O}(\mathcal{S})$-dimensional quantum state $U_\mathcal{D}W^{(k+1)}U_\mathcal{D}^\dagger$. 
    By proposition 49 of Ref.~\cite{van2023quantum}, at least $\Omega(\mathcal{S}/\epsilon^2)$ samples are required to learn such a state with error $\epsilon$, which is a contradiction to Eq.~\eqref{aeq:contradiction}.    
\end{proof}
\noindent  In some ways, the end-to-end complexity lower bound given in Theorem~\ref{thm:1} can not be significantly improved unless a new iterative QLSS model is proposed. Alternatively, we argue that one central task is to find practical problems with low iteration-tolerant $\mathcal{S}$ with corresponding end-to-end algorithm $(\mathcal{E}, U_\mathrm{QLSS}, \mathcal{D})$.

\section{End-to-End Spectral-Based Quantum Navier-Stokes Solver}

In this section, we detail the end-to-end spectral-based quantum Navier-Stokes solver.
We first review the basics of CFD and the block-encoding arithmetic utilized in this work.
Following that, we detail our technical results about hierarchy spectral block-encoding.
Equipped with this method, we show how to construct the desired block-encodings of the flux Jacobian matrix $A$ and the residual vector $b$ via the hierarchy spectral block-encoding methods.
Then we introduce our sparse spectral state tomography.
Finally, we integrate all of the technical blocks with a rigorous logical resource analysis.

\subsection{Computational Fluid Dynamics Preliminaries}
\subparagraph{Compressible Navier-Stokes Equations} 
The compressible Navier-Stokes (N-S) equations are a set of partial differential equations, consist of the continuity, momentum, and energy equations.

The continuity equation is given by
\begin{equation}\label{eq:continuity}
    \frac{\partial\rho}{\partial t}+\nabla\cdot(\rho\mathbf{v})=0,
\end{equation}
the momentum equation is
\begin{equation}\label{eq:momentum}
    \frac{\partial (\rho\mathbf{v})}{\partial t}+\nabla\cdot(\rho\mathbf{v}\otimes\mathbf{v})=\nabla\cdot\mathbf{\boldsymbol{\sigma}}+\rho\mathbf{f},
\end{equation}
and the energy equation is
\begin{equation}\label{eq:energy}
    \frac{\partial (\rho E)}{\partial t}+\nabla\cdot(\rho E\mathbf{v})=\rho\mathbf{f}\cdot\mathbf{v}+\nabla\cdot(\mathbf{\boldsymbol{\sigma}}\cdot\mathbf{v})-\nabla\cdot\mathbf{q}.
\end{equation}
In these equations, $\rho$ is the fluid density, $\mathbf{v}$ is the velocity vector, and $E$ is the specific total energy, defined as the sum of specific internal energy $e$ and kinetic energy:
\begin{equation}\label{eq:total-energy}
    E = e+ \frac{\|\mathbf{v}\|^2}{2}.
\end{equation}
Furthermore, $\mathbf{f}$ represents the specific body force, and $\mathbf{q}$ is the heat flux density, which follows Fourier's law:
\begin{equation}
    \mathbf{q}=-k\nabla T,
\end{equation}
where $k$ is the thermal conductivity and $T$ is the temperature. 
Under the calorically perfect gas assumption, the pressure $p$ and specific internal energy $e$ are given by
\begin{equation}\label{eq:thermodynamic}
p = \rho R T, \qquad
e = c_v T = \frac{R}{\gamma -1}T,
\end{equation}
where $R$ is the specific gas constant, and $\gamma$ is the heat capacity ratio, both of which are assumed to be constant.

Under the Stokes hypothesis, the stress tensor $\mathbf{\boldsymbol{\sigma}}$ is decomposed into an isotropic pressure component and a deviatoric shear stress component $\boldsymbol{\tau}$, as 
\begin{equation}
    \boldsymbol{\sigma} = -p \boldsymbol{}{I} + \boldsymbol{\tau}.
\end{equation}
Here, $p$ is the thermodynamic pressure from the equation of state, and $\mathbf{I}$ is the identity tensor. The shear stress tensor $\boldsymbol{\tau}$ is given by
\begin{equation}
\boldsymbol{\tau} = \mu \left( \nabla \mathbf{v} + (\nabla \mathbf{v})^\mathsf{T} - \frac{2}{3} (\nabla \cdot \mathbf{v}) \boldsymbol{I} \right).
\end{equation}
Here, $\mu$ is the dynamic viscosity. Its dependence on temperature is commonly described by Sutherland’s law:
\begin{equation}\label{eq:sutherland}
\mu = \mu_0 \left(\frac{T}{T_0}\right)^{3/2}\frac{T_0+T_S}{T+T_S}.
\end{equation}

To nondimensionalize these equations, we introduce the following characteristic scales: length $L$, velocity $U$, time $L/U$, density $\rho_\infty$, temperature $T_\infty$, pressure $\rho_\infty U^2$, specific energy $U^2$, speed of sound $c_\infty = \sqrt{\gamma R T_\infty}$, and dynamic viscosity $\mu_\infty$. The dimensionless variables, denoted with a superscript asterisk, are introduced as:
\begin{equation}\label{eq:nondimensionalize}
    \begin{split}
        &\mathbf{x}^* = \frac{\mathbf{x}}{L}, \quad
        \mathbf{v}^* = \frac{\mathbf{v}}{U}, \quad
        t^* = \frac{t}{L/U}, \quad
        \rho^* = \frac{\rho}{\rho_\infty}, \\
        &p^* = \frac{p}{\rho_\infty U^2}, \quad
        T^* = \frac{T}{T_\infty}, \quad
        E^* = \frac{E}{U^2}, \quad
        c^* = \frac{c}{c_\infty}, \quad
        \mu^* = \frac{\mu}{\mu_\infty}
    \end{split}
\end{equation}
Substituting Eqs.~\eqref{eq:nondimensionalize} into the governing equations (Eqs.~\eqref{eq:continuity}--\eqref{eq:energy}) and then omitting the asterisks for clarity yields the dimensionless form of the compressible NSE. In the absence of body forces $\mathbf{f}$, the equations can be written in the conservative form:
\begin{equation}
\frac{\partial W}{\partial t}=\nabla \cdot \left(\mathcal{F}_{\mu}-\mathcal{F}_C\right),
\end{equation}
where the conservative variables $W$, the convective flux $\mathcal{F}_C$, and the viscous flux $\mathcal{F}_\mu$ are defined as:
\begin{equation}\label{eq:conservative-form}
W=
\begin{bmatrix}
\rho \\
\rho \mathbf{v} \\
\rho E
\end{bmatrix},\quad
\mathcal{F}_C=
\begin{bmatrix}
\rho \mathbf{v} \\
\rho \mathbf{v} \otimes \mathbf{v} + p \boldsymbol{I}\\
(\rho E + p)\mathbf{v}
\end{bmatrix},\quad
\mathcal{F}_\mu=\frac{\mu}{\mathrm{Re}}
\begin{bmatrix}
0 \\
\boldsymbol{\tau} \\
 \boldsymbol{\tau} \cdot \boldsymbol{u}
+\frac{1}{\beta} \nabla T
\end{bmatrix}, 
\end{equation}
where $\beta = \mathrm{Pr}\cdot\mathrm{Ma}^2\cdot(\gamma-1)$.
The dimensionless parameters that arise are the Reynolds number 
\begin{equation}
    \mathrm{Re} = \frac{\rho_\infty U L}{\mu_\infty},
\end{equation}
the Mach number 
\begin{equation}
    \mathrm{Ma} = \frac{U}{c_\infty},
\end{equation}
and the Prandtl number 
\begin{equation}
    \mathrm{Pr}=\frac{\mu_\infty c_p}{\kappa} = \frac{\gamma R \mu_\infty}{(\gamma - 1) \kappa},
\end{equation}
where the thermal conductivity $\kappa$ is assumed to be constant.
Consequently, the closure relations, which correspond to Eqs.~\eqref{eq:total-energy}--\eqref{eq:sutherland}, are expressed in their dimensionless form. Omitting the asterisks for clarity, these relations are:
\begin{align}
    p &= (\gamma-1)\rho(E-\frac{1}{2}\|\mathbf{v}\|^2),\\
    T &= \gamma(\gamma-1) \mathrm{Ma}^2 (E-\frac{1}{2}\|\mathbf{v}\|^2),\\
    \mu &= T^{3/2} \frac{1+T_S/T_0}{T+T_S/T_0}. \label{eq:sutherland-dimless}
\end{align}
In Eq.~\eqref{eq:sutherland-dimless}, the characteristic values $T_\infty$ and $\mu_\infty$ have been set to the reference temperature $T_0$ and reference viscosity $\mu_0$ from Eq.~\eqref{eq:sutherland}, respectively.
\subparagraph{Spatial Discretization via Finite Volume Method}
For simplicity, we consider the 2-dimensional structured grid spatial discretization of
\begin{equation}
    N = N_x\times N_y
\end{equation}
cells that can be indexed by two coordinates as
\begin{equation}\label{eq:double_index}
    \begin{bmatrix}
        (1, 1) & (1, 2) & \cdots & (1, N_x)\\
        (2, 1) & (2, 2) & \cdots & (2, N_x)\\
        \vdots & \vdots & \ddots & \vdots\\
        (N_y, 1) & (N_y, 2) & \cdots & (N_x, N_y)
    \end{bmatrix}
\end{equation}
or be indexed by a single number as
\begin{equation}\label{eq:single_index}
    \begin{bmatrix}
        1 & 2 & \cdots & N_x\\
        N_x+1 & N_x+2 & \cdots & 2N_x\\
        \vdots & \vdots & \ddots & \vdots\\
        (N_y-1)N_x+1 & (N_y-1)N_x+2 & \cdots & N_xN_y
    \end{bmatrix}.
\end{equation}

By substituting the velocity vector $\mathbf{v}=(u,v)$ into Eqs.~\eqref{eq:conservative-form}, the governing equations are expressed as
\begin{equation}
    \frac{\partial W}{\partial t} 
    = \frac{\partial \mathcal{F}_\mu}{\partial x} + \frac{\partial \mathcal{G}_\mu}{\partial y} - \frac{\partial \mathcal{F}_C}{\partial x} - \frac{\partial \mathcal{G}_C}{\partial y}.
\end{equation}
The state vector $W$, the convective fluxes $\mathcal{F}_C$ and $\mathcal{G}_C$ are expanded as
\begin{equation}
    W = 
    \begin{bmatrix}
        \rho \\
        \rho u \\
        \rho v \\
        \rho E
    \end{bmatrix},\quad
    \mathcal{F}_C = 
    \begin{bmatrix}
        \rho u \\
        \rho (u^2+p) \\
        \rho uv \\
        (\rho E + p)u
    \end{bmatrix}, \quad
    \mathcal{G}_C = 
    \begin{bmatrix}
        \rho v \\
        \rho uv \\
        \rho (v^2+p) \\
        (\rho E + p)v
    \end{bmatrix},
\end{equation}
the viscous fluxes $\mathcal{F}_\mu$ and $\mathcal{G}_\mu$ are expanded as
\begin{equation}
    \mathcal{F}_\mu = \frac{\mu}{\mathrm{Re}}
    \begin{bmatrix}
        0 \\
        \frac{2}{3}(2\frac{\partial u}{\partial x} -\frac{\partial v}{\partial y})\\
        \frac{\partial u}{\partial y} + \frac{\partial v}{\partial x} \\
        \frac{2}{3}u(2\frac{\partial u}{\partial x} -\frac{\partial v}{\partial y}) + v(\frac{\partial u}{\partial y} + \frac{\partial v}{\partial x})+\frac{1}{\beta}\frac{\partial T}{\partial x}
    \end{bmatrix}, \quad
    \mathcal{G}_C = \frac{\mu}{\mathrm{Re}}
    \begin{bmatrix}
        0 \\
        \frac{\partial u}{\partial y} + \frac{\partial v}{\partial x} \\
        \frac{2}{3}(2\frac{\partial u}{\partial x} -\frac{\partial v}{\partial y})\\
        u(\frac{\partial u}{\partial y} + \frac{2}{3}v(2\frac{\partial u}{\partial x} -\frac{\partial v}{\partial y}) + \frac{\partial v}{\partial x})+\frac{1}{\beta}\frac{\partial T}{\partial y}
    \end{bmatrix}.
\end{equation}

For the flux evaluation, we use piecewise linear approximations of the solution in each cell, which results in a central difference scheme that is second-order accurate in space. The resulting approximations for the flux derivatives in the $x$ and $y$ directions are given as
\begin{align}
    \left(\frac{\partial\mathcal{F}}{\partial x}\right)_I &\approx \frac{\mathcal{F}_{I+\frac{1}{2}} - \mathcal{F}_{I-\frac{1}{2}}}{\Delta x} = \frac{\frac{\mathcal{F}_{I+1}+\mathcal{F}_{I}}{2}-\frac{\mathcal{F}_{I}+\mathcal{F}_{I-1}}{2}}{\Delta x} = \frac{\mathcal{F}_{I+1}-\mathcal{F}_{I-1}}{2 \Delta x}, \\
    \left(\frac{\partial\mathcal{G}}{\partial y}\right)_J &\approx \frac{\mathcal{G}_{J+\frac{1}{2}} - \mathcal{G}_{J-\frac{1}{2}}}{\Delta y} = \frac{\frac{\mathcal{G}_{J+1}+\mathcal{G}_{J}}{2}-\frac{\mathcal{G}_{J}+\mathcal{G}_{J+1}}{2}}{\Delta y} = \frac{\mathcal{G}_{J+1}-\mathcal{G}_{J-1}}{2 \Delta y},
\end{align}
where the total flux in $x$ direction $\mathcal{F}=\mathcal{F}_\mu-\mathcal{F}_C$ and the total flux in $y$ direction $\mathcal{G}=\mathcal{G}_\mu - \mathcal{G}_C$. By substituting these approximations into the governing equation, we obtain the semi-discretized equation for cell $(I,J)$:
\begin{equation}\label{eq:residual}
    \begin{split}
        \frac{\partial W_{I, J}}{\partial t} 
    &= \frac{\mathcal{F}_{I+1, J}-\mathcal{F}_{I-1, J}}{2 \Delta x}+\frac{\mathcal{G}_{I, J+1}-\mathcal{G}_{I, J-1}}{2 \Delta y}\\
    &= \left(\frac{\mathcal{F}_{x+1}-\mathcal{F}_{x-1}}{2 \Delta x}+\frac{\mathcal{G}_{y+1}-\mathcal{G}_{y-1}}{2 \Delta y}\right)(W_{I, J})\\
    &= R(W_{I, J})
    \end{split}
\end{equation}

The right-hand side of the equation is the residual vector $R(W_{I,J})$, which is a vector function of the conservative variables $W_{I,J}$.

\subparagraph{Temporal Discretization via Implicit Euler Method}
The semi-discretized equation is advanced in time using the first-order implicit Euler method, given by:
\begin{equation}
\frac{W^{n+1}-W^{n}}{\Delta t} = R(W^{n+1})
\end{equation}
where $W^n$ is the solution at time step $n$, $\Delta t$ is the time step size, and $R(W)$ is the residual (or spatial operator).

To solve this nonlinear equation for $W^{n+1}$, we linearize the residual term $R(W^{n+1})$ by applying a Taylor series expansion around the known state $W^n$:
\begin{equation}
R(W^{n+1}) \approx R(W^{n})+\frac{\partial R(W^n)}{\partial W} (W^{n+1}-W^{n})
\end{equation}
where higher-order terms are neglected.

Substituting the linearized residual back into the implicit Euler equation and rearranging yields a linear system for the state update $\Delta W = W^{n+1}-W^{n}$:
\begin{equation}
\left[\frac{I}{\Delta t}-\frac{\partial R(W^n)}{\partial W} \right]\Delta W = R(W^n)
\end{equation}
This system must be solved at each time step to find the solution $W^{n+1}$.
The flux Jacobian matrix \begin{equation}
    J=\frac{\partial R}{\partial W}
\end{equation} measures the sensitivity of the residual with respect to the conservative variables. The convective Jacobian block relating cell $(I',J')$ to the residual at cell $(I,J)$ is derived from \eqref{eq:residual} as follows:  
\begin{equation}
    \begin{split}
        &\left.(J_C)_{(I,J),(I',J')} \triangleq \frac{\partial (R_C)_{I,J}}{\partial W_{I',J'}}\right. \\
        =& \left.\frac{\partial }{\partial W_{I', J'}}\middle( \frac{\mathcal{F}_{I, J} + \mathcal{F}_{I+1, J}}{2\Delta x} - \frac{\mathcal{F}_{I-1, J} + \mathcal{F}_{I, J}}{2\Delta x} + \frac{\mathcal{G}_{I, J} + \mathcal{G}_{I, J+1}}{2\Delta y} - \frac{\mathcal{G}_{I, J-1} + \mathcal{G}_{I, J}}{2\Delta y} \right) \\
        =& \frac{1}{2}\left(\frac{1}{\Delta x} \frac{\partial \mathcal{F}_{I+1, J}}{\partial W_{I', J'}} -\frac{1}{\Delta x} \frac{\partial \mathcal{F}_{I-1, J}}{\partial W_{I', J'}} + \frac{1}{\Delta y}\frac{\partial \mathcal{F}_{I, J+1}}{\partial W_{I', J'}} - \frac{1}{\Delta y} \frac{\partial \mathcal{F}_{I, J-1}}{\partial W_{I', J'}}\right)
    \end{split}
\end{equation}
For the 2-dimension case where each cell has at most $4$ neighboring cells (left, right, up, and down), this can be further simplified to be 
\begin{equation}
    \begin{split}
        & \frac{\delta_{I+1, I'}\delta_{J, J'}}{2\Delta x}\frac{\partial \mathcal{F}}{\partial W}\left(W_{I+1, J}\right) - \frac{\delta_{I-1, I'}\delta_{J, J'}}{2\Delta x}\frac{\partial \mathcal{F}}{\partial W}\left(W_{I-1, J}\right)\\
        &+\frac{\delta_{I, I'}\delta_{J+1, J'}}{2\Delta y}\frac{\partial \mathcal{G}}{\partial W}\left(W_{I, J+1}\right) - \frac{\delta_{I, I'}\delta_{J-1, J'}}{2\Delta y}\frac{\partial \mathcal{G}}{\partial W}\left(W_{I, J-1}\right).
    \end{split}
\end{equation}
Therein the Kronecker delta function $\delta_{I, J}$ induces a four-ban structured convective Jacobian matrix of size $4N\times4N$.
\begin{widetext}
    \begin{equation}\label{eq:flux_jacobian_matrix}
        \begin{split}
            J_{C} =
         \begin{pmatrix}
            0 & \left.\frac{\partial \mathcal{F}}{\partial W}\right|_{1,2} & 0& \cdots & \left.\frac{\partial \mathcal{G}}{\partial W}\right|_{2, 1} & 0 & \cdots & 0 \\
            -\left.\frac{\partial \mathcal{F}}{\partial W}\right|_{1,1} & 0 & \left.\frac{\partial \mathcal{F}}{\partial W}\right|_{1,3} & 0 & \cdots & \left.\frac{\partial \mathcal{G}}{\partial W}\right|_{2, 2} & \ddots&\vdots\\
            0 & -\left.\frac{\partial \mathcal{F}}{\partial W}\right|_{1,2} & \ddots & \ddots & \ddots & \ddots & \ddots&0\\
            \vdots & 0 & \ddots &\ddots&\ddots&\ddots&\ddots&\left.\frac{\partial \mathcal{G}}{\partial W}\right|_{N_y,N_x}\\
            -\left.\frac{\partial \mathcal{G}}{\partial W}\right|_{1, 1} & \vdots & \ddots & \ddots & \ddots & \ddots &\ddots & 0\\
            0& -\left.\frac{\partial \mathcal{G}}{\partial W}\right|_{1, 2}&\ddots&\ddots&\ddots&\ddots&\ddots&\vdots\\
            \vdots & \ddots & \ddots & \ddots & \ddots & \ddots & \ddots & \left.\frac{\partial \mathcal{F}}{\partial W}\right|_{N_y,N_x}\\
            0 & \cdots & 0 & -\left.\frac{1}{\Delta y}\frac{\partial \mathcal{G}}{\partial W}\right|_{N_y-1, N_x} & 0 &\cdots&-\left.\frac{\partial \mathcal{F}}{\partial W}\right|_{N_y,N_x-1} & 0
        \end{pmatrix}
        \end{split}
    \end{equation}
\end{widetext}

Each block is a $4\times4$ submatrix, where the spatial scaling factors $1/(2\Delta x)$ for $\mathcal{F}$ and $1/(2\Delta y)$ for $\mathcal{G}$ are omitted for clarity.  
While we showcase the 2-dimension result for simplicity, a generalization to the $D$-dimension case can be easily checked with $2D$-ban $(3+D)N\times(3+D)N$-shape convective flux Jacobian.

Go one step further from block-wise to element-wise level, and note that $\frac{\partial \mathcal{F}}{\partial W}$ and $\frac{\partial \mathcal{G}}{\partial W}$ are the horizontal and vertical flux Jacobian matrices concerning the conservative variables $W$ computed as
\begin{equation}\label{appendix_eq:horizontal_jacobian}
    \frac{\partial \mathcal{F}_C}{\partial W}=
    \begin{pmatrix}
        0 & 1 & 0 & 0\\
        \frac{\gamma-1}{2}(u^2+v^2)-u^2&(3-\gamma)u & (1-\gamma)v&\gamma-1\\
        -uv & v & u & 0\\
        \left(\frac{\gamma-1}{2}(u^2+v^2)-H\right)u & H+(1-\gamma)u^2&(1-\gamma)uv&\gamma u
    \end{pmatrix}
\end{equation}
and
\begin{equation}\label{appendix_eq:vertical_jacobian}
    \frac{\partial \mathcal{G}_C}{\partial W}=
    \begin{pmatrix}
        0 & 0 & 1 & 0\\
        -uv & v & u & 0\\
        \frac{\gamma-1}{2}(u^2+v^2)-v^2&(1-\gamma)u & (3-\gamma)v&\gamma-1\\
        \left(\frac{\gamma-1}{2}(u^2+v^2)-H\right)v &(1-\gamma)uv& H+(1-\gamma)v^2&\gamma v
    \end{pmatrix},
\end{equation}
respectively, where $H = \gamma E - \frac{\gamma-1}{2} (u^2 + v^2)$
Here the subscripts denoting cell index have been omitted for simplicity.
The viscous Jacobian matrix is approximated in a simplified manner by filling each block on the block-diagonal with the spectral radius of the Jacobian, thereby enhancing diagonal dominance. Thus, we obtain
\begin{equation}\label{aeq:A}
    \begin{split}
        A &= \frac{I}{\Delta t} - J_\mu + J_C \\
        &= 
        \left(
         \begin{array}{cccccc}
            I+\left.\frac{\mu}{\rho}\right|_{1,1} & \left.\frac{\partial \mathcal{F}_C}{\partial W}\right|_{1,2} & & \left.\frac{\partial \mathcal{G}_C}{\partial W}\right|_{2, 1} & & -\left.\frac{\partial \mathcal{F}_C}{\partial W}\right|_{N_y,N_x} \\
            -\left.\frac{\partial \mathcal{F}_C}{\partial W}\right|_{1,1} & I+\left.\frac{\mu}{\rho}\right|_{1,2} & \left.\frac{\partial \mathcal{F}_C}{\partial W}\right|_{1,3} & & \ddots & \\
             & -\left.\frac{\partial \mathcal{F}_C}{\partial W}\right|_{1,2} & \ddots & \ddots & & \left.\frac{\partial \mathcal{G}_C}{\partial W}\right|_{N_y,N_x} \\
            -\left.\frac{\partial \mathcal{G}_C}{\partial W}\right|_{1, 1} &  & \ddots & \ddots & \ddots & \\
            & \ddots & & \ddots & \ddots & \left.\frac{\partial \mathcal{F}_C}{\partial W}\right|_{N_y,N_x} \\
            \left.\frac{\partial \mathcal{F}_C}{\partial W}\right|_{1,1} & & -\left.\frac{\partial \mathcal{G}_C}{\partial W}\right|_{N_y-1, N_x} & & -\left.\frac{\partial \mathcal{F}_C}{\partial W}\right|_{N_y,N_x-1} & I+\left.\frac{\mu}{\rho}\right|_{N_y, N_x}
        \end{array}
        \right),
    \end{split}
\end{equation}

\subsection{Block-encoding Preliminaries}

In this work, we utilize the following two lemmas to implement the arithmetic of block-encoding matrices. The first one is the linear combination of block-encoded matrices. Since the original lemma in Ref.~\cite{gilyen2019quantum} can not be directly applied to our case, we state the following modified version for self-consistency:
        \begin{alemma}[Linear combination of block-encoding matrices, modified from Lemma 52 of \cite{gilyen2019quantum}]\label{lem:lcu}
            Let $A=\sum_{j=0}^{m-1}y_jA_j$ be an $n$-qubits operator and $\epsilon>0$. Suppose that $A_j$ can be $(\alpha_j, a, \epsilon_j)$-block-encoded and $\Tilde{\beta} = \sum_{j=0}^{m-1}\lvert y_j\alpha_j\rvert$. Further suppose the pair of unitaries $(P_L, P_R)$ such that $P_L\ket{0}^{\otimes b} = \sum_{j=0}^{2^b-1}c_j\ket{j}$, $P_R\ket{0}^{\otimes b} = \sum_{j=0}^{2^b-1}d_j\ket{j}$, $\sum_{j=0}^{m-1}\left\lvert y_j\alpha_j - \Tilde{\beta} (c_j^*d_j) \right\rvert\leq\delta$, and $c_j^*d_j=0$ for $j\geq m$.
            Then one can implement a $\left( \Tilde{\beta}, a+b, \delta + \sum_{j=0}^{m-1}\lvert y_j \rvert\epsilon_j \right)$-block-encoding of $A$.
        \end{alemma}
        \noindent Herein, we allow the normalization coefficients to be all different to enable more flexible utilization at a cost of a more complicated error bound. The basic idea is to absorb the normalization constant as a part of the linear combination coefficients:
        \begin{proof}
            Consider the unitary
            \begin{equation}
                \Tilde{W} = \left(P_L^\dagger\otimes I_{a+n}\right)
                    \left( \sum_{j=0}^{m-1}\ket{j}\bra{j}\otimes U_j + \left(\sum_{j=m}^{2^b-1}\ket{j}\bra{j}\right) \otimes I_{a+n}\right)
                    \left(P_R\otimes I_{a+n}\right).
            \end{equation}
            Then a direct computation shows that
            \begin{equation}
                \begin{split}
                    &\left\lVert A-\Tilde{\beta}
                    \left(\bra{0}^{\otimes(a+b)}\otimes I_n\right)
                    \Tilde{W}
                    \left(\ket{0}^{\otimes(a+b)}\otimes I_n\right) \right\rVert\\
                    =&\left\lVert A-\sum_{j=0}^{m-1}\Tilde{\beta}(c_j^*d_j)\left(\bra{0}^{\otimes a}\otimes I_n\right)U_j
                    \left(\ket{0}^{\otimes a}\otimes I_n\right)\right\rVert\\
                    \leq& \delta + \left\lVert A-\sum_{j=0}^{m-1}y_j\alpha_j\left(\bra{0}^{\otimes a}\otimes I_n\right)U_j
                    \left(\ket{0}^{\otimes a}\otimes I_n\right) \right\rVert\\
                    \leq& \delta + \sum_{j=0}^{m-1}\lvert y_j \rvert\left\lVert A_j-\alpha_j\left(\bra{0}^{\otimes a}\otimes I_n\right)U_j
                    \left(\ket{0}^{\otimes a}\otimes I_n\right) \right\rVert\\
                    =& \delta + \sum_{j=0}^{m-1}\lvert y_j \rvert\epsilon_j.\\
                \end{split}
            \end{equation}
        \end{proof}
\noindent The second Lemma is for the product of two block-encoding matrices, which will be heavily utilized in our block-encoding methods:
        \begin{alemma}[Product of block-encoding matrices, Lemma 53 of \cite{gilyen2019quantum}]\label{lem:product}
            If U is an $(\alpha, a, \delta)$-block-encoding of an $s$-qubit operator $A$, and $V$ is an $(\beta, b, \epsilon)$-block-encoding of an $s$-qubit operator $B$, then $(I_b\otimes U)(I_a\otimes V)$ is an $(\alpha\beta, a+b, \alpha\epsilon+\beta\delta)$-block-encoding of $AB$.
        \end{alemma}
\noindent For the product of two unitaries, the ancillary qubits can be saved as:
        \begin{alemma}[Product of block-encoding unitaries, Lemma 54 of \cite{gilyen2019quantum}]\label{lem:unitary_product}
            If U is an $(1, a, \delta)$-block-encoding of an $s$-qubit unitary operator $A$, and $V$ is an $(1, a, \epsilon)$-block-encoding of an $s$-qubit unitary operator $B$, then $UV$ is an $(1, a, \epsilon+\delta+2\sqrt{\delta\epsilon})$-block-encoding of $AB$.
        \end{alemma}
\noindent We also utilize the following celebrated result to implement the polynomial transformation of the block-encoding's singular values:
        \begin{atheorem}[Quantum singular value transformation, reformulated from Corollary 18 of \cite{gilyen2019quantum}]\label{thm:qsvt}
            Given unitary $U$, orthogonal projectors $\Pi, \Tilde{\Pi}$, and a degree-$n$ polynomial $P_\mathcal{R}\in \mathbb{R}[x]$ satisfying that $P_\mathcal{R}$ has parity-($n$ mod 2) and  $\lvert P_\mathcal{R}(x) \rvert\leq 1$ for all $-1\leq x\leq1$.
            Then there exists $\Phi\in\mathbb{R}^n$ such that
            \begin{equation}
                P_\mathcal{R}^{SV}(x) = 
                \begin{cases}
                    \left(\bra{+}\otimes\Tilde{\Pi}\right) 
                    \left(\ket{0}\bra{0}\otimes U_\Phi + \ket{1}\bra{1}\otimes U_{-\Phi}\right)
                    \left(\ket{+}\otimes{\Pi}\right) & \text{if }  n  \text{ is odd, and}\\
                    \left(\bra{+}\otimes{\Pi}\right) 
                    \left(\ket{0}\bra{0}\otimes U_\Phi + \ket{1}\bra{1}\otimes U_{-\Phi}\right)
                    \left(\ket{+}\otimes{\Pi}\right) & \text{if }  n  \text{ is even,}
                \end{cases}
            \end{equation}
            where
            \begin{equation}
                U_\Phi = 
                \begin{cases}                    
                    e^{i\phi_{1}(2\Tilde{\Pi}-I)}U\prod_{j=1}^{n/2} e^{i\phi_{2j}(2{\Pi}-I)}U^\dagger e^{i\phi_{2j+1}(2\Tilde{\Pi}-I)}U & \text{if }  n  \text{ is odd, and}\\
                    \prod_{j=1}^{n/2} e^{i\phi_{2j-1}(2{\Pi}-I)}U^\dagger e^{i\phi_{2j}(2\Tilde{\Pi}-I)}U & \text{if }  n  \text{ is even,}
                \end{cases}
            \end{equation}
        \end{atheorem}
        \noindent In this work, we also use the block-encoding of the linear function $[0, \frac{1}{2^n-1}, \frac{2}{2^n-1}, \frac{3}{2^n-1}, ..., 1]$:
        \begin{alemma}[Block-encoding of linear function, Lemma 8 of \cite{zhuang2024statistics}]\label{lem:lft}
            There is a $(1, \log{n}, 0)$ block-encoding of the $2^n$-dimension diagonal matrix $A = diag(0, \frac{1}{2^n-1}, \frac{2}{2^n-1}, \frac{3}{2^n-1}, ..., 1)$ with $\mathcal{O}(n)$ circuit depth.
        \end{alemma}
    
\subsection{Hierarchy spectral block-encoding}\label{sec:hsbe}
    In this subsection, we develop the following two Lemmas to block-encode the diagonal matrix of an arbitrary variable with sparse spectral structure and to assemble these variables as a hierarchy-structured matrix, respectively. 
    
    To assemble these variables into a structured matrix, conventional block-encoding of a square matrix can be decomposed into a state-preparation unitary to create the superposition of rows and a controlled-state preparation unitary to create the superposition of each element according to the row index. For simplicity, we define these two unitaries as a block-encoding pair.
        \begin{adefinition}[Block-encoding pair]
            Given a $2^m\times2^m$ matrix $A$. Then a pair of $(2m+1)$-qubit unitaries $(U_L, U_R^\dagger)$ is called a block-encoding pair of $A$ if they satisfies
        \begin{equation}
            U_L\ket{0}^{\otimes m}\ket{0}\ket{I}^{\otimes m} = \ket{I}^{\otimes m}\ket{0}
            \left(
                \sum_{I'}\frac{\lvert A_{I'}\rvert}{\lvert A\rvert}\ket{I'}^{\otimes m}
            \right),
        \end{equation}
        and
        \begin{equation}
            U_R^\dagger\ket{I}^{\otimes m}\ket{0}\ket{I'}^{\otimes m} = \frac{A_{I, I'}}{\lvert A_{I'}\rvert} \ket{0}^{\otimes m}\ket{0}\ket{I'}^{\otimes m},
        \end{equation}
        where $\lvert A_I\rvert=\sqrt{\sum_{I'=1}^{2^m} A_{I,I'}^2}$ and $\lvert A\rvert=\sqrt{\sum_{I, I'=1}^{2^m} A_{I,I'}^2}$.
        \end{adefinition}
        Our basic idea is to use a divide-and-conquer method to divide these large matrix into small submatrices. More formally, we have:
        \begin{alemma}[Hierarchy block-encoding]\label{lem:hierarchy}
            A $2^{m+n}\times2^{m+n}$ matrix can be block-encoded within $2^m\mathcal{C}$ depth if its $2^{n}\times2^{n}$-shaped submatrices
            \begin{equation}
                A_{I,I'}=\begin{pmatrix}
                A_{I\cdot2^n, I'\cdot2^n} & A_{I\cdot2^n, I'\cdot2^n+1} & \cdots & A_{I\cdot2^n, I'\cdot2^n+2^n-1}\\
                A_{I\cdot2^n+1, I'\cdot2^n} & A_{I\cdot2^n+1, I'\cdot2^n+1} & \cdots & A_{I\cdot2^n+1, I'\cdot2^n+2^n-1}\\
                \vdots & \vdots & \ddots & \vdots \\
                A_{I\cdot2^n+2^n-1, I'\cdot2^n} & A_{I\cdot2^n+2^n-1, I'\cdot2^n+1} & \cdots & A_{I\cdot2^n+2^n-1, I'\cdot2^n+2^n-1}\\
            \end{pmatrix} \text{ for } 0\leq I,I'\leq2^m-1,
            \end{equation}
            can be block-encoded within $\mathcal{C}$ depth.
        \end{alemma}
        \begin{proof}        
        Our key observation is that, for a $2^{m+n}\times2^{m+n}$ matrix of two levels, we can first prepare a block-encoding pair for the high-level $2^{m}\times2^{m}$ matrix
        \begin{equation}
            \Tilde{A} = 
            \begin{pmatrix}
                \sqrt{\sum_{j,j'=0}^{2^n-1}A_{0\cdot2^n+j, 0\cdot2^n+j'}^2} & \sqrt{\sum_{j,j'=0}^{2^n-1}A_{0\cdot2^n+j, 1\cdot2^n+j'}^2} & \cdots & \sqrt{\sum_{j,j'=0}^{2^n-1}A_{0\cdot2^n+j, (2^m-1)\cdot2^n+j'}^2}\\
                \sqrt{\sum_{j,j'=0}^{2^n-1}A_{1\cdot2^n+j, 0\cdot2^n+j'}^2} & \sqrt{\sum_{j,j'=0}^{2^n-1}A_{1\cdot2^n+j, 1\cdot2^n+j'}^2} & \cdots & \sqrt{\sum_{j,j'=0}^{2^n-1}A_{1\cdot2^n+j, (2^m-1)\cdot2^n+j'}^2}\\
                \vdots & \vdots & \ddots & \vdots \\
                \sqrt{\sum_{j,j'=0}^{2^n-1}A_{(2^m-1)\cdot2^n+j, 0\cdot2^n+j'}^2} & \sqrt{\sum_{j,j'=0}^{2^n-1}A_{(2^m-1)\cdot2^n+j, 1\cdot2^n+j'}^2} & \cdots & \sqrt{\sum_{j,j'=0}^{2^n-1}A_{(2^m-1)\cdot2^n+j, (2^m-1)\cdot2^n+j'}^2}\\
            \end{pmatrix},
        \end{equation}
        and then insert a sequence of controlled block-encoding subroutines between them to block-encode each low-level $2^{n}\times2^{n}$ sub-matrix.
        In the following, we use a two-level index $(I,I';J,J')$ to record the elements at $(I\cdot2^n+J, I'\cdot2^n+J')$ for simplicity.      
        By assumption, we have the following block-encoding pairs $U_L$ and $U_R^\dagger$ satisfying
        \begin{equation}\label{aeq:u_l}
            U_L\ket{0}^{\otimes m}\ket{0}\ket{I}^{\otimes m} = \ket{I}^{\otimes m}\ket{0}
            \left(
                \sum_{I'}\frac{\lvert A_{I'}\rvert}{\lvert A\rvert}\ket{I'}^{\otimes m}
            \right),
        \end{equation}
        and
        \begin{equation}\label{aeq:u_r}
            U_R^\dagger\ket{I}^{\otimes m}\ket{0}\ket{I'}^{\otimes m} = \frac{\lvert A_{I, I'}\rvert}{\lvert A_{I'}\rvert} \ket{0}^{\otimes m}\ket{0}\ket{I'}^{\otimes m}.
        \end{equation}
        Since Eq.~\eqref{aeq:u_l} and Eq.~\eqref{aeq:u_r} are (controlled) state preparation operators, they can be implemented by (multi-controlled) rotation gates with angles
        \begin{equation}
            \theta_L = 2\left(\arctan\left(\sqrt{\frac{\sum_{I'=2^{m-1}}^{2^{m}-1}\lvert A_{I'}\rvert^2}{\sum_{I'=0}^{2^{m-1}-1}\lvert A_{I'}\rvert^2}}\right), 
            \arctan\left(\sqrt{\frac{\sum_{I'=2^{m-1}-1}^{2^{m-2}}\lvert A_{I'}\rvert^2}{\sum_{I'=0}^{2^{m-2}-1}\lvert A_{I'}\rvert^2}}\right), 
            \arctan\left(\sqrt{\frac{\sum_{I'=3\times2^{m-2}}^{2^{m}-1}\lvert A_{I'}\rvert^2}{\sum_{I'=2^{m-1}}^{3\times2^{m-2}-1}\lvert A_{I'}\rvert^2}}\right), \cdots\right)
        \end{equation}
        and
        \begin{equation}
            \theta_{R(I)} = 2\left(\arctan\left(\sqrt{\frac{\sum_{I=2^{m-1}}^{2^{m}-1}\lvert A_{I, I'}\rvert^2}{\sum_{I=0}^{2^{m-1}-1}\lvert A_{I, I'}\rvert^2}}\right), 
            \arctan\left(\sqrt{\frac{\sum_{I=2^{m-1}-1}^{2^{m-2}}\lvert A_{I, I'}\rvert^2}{\sum_{I=0}^{2^{m-2}-1}\lvert A_{I, I'}\rvert^2}}\right), 
            \arctan\left(\sqrt{\frac{\sum_{I=3\times2^{m-2}}^{2^{m}-1}\lvert A_{I, I'}\rvert^2}{\sum_{I=2^{m-1}}^{3\times2^{m-2}-1}\lvert A_{I, I'}\rvert^2}}\right), \cdots\right)
        \end{equation}
        By assumption, we also have a sequence of indexed block-encoding unitaries $U_{I, I'}$ so that
        \begin{equation}
            \bra{0}^{\otimes a}\bra{J}^{\otimes n} U_{I, I'} \ket{0}^{\otimes a}\ket{J'}^{\otimes n} = \frac{A_{I, I'; J, J'}}{\lvert A_{I, I'}\rvert}.
        \end{equation}
        Consequently, the product of the control version
        \begin{equation}
            C_{I, I'}-U_{I, I'} = \ket{I, I'}^{\otimes 2m}\bra{I, I'}^{\otimes 2m}\otimes U_{I, I'} + (I_{2m} - \ket{I, I'}^{\otimes 2m}\bra{I, I'}^{\otimes 2m})\otimes I_n
        \end{equation}
        acts as
        \begin{equation}
            \bra{I, I'}^{\otimes 2m}\bra{0}^{\otimes a}\bra{J}^{\otimes n}\left(\prod_{I'', I'''}C_{I'', I'''}-U_{I'', I'''}\right)\ket{I, I'}^{\otimes 2m}\ket{0}^{\otimes a}\ket{J'}^{\otimes n}=\frac{A_{I, I'; J, J'}}{\lvert A_{I, I'}\rvert}.
        \end{equation}
        Now the hierarchy block-encoding unitary
        \begin{equation}
            U_A = (U_R^\dagger\otimes I_n)\left(\prod_{I, I'}C_{I, I'}-U_{I, I'}\otimes I_1\right)(U_L\otimes I_n)
        \end{equation}
        acts on the state $\ket{0}\ket{I}\ket{J}$ as
        \begin{equation}
            \begin{split}
                &(U_R^\dagger\otimes I_{a+n})\left(\prod_{I, I'}C_{I, I'}-U_{I, I'}\otimes I_1\right)(U_L\otimes I_{a+n})\ket{0}^{\otimes(m+1)}\ket{I}^{\otimes m}\ket{0}^{\otimes a}\ket{J}^{\otimes n}\\
                =&(U_R^\dagger\otimes I_{a+n})\left(\prod_{I, I'}C_{I, I'}-U_{I, I'}\otimes I_1\right)\ket{I}^{\otimes m}\ket{0}
                \left(
                    \sum_{I'}\frac{\lvert A_{I'}\rvert}{\lvert A\rvert}\ket{I'}^{\otimes m}
                \right)\ket{0}^{\otimes a}\ket{J}^{\otimes n}\\
                =&(U_R^\dagger\otimes I_{a+n})\ket{I}^{\otimes m}\ket{0}
                \left(
                    \sum_{I'}\frac{\lvert A_{I'}\rvert}{\lvert A\rvert}\ket{I'}^{\otimes m} (U_{I, I'}\ket{0}^{\otimes a}\ket{J}^{\otimes n})
                \right)\\
                =&\sum_{I'} \frac{\lvert A_{I, I'}\rvert}{\lvert A_{I'}\rvert}\cdot\frac{\lvert A_{I'}\rvert}{\lvert A\rvert}\ket{0}^{\otimes(m+1)}\ket{I'}^{\otimes m}\left(U_{I, I'}\ket{0}^{\otimes a}\ket{J}^{\otimes n}\right)\\
                =&\sum_{I'} \frac{\lvert A_{I, I'}\rvert}{\lvert A_{I'}\rvert}\cdot\frac{\lvert A_{I'}\rvert}{\lvert A\rvert}\ket{0}^{\otimes(m+1)}\ket{I'}^{\otimes m}\left(\sum_{J'}\frac{A_{I, I'; J, J'}}{\lvert A_{I, I'}\rvert}\ket{0}^{\otimes a}\ket{J'}^{\otimes n}+\ket{\perp}_{a+n}\right)\\
                =&\sum_{I', J'} \frac{A_{I, I'; J, J'}}{\lvert A\rvert}\ket{0}^{\otimes(m+1)}\ket{I'}^{\otimes m}\ket{0}^{\otimes a}\ket{J'}^{\otimes n} + \ket{\perp'}_{a+2m+n+1}\\
            \end{split}            
        \end{equation}
        Consequently, we have that
        \begin{equation}
            \bra{0}^{\otimes(m+1)}\bra{I'}^{\otimes m}\bra{0}^{\otimes a}\bra{J'}^{\otimes n}U_A\ket{0}^{\otimes(m+1)}\ket{I}^{\otimes m}\ket{0}^{\otimes a}\ket{J}^{\otimes n} = \frac{A_{I, I'; J, J'}}{\lvert A\rvert}.
        \end{equation}
        $U_A$ is an $(m+a+1, A, 0)$-block-encoding of A.
        \end{proof}
    
    To block-encode each variable as a diagonal matrix, we have:
    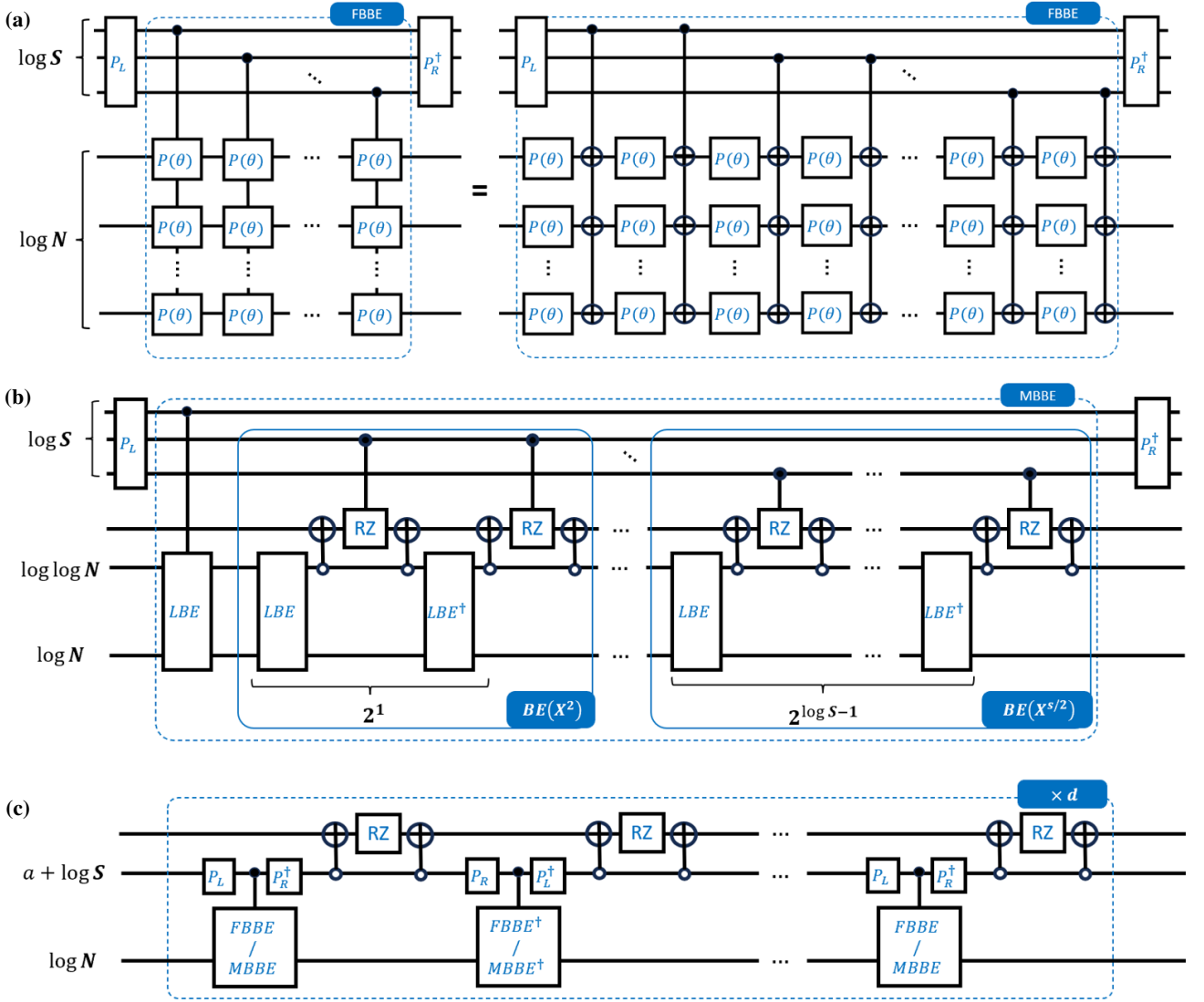
\begin{figure}
        \centering
        \begin{tikzpicture}
            \node{\includegraphics[width=0.95\textwidth]{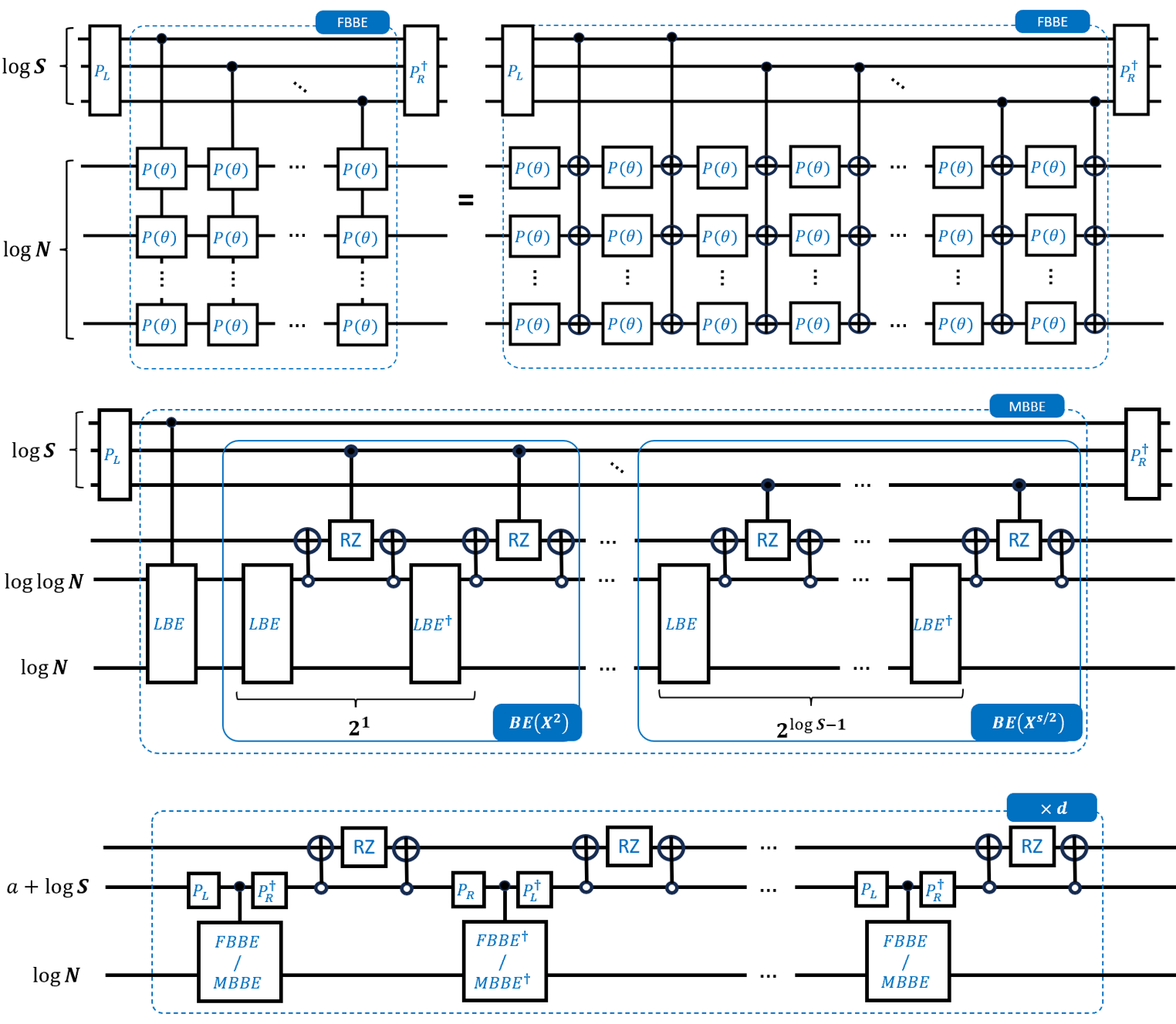}};
            \node[] at (-8.5, 7) {\textbf{(a)}};
            \node[] at (-8.5, 1.5) {\textbf{(b)}};
            \node[] at (-8.5, -4.5) {\textbf{(c)}};
        \end{tikzpicture}
        \caption{Spectral block-encoding circuit.}
        \label{fig:ss-be}
    \end{figure}
        \begin{alemma}[Sparse spectral block-encoding]\label{lem:spectral}
            Suppose $\mathcal{S}$ independent basis defined on $2^n$ grids that can be block-encoded with circuit complexity $\mathcal{O}(\mathcal{C})$, then any function in the spanned space can be block-encoded with $\mathcal{O}(\mathcal{S})$ ancillaries within $\mathcal{O}(\mathcal{S}\mathcal{C})$ depth.
            Specifically: 
            \begin{itemize}
                \item Fourier-based case: $(2\mathcal{S}+2\log\mathcal{S}-2)\cdot\mathrm{ROT}$ T-depth;
                \item Monomial-based case:  $(2(\mathcal{S}-1)(n+1)-2)\cdot\mathrm{ROT} + (2(\mathcal{S}-1)(n+1)(2\log\log n-1)+2n-4\log\log n-2)\cdot\mathrm{TOF}$ T-depth;
                \item Composed case: the composition of a $d$-degree polynomial $P$ with the Fourier-based and monomial-based spectral block-encodings can be implemented within $(2\mathcal{S}+2\log\mathcal{S}-1)d\cdot\mathrm{ROT}+(4\log{\log\mathcal{S}}-2)d\cdot\mathrm{TOF}$ T-depth and $(2(\mathcal{S}-1)(n+1)-1)d\cdot\mathrm{ROT} + (2(\mathcal{S}-1)(n+1)(2\log\log n-1)+2n+4\log\log\mathcal{S} n-4\log\log n-4)d\cdot\mathrm{TOF}$ T-depth, respectively, if $P$ has parity $d\text{ mod }2$ or be doubled if not;
            \end{itemize}
            wherein $ROT$ and $TOF$ are the depths for the rotation and Toffoli gate, respectively. Notably, in the conventional FTQC context, we have $ROT=3\lceil\log\frac{1}{\epsilon}\rceil$ and $TOF = 5$, while in our QEC protocol, this will be implemented within one error correction cycle by magic state injection.
        \end{alemma}
        \begin{proof}
            The basic idea
            Without loss of generality, we consider the uniform $2^n$-grids of the interval $[0, 1]$ as $\{0, \frac{1}{2^n-1}, \frac{2}{2^n-1}, ..., 1 \}$.
            
            \subparagraph{Fourier-based Case.} We first consider the diagonal block-encoding of the Fourier basis $F_k(x)=e^{ikx}$, i.e., 
            \begin{equation}\label{eq:fourier_basis}
                \begin{pmatrix}
                    e^{\frac{i\cdot0}{2^n-1}} &  &  &  \\
                      & e^{\frac{ik}{2^n-1}} & & \\
                      &  & \ddots & \\
                      &  &  &  e^{\frac{ik\cdot (2^n-1)}{2^n-1}}
                \end{pmatrix}
                =                
                \begin{pmatrix}
                    \left( e^{\frac{ik}{2^n-1}} \right)^0 &  &  &  \\
                      & \left( e^{\frac{ik}{2^n-1}} \right)^1 & & \\
                      &  & \ddots & \\
                      &  &  &  \left( e^{\frac{ik}{2^n-1}} \right)^{2^n-1}
                \end{pmatrix}.
            \end{equation}
            Note that Eq.~\eqref{eq:fourier_basis} is indeed a tensor product
            \begin{equation}
                \bigotimes_{j=0}^{n-1}
                \begin{pmatrix}
                    1 & 0\\
                    0 & \left( e^{\frac{ik}{2^n-1}} \right)^{2^j}
                \end{pmatrix},
            \end{equation}
            and hence can be $(1, 0, 0)$-block-encoded by a sequence of parallel Z-rotation gates
            \begin{equation}
                U_{\mathrm{BE}(F_k)} = \bigotimes_{j=0}^{n-1} Z_j(e^{\frac{ik2^j}{2^n-1}}),
            \end{equation}
            wherein $Z_j$ targets on the $j^{th}$ qubit.
            Now we consider the binary representation of $k=\sum_{j=0}^{\log{\mathcal{S}}}c_j\cdot2^j, c_j\in\{0, 1\}$, then the $k^{th}$ Fourier basis function can be $(1, 0, 0)$-block-encoded, by Lemma~\ref{lem:unitary_product}, decomposed as a product
            \begin{equation}
                U_{\mathrm{BE}(F_k)} = \prod_{j=0}^{\log{\mathcal{S}}} U_{\mathrm{BE}(F_{2^j})}^{c_j}.
            \end{equation}
            Consequently, the block-encoding of the Fourier basis $F_k$ controlled by $\ket{k}$ can be decomposed into:
            \begin{equation}
                \begin{split}
                    C_k-U_{\mathrm{BE}(F_k)} &= \ket{k}\bra{k}\otimes U_{\mathrm{BE}(F_k)} + (I_{\log{\mathcal{S}}}-\ket{k}\bra{k})\otimes I_n\\
                    &=\prod_{j=0}^{\log{\mathcal{S}}} C-U_{\mathrm{BE}(F_{2^j})}.
                \end{split}                
            \end{equation}
            By applying Lemma~\ref{lem:lcu}, we can implement the linear combination of $U_{\mathrm{BE}(F_k)}$ as a state-preparation-pair $(P_L, P_R^\dagger)$ and a Fourier basis block-encoding (FBBE) subroutine depicted in Fig.~\ref{fig:ss-be}(a), wherein each layer of controlled Z-rotation gates can be synthesized as two layers of rotation gates and two fan-out X gates, and the rotation depth and count are $2\log\mathcal{S}$ and $2n\log\mathcal{S}$, respectively.
            By applying uniformly controlled rotation technique in Ref.~\cite{mottonen2004transformation}, the $\log\mathcal{S}$-qubit state preparation pair of rotation depth and count are both $2(\mathcal{S}-1)$.
            In summary, the rotation depth and count of the Fourier-based spectral block-encoding are
            \begin{equation}
                \mathcal{RD}_\mathrm{FBSBE} = 2\mathcal{S}+2\log\mathcal{S}-2
            \end{equation} and 
            \begin{equation}
                \mathcal{RC}_\mathrm{FBSBE} = 2\mathcal{S}+2n\log\mathcal{S}-2,
            \end{equation} respectively with no Toffoli gates.
            In the control version, $2\log\mathcal{S}$ Toffoli gates and one ancillary qubit are introduced for controlled FBBE, as well as $2(\mathcal{S}-2)$ Toffoli gates for controlled state preparation pair, and the rotation gates are doubled for the control behavior.
            
            \subparagraph{Monomial-based Case.} We now consider the diagonal block-encoding of the monomials $P_k(x)=x^k$, i.e., 
            \begin{equation}\label{eq:monomial_basis}          
                \begin{pmatrix}
                    0^k &  &  &  \\
                      & \left( {\frac{1}{2^n-1}} \right)^k & & \\
                      &  & \ddots & \\
                      &  &  &  1^k
                \end{pmatrix}.
            \end{equation}
            By combining Theorem~\ref{thm:qsvt} and Lemma~\ref{lem:lft}, Eq.~\eqref{eq:monomial_basis} can be $(1, \log{n}+1, 0)$-block-encoded.
            Similarly, we consider the binary representation of $k=\sum_{j=0}^{\log{\mathcal{S}}}c_j\cdot2^j, c_j\in\{0, 1\}$, then the $k^{th}$ monomial basis function can be $(1, \log{\mathcal{S}}(\log n+1), 0)$-block-encoded, by Lemma~\ref{lem:product}, as a product of matrices
            \begin{equation}
                U_{\mathrm{BE}(P_k)} = \prod_{j=0}^{\log{\mathcal{S}}} U_{\mathrm{BE}(P_{2^j})}^{c_j}.
            \end{equation}
            Consequently, the block-encoding of the monomial basis $P_k$ controlled by $\ket{k}$ can be decomposed into:
            \begin{equation}
                \begin{split}
                    C_k-U_{\mathrm{BE}(P_k)} &= \ket{k}\bra{k}\otimes U_{\mathrm{BE}(P_k)} + (I_{\log{\mathcal{S}}}-\ket{k}\bra{k})\otimes I_n\\
                    &=\prod_{j=0}^{\log{\mathcal{S}}} C-U_{\mathrm{BE}(P_{2^j})}.
                \end{split}                
            \end{equation}
            By applying Lemma~\ref{lem:lcu}, we can implement the linear combination of $U_{\mathrm{BE}(P_k)}$ via a state preparation pair $(P_L, P_R^\dagger)$ and a
             monomial basis block-encoding (MBBE), as illustrated in Fig.~\ref{fig:ss-be}(b). Herein, the $\log\mathcal{S}$-qubit state-preparation-pair is called once with $2(\mathcal{S}-1)$ rotation depth/count, the (controlled) linear function block-encoding subroutine $LBE/{LBE}^\dagger$ defined in Lemma~\ref{lem:lft} is called $1+2+\cdots+2^{\log\mathcal{S}-1} = \mathcal{S}-1$ times with $2(n-1)$ rotation depth/count and $2n(2\log\log{n}-1)$ Toffoli depth/count in each call ($4(n-1)$ and $2n(2\log\log{(2n)}-1)+2n-4$ for the controlled version), and projector-controlled $RZ$ subroutine defined in Theorem~\ref{thm:qsvt} is called $2+\cdots+2^{\log\mathcal{S}-1} = \mathcal{S}-2$ times with $2$ rotation depth/count and $2(2\log\log{n}-1)$ Toffoli depth/count in each call.
            In summary, the rotation depth and count of Monomial-based spectral block-encoding are 
            \begin{equation}
                \mathcal{RD}_\mathrm{MBSBE} = \mathcal{RC}_\mathrm{MBSBE} = 2(\mathcal{S}-1)(n+1)-2,
            \end{equation}
            and the Toffoli depth and count are
            \begin{equation}
                \mathcal{TD}_\mathrm{MBSBE} = \mathcal{TC}_\mathrm{MBSBE} = 2(\mathcal{S}-1)(n+1)(2\log\log{n}-1)+2n-4\log\log{n}-2.
            \end{equation}
            In the control version, only the $\log\mathcal{S}$-qubit state-preparation-pair, the controlled $LBE$ subroutine, and the $RZ$ gates on the ancillary qubit of QSVT need to be controlled, introducing one additional ancillary qubit, $4(\mathcal{S}-2)+2 = 4\mathcal{S}-6$ additional Toffoli gates and $2(\mathcal{S}-1)$ additional rotation gates.

            \subparagraph{Composed Case.} Finally, we consider the general case to compose a (perhaps uncompact) operator $f$ with a variable $x$ of finite spectral sparsity. By applying Theorem~\ref{thm:qsvt} to the subcircuit constructed in the above two cases to block encode the diagonal matrix $diag\left( x_0, x_1, \cdots, x_{2^n-1} \right)$, we can implement a $d$-degree polynomial transformation $P(\cdot)$ to the eigenvalues to block encode the diagonal matrix $diag\left( P(x_0), P(x_1), \cdots, P(x_{2^n-1}) \right)$, as illustrated in Fig.~\ref{fig:ss-be}(c).
            Since the subcircuit and the projector-controlled rotation are both called $d$ times if $P$ has parity $d \text{ mod } 2$, or $2d$ times if $P$ is an arbitrary polynomial, the circuit depth can be checked by a direct computation.
        \end{proof}

\subsection{Block-encoding of flux Jacobian matrix A and residual vector b}\label{sec:Abbe}
    \begin{figure}
        \centering
        \begin{tikzpicture}
            \node{\includegraphics[width=0.95\textwidth]{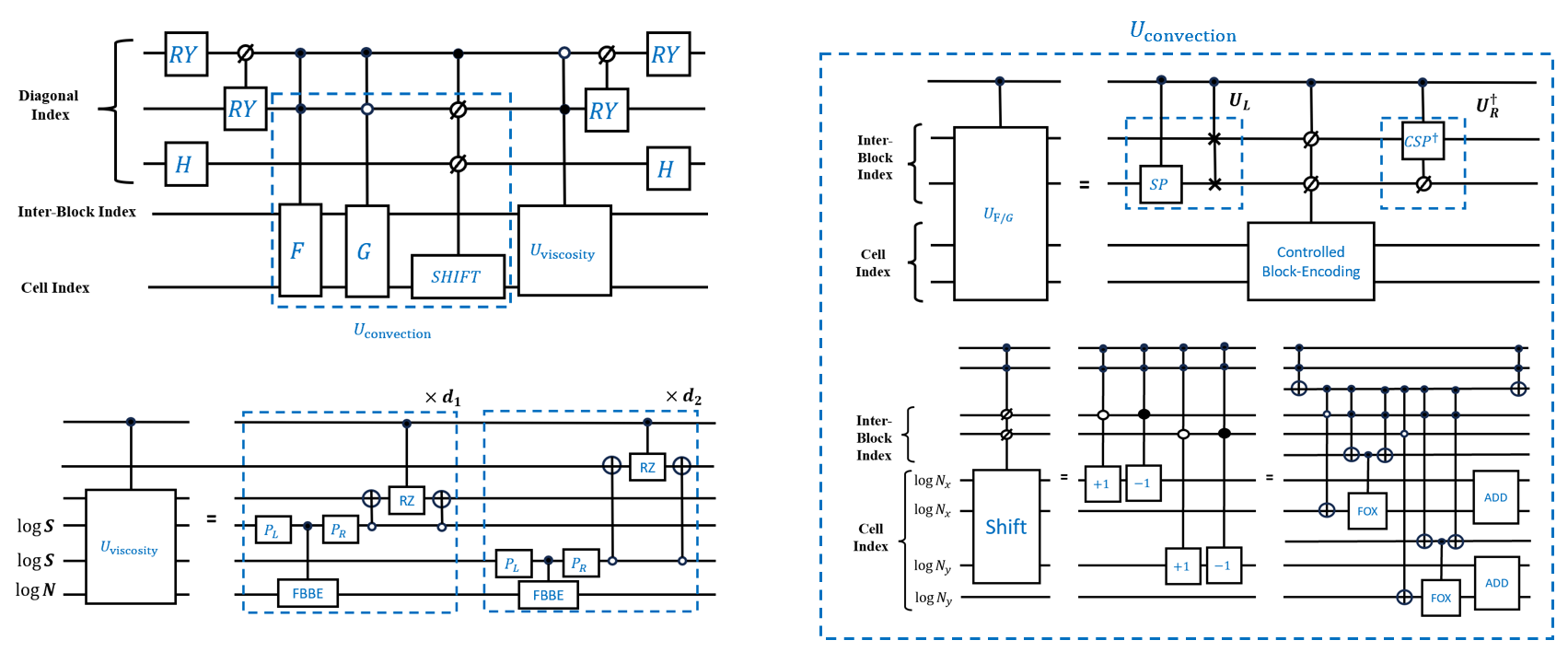}};
            \node[] at (-8.5, 3.25) {\textbf{(a)}};
            \node[] at (0, 3.25) {\textbf{(b)}};
            \node[] at (-8.5, -1) {\textbf{(c)}};
        \end{tikzpicture}
        \caption{Block-encoding circuit of Jacobian matrix $A$. (a) The global level circuit to block encode the Jacobian matrix $A$. (b) The circuit to block encode convective flux. (c) The circuit to block encode numerical dissipation for viscosity.}
        \label{fig:bea}
    \end{figure}
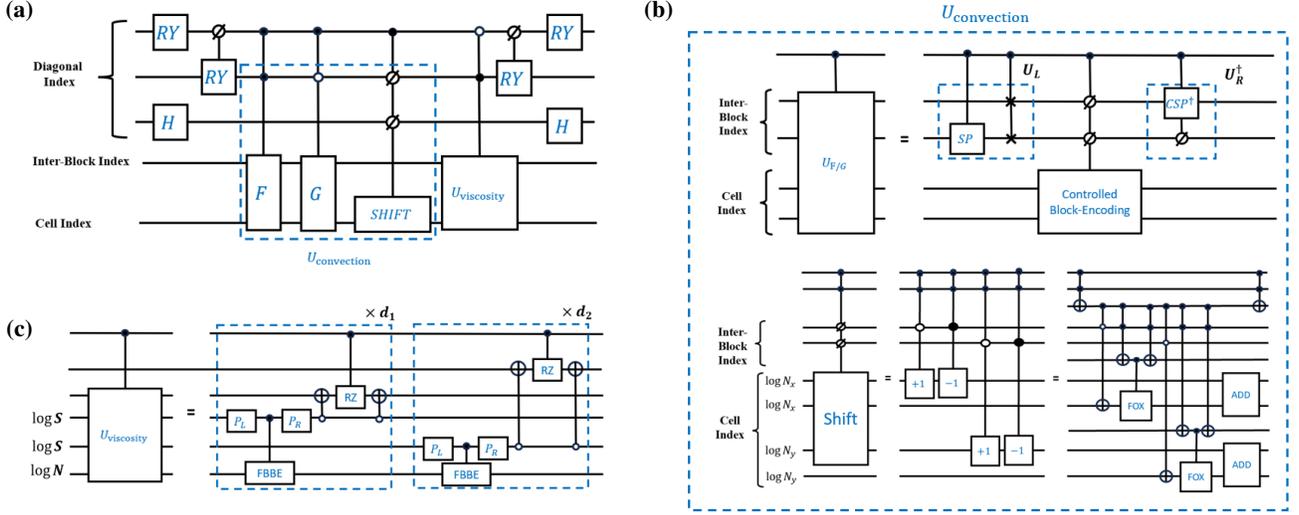
    We first describe the quantum circuit to block encode the flux Jacobian matrix $A$ in Eq.~\eqref{aeq:A}. Recall that $A$ consists of a principal diagonal $A^{(0)}$ for numerical dissipation and $2D$ diagonals $^{(\pm k)}$ of sub-matrices for convective flux at the interfaces of each cell. Consequently, at the diagonal level, we utilize a $(\log D+2)$-qubit quantum register $\ket{0}_\mathrm{Diag}^{\otimes \log D+2}$ to encode the diagonal index $0, \pm1, \pm2$ and create a superposition of these diagonals via the linear combination of $U_\mathrm{convection}$ and $U_\mathrm{viscosity}$, as depicted in Fig.~\ref{fig:bea} (a).
    Herein, we apply the state preparation pair $(P_L, P_R^\dagger)$ with rotation angles
        \begin{equation}\label{aeq:lcu}
            \overrightarrow{\theta}_D = \left(
            2\arcsin(\sqrt{\frac{4\overline{\mu}/\underline{\rho}+1}{\alpha_A}}), 2\arctan(\sqrt{\frac{\alpha_{\frac{\partial \mathcal{F}_C}{\partial W}}}{\alpha_{\frac{\partial \mathcal{G}_C}{\partial W}}}}), 2\arctan(\sqrt{{4\overline{\mu}/\underline{\rho}}}), \frac{\pi}{2}, \frac{\pi}{2}, 0, 0
            \right),
        \end{equation}
    the convective flux diagonals are block-encoded via operators $U_F$ and $U_G$ to encode the horizontal and vertical convective fluxes as a diagonal of sub-matrices defined in Eqs.~\eqref{appendix_eq:horizontal_jacobian} and \eqref{appendix_eq:vertical_jacobian}, as well as a shift operator to allocate these diagonals to the target position, as illustrated in Fig.~\ref{fig:bea} (b). Each $U_{F/G}$ operators consist of $N$ submatrices of $4\times4$ shape
        \begin{equation}
            \begin{pmatrix}
                \left.\frac{\partial \mathcal{F}}{\partial W}\right|_{1, 1} & & &\\
                & \left.\frac{\partial \mathcal{F}}{\partial W}\right|_{1,2} & &\\
                &&\ddots&\\
                &&&\left.\frac{\partial \mathcal{F}}{\partial W}\right|_{N_y,N_x}
            \end{pmatrix}.
        \end{equation}
        By (conceptually) interchanging the inter-block index and cell index, this is transformed into
        \begin{equation}
            \begin{pmatrix}
                diag(0) & diag(1) & diag(0) & diag(0)\\
                diag(\frac{\gamma-3}{2}u^2+\frac{\gamma-1}{2}v^2) & diag((3-\gamma)u) & diag((1-\gamma)v)&diag(\gamma-1)\\
        diag(-uv) & diag(v) & diag(u) & diag(0)\\
        diag\left(\left(\frac{\gamma-1}{2}u^2+\frac{\gamma-1}{2}v^2-H\right)u\right) & diag\left(H+(1-\gamma)u^2\right) & diag\left((1-\gamma)uv\right) & diag(\gamma u)
            \end{pmatrix},
        \end{equation}
        wherein each $diag(\cdot)$ is a diagonal $2^n\times2^n$ matrix of the corresponding formula evaluated on the cells $(1, 1), (1, 2), \cdots, (N_y, N_x)$.
        By combining Lemma~\ref{lem:spectral} and Lemma~\ref{lem:hierarchy}, this can be realized by a block-encoding pair $(U_L, U_R^\dagger)$ with rotation angles
        \begin{equation}\label{aeq:angle_l}
            \theta_L = 2\left(\arctan\left(\sqrt{\frac{\sum_{I'=2^{m-1}}^{2^{m}-1}\lvert A^{(k)}_{I'}\rvert^2}{\sum_{I'=0}^{2^{m-1}-1}\lvert A^{(k)}_{I'}\rvert^2}}\right), 
            \arctan\left(\sqrt{\frac{\sum_{I'=2^{m-1}-1}^{2^{m-2}}\lvert A^{(k)}_{I'}\rvert^2}{\sum_{I'=0}^{2^{m-2}-1}\lvert A^{(k)}_{I'}\rvert^2}}\right), 
            \arctan\left(\sqrt{\frac{\sum_{I'=3\times2^{m-2}}^{2^{m}-1}\lvert A^{(k)}_{I'}\rvert^2}{\sum_{I'=2^{m-1}}^{3\times2^{m-2}-1}\lvert A^{(k)}_{I'}\rvert^2}}\right), \cdots\right)
        \end{equation}
        and
        \begin{equation}\label{aeq:angle_r}
            \theta_{R(I)} = 2\left(\arctan\left(\sqrt{\frac{\sum_{I=2^{m-1}}^{2^{m}-1}\lvert A^{(k)}_{I, I'}\rvert^2}{\sum_{I=0}^{2^{m-1}-1}\lvert A^{(k)}_{I, I'}\rvert^2}}\right), 
            \arctan\left(\sqrt{\frac{\sum_{I=2^{m-1}-1}^{2^{m-2}}\lvert A^{(k)}_{I, I'}\rvert^2}{\sum_{I=0}^{2^{m-2}-1}\lvert A^{(k)}_{I, I'}\rvert^2}}\right), 
            \arctan\left(\sqrt{\frac{\sum_{I=3\times2^{m-2}}^{2^{m}-1}\lvert A^{(k)}_{I, I'}\rvert^2}{\sum_{I=2^{m-1}}^{3\times2^{m-2}-1}\lvert A^{(k)}_{I, I'}\rvert^2}}\right), \cdots\right)
        \end{equation}
        and a sequence of controlled block-encoding subroutines for each element-level formula:
        According to Lemma~\ref{lem:spectral}, $U_\rho$ is an $(\alpha_\rho, \log{\mathcal{S}}, 0)$-block-encoding of $\rho$, where $\alpha_\rho = \left\lVert \mathrm{Spec}(\rho) \right\rVert_1$ equals to the $L_1$-norm of the Fourier spectrum of $\rho$.
        $u$, $v$, and $e$ can be block-encoded in the same way with $\alpha_u = \left\lVert \mathrm{Spec}(u) \right\rVert_1$, $\alpha_v = \left\lVert \mathrm{Spec}(v) \right\rVert_1$, and $\alpha_e = \left\lVert \mathrm{Spec}(e) \right\rVert_1$.
        Utilizing Lemma~\ref{lem:product} and Theorem~\ref{thm:qsvt},  $uv$ and $u^2$ can be $(\alpha_u\alpha_v, 2\log{\mathcal{S}}, 0)$-block-encoded by unitary $U_{uv}$ and $(\alpha_u^2, \log{\mathcal{S}}+1, 0)$-block-encoded by $U_{u^2}$, respectively. 
        Since $H$ is the linear combination of $e$, $u^2$ and $v^2$, it can be $(\alpha_H, \log{\mathcal{S}}+3, 0)$-block-encoded by Lemma~\ref{lem:lcu}, where
        \begin{equation}
            \alpha_H = \gamma\alpha_e+\frac{\gamma-1}{2}(\alpha_u^2+\alpha_v^2)
        \end{equation}
        is the corresponding normalization constant.
        By Lemma~\ref{lem:hierarchy}, the horizontal convection flux Jacobian matrix $\frac{\partial \mathcal{F}_C}{\partial W}$ and the vertical convection flux Jacobian matrix $\frac{\partial \mathcal{G}_C}{\partial W}$ can be $\left(\alpha_{\frac{\partial \mathcal{F}}{\partial W}}, 2\log{\mathcal{S}}+5, 0\right)$- and $\left(\alpha_{\frac{\partial \mathcal{G}}{\partial W}}, 2\log{\mathcal{S}}+5, 0\right)$-block-encoded with
        \begin{equation}
            \begin{split}
                \alpha_{\frac{\partial \mathcal{F}_C}{\partial W}} =& 1 + \left(\frac{3-\gamma}{2}\alpha_u+\frac{\gamma-1}{2}\alpha_v\right) +
            (3-\gamma)\alpha_u+
            (\gamma-1)\alpha_v+
            (\gamma-1)+
            \alpha_u\alpha_v+
            \alpha_u+
            \alpha_v+\\
            &\left[ {(\gamma-1)}(\alpha_u^2+\alpha_v^2) + \gamma\alpha_e\right]\alpha_u+
            \left[\gamma\alpha_e+\frac{\gamma-1}{2}(3\alpha_u^2+\alpha_v^2)\right]+
            (\gamma-1)\alpha_u\alpha_v+
            \gamma\alpha_u\\
            =&(\gamma-1)\alpha_u^3 + (\gamma-1)\alpha_u^2\alpha_v 
            + \frac{3(\gamma-1)}{2}\alpha_u^2 + \frac{\gamma-1}{2}\alpha_v^2
            + \gamma\alpha_u\alpha_v + \gamma\alpha_u\alpha_e 
            + \frac{11-\gamma}{2}\alpha_u +\frac{3\gamma-1}{2}\alpha_v +\gamma\alpha_e+\gamma\\
            =&0.4\alpha_u^3 + 0.4\alpha_u^2\alpha_v 
            + 0.6\alpha_u^2 + 0.2\alpha_v^2
            + 1.4\alpha_u\alpha_v + 1.4\alpha_u\alpha_e 
            + 4.8\alpha_u +1.6\alpha_v +1.4\alpha_e+1.4,
            \end{split}
        \end{equation}
        and
        \begin{equation}
            \begin{split}
                \alpha_{\frac{\partial \mathcal{G}_C}{\partial W}} =& 1 + \left(\frac{3-\gamma}{2}\alpha_v+\frac{\gamma-1}{2}\alpha_u\right) +
            (3-\gamma)\alpha_v+
            (\gamma-1)\alpha_u+
            (\gamma-1)+
            \alpha_u\alpha_v+
            \alpha_u+
            \alpha_v+\\
            &\left[ {(\gamma-1)}(\alpha_u^2+\alpha_v^2)-\gamma\alpha_e \right]\alpha_v+
            \left[\gamma\alpha_e-\frac{\gamma-1}{2}(3\alpha_v^2+\alpha_u^2)\right]+
            (\gamma-1)\alpha_u\alpha_v+
            \gamma\alpha_v\\
            =&(\gamma-1)\alpha_v^3 + (\gamma-1)\alpha_v^2\alpha_u 
            + \frac{3(\gamma-1)}{2}\alpha_v^2 + \frac{\gamma-1}{2}\alpha_u^2
            + \gamma\alpha_u\alpha_v + \gamma\alpha_v\alpha_e 
            + \frac{11-\gamma}{2}\alpha_v +\frac{3\gamma-1}{2}\alpha_u +\gamma\alpha_e+\gamma\\
            =&0.4\alpha_v^3 + 0.4\alpha_v^2\alpha_u 
            + 0.6\alpha_v^2 + 0.2\alpha_u^2
            + 1.4\alpha_u\alpha_v + 1.4\alpha_v\alpha_e 
            + 4.8\alpha_v +1.6\alpha_u +1.4\alpha_e+1.4.
            \end{split}
        \end{equation}
        The shift operator is then used to allocate each diagonal by implementing arithmetic on the corresponding horizontal or vertical cell index as
        \begin{equation}\label{aeq:shift_p1}
            \mathrm{Shift}(+1)
            \begin{pmatrix}
                \left.\frac{\partial \mathcal{F}_C}{\partial W}\right|_{1, 1} & & &\\
                & \left.\frac{\partial \mathcal{F}_C}{\partial W}\right|_{1,2} & &\\
                &&\ddots&\\
                &&&\left.\frac{\partial \mathcal{F}_C}{\partial W}\right|_{N_y,N_x}
            \end{pmatrix}
            =
            \begin{pmatrix}
                &\left.\frac{\partial \mathcal{F}_C}{\partial W}\right|_{1, 1} & & \\
                && \left.\frac{\partial \mathcal{F}_C}{\partial W}\right|_{1,2} & \\
                &&&\ddots\\
                \left.\frac{\partial \mathcal{F}_C}{\partial W}\right|_{N_y,N_x}&&&
            \end{pmatrix},
        \end{equation}
        and
        \begin{equation}\label{aeq:shift_m1}
            \mathrm{Shift}(-1)
            \begin{pmatrix}
                \left.\frac{\partial \mathcal{F}_C}{\partial W}\right|_{1, 1} & & &\\
                & \left.\frac{\partial \mathcal{F}_C}{\partial W}\right|_{1,2} & &\\
                &&\ddots&\\
                &&&\left.\frac{\partial \mathcal{F}_C}{\partial W}\right|_{N_y,N_x}
            \end{pmatrix}
            =
            \begin{pmatrix}
                & & & -\left.\frac{\partial \mathcal{F}_C}{\partial W}\right|_{1, 1}\\
                -\left.\frac{\partial \mathcal{F}_C}{\partial W}\right|_{1,2}&&& \\
                &\ddots&&\\
                &&-\left.\frac{\partial \mathcal{F}_C}{\partial W}\right|_{N_y,N_x}&
            \end{pmatrix}.
        \end{equation}   
        The numerical dissipation term for viscosity can be block-encoded in a similar way as depicted in Fig.~\ref{fig:bea} (c).
        By applying Theorem~\ref{thm:qsvt} to implement the polynomial approximation $\mathbb{P}_\mu(U_T)$ of the block-encoding unitary $U_T$ with truncation error to be $\epsilon_\mu$, we derive an $\left(2\overline{\mu},  \log{\mathcal{S}}+5, \epsilon_\mu \right)$-block-encoding of $\mu$.
        Similarly, one can implement an $\left(2/\underline{\rho},  \log{\mathcal{S}}+1, \epsilon_{1/\rho} \right)$-block-encoding of $1/\rho$.
        Consequently, the diagonal numerical viscosity $\frac{\mu}{\rho}$ can be $\left({4\overline{\mu}}/{\underline{\rho}}, 2\log{\mathcal{S}}+6, 2\overline{\mu}\epsilon_{1/\rho} + 2\epsilon_\mu/\underline{\rho}\right)$-block-encoded. 
        In summary, A can be $(\alpha_A, 2\log{\mathcal{S}}+9, \epsilon_A)$-block-encoded by Algorithm~\ref{alg:bea}, wherein
        \begin{equation}
            \begin{split}
                \alpha_A =& 2\alpha_{\frac{\partial \mathcal{F}_C}{\partial W}} + 2\alpha_{\frac{\partial \mathcal{G}_C}{\partial W}} + 4\overline{\mu}/\underline{\rho}+1\\
                =& 0.8\alpha_u^3 + 0.8\alpha_v^3 
                + 0.8\alpha_u^2\alpha_v
                + 0.8\alpha_v^2\alpha_u
                + 1.6\alpha_u^2
                + 1.6\alpha_v^2
                + 2.8\alpha_u\alpha_e
                + 2.8\alpha_v\alpha_e
                + 5.6\alpha_u\alpha_v\\
                &+ 12.8\alpha_u + 12.8\alpha_v + 5.6\alpha_e + 4\overline{\mu}/\underline{\rho}+6.6,
            \end{split}            
        \end{equation}
        and
        \begin{equation}\label{aeq:epsilon_A}
            \epsilon_A = 2\overline{\mu}\epsilon_{1/\rho} + 2\epsilon_\mu/\underline{\rho}.
        \end{equation}

    \begin{algorithm}[H]
        \caption{Hierarchy Spectral Block-Encoding of Jacobian Matrix (Conventional Synthesis Only)}
        \label{alg:bea}
        \begin{algorithmic}[1]
            \Require Spectrum $\mathrm{Spec}(\rho)$, $\mathrm{Spec}(u)$, $\mathrm{Spec}(v)$, $\mathrm{Spec}(e)$,             
            Viscosity upper bound $\Bar{\mu}$, Density lower bound $\underline{\rho}$, Truncation degrees $d_\mu, d_\rho$,
            
            Quantum registers $\ket{0}_\mathrm{Diag}^{\otimes \log D+2}$, $\ket{0}_\mathrm{Block}^{\otimes 2\log(D+2)}$, $\ket{0}_\mathrm{Cell}^{\otimes n}$, $\ket{0}_\mathrm{Spectral}^{\otimes (2\mathcal{S}+3)}$.
            \Ensure Block-encoding quantum circuit $\mathrm{BE}(A)$.
            \State Classically compute normalization constants $\alpha_\rho, \alpha_u, \alpha_v, \alpha_e, \alpha_H, \alpha_{\frac{\partial \mathcal{F}_C}{\partial W}}, \alpha_{\frac{\partial \mathcal{G}_C}{\partial W}}, \alpha_A$; \textcolor{myblue}{$\triangleright\mathcal{O}(\mathcal{S})$}
            \State Classically compute rotation angles $\overrightarrow{\theta}_D$ in Eq.~\eqref{aeq:lcu}, $(\overrightarrow{\theta}_L^{\pm k}, \overrightarrow{\theta}_R^{\pm k})$ in Eqs.~(\ref{aeq:angle_l}-\ref{aeq:angle_r}), and phase angles $\overrightarrow{\phi}_\mu$, $\overrightarrow{\phi}_{1/\rho}$, $\overrightarrow{\phi}_{x^2}$;
            \textcolor{myblue}{$\triangleright\mathcal{O}(D^2)$}
            \State Apply $P_L(\overrightarrow{\theta}_D)$ on $\ket{0}_\mathrm{Diag}^{\otimes \log D+2}$; \textcolor{mypink}{ $\triangleright\mathcal{RD}=\mathcal{RC}=D+1$}
            \State \textcolor{myred}{/* Block-encoding of convective flux*/}
            \For{$k\in[D]$}
                \For{$Sign \in \{+, -\}$} 
                    \State Apply $U_L(\overrightarrow{\theta}_L^{\pm k})$ in Eq.~\eqref{aeq:u_l} on $\ket{0, I}_\mathrm{Block}^{\otimes 2\log(D+2)}$ controlled by $\ket{0, k, \pm}_\mathrm{Diag}^{\otimes \log D+2}$ and parallel SWAP on $\ket{I', I}_\mathrm{Block}^{\otimes 2\log(D+2)}$;
                    \State \textcolor{mypink}{$\triangleright\mathcal{TD}=2D+\lceil\log{(\lceil\log D\rceil+2)}\rceil, \mathcal{TC}=2D+\lceil\log D\rceil+1$, $\mathcal{RD}=\mathcal{RC}=2D+2$}
                    \For{$I, I' \in [D+2]\times[D+2]$}
                        \State Apply $U_{I,I'}^{\pm k}$ on $\ket{0}_\mathrm{Cell}^{\otimes n}\ket{0}_\mathrm{Spectral}^{\otimes (2\mathcal{S}+3)}$ controlled by $\ket{0, k, \pm}_\mathrm{Diag}^{\otimes \log D+2}\ket{I, I'}_\mathrm{Block}^{\otimes 2\log(D+2)}$:
                        \textcolor{myred}{/* e.g., $U_{I,I'}^{+1}=\mathrm{BE}\left(+\frac{\partial \mathcal{F}_C}{\partial W}\right)$ */}
                        \If{$(I, I') \in \{(1, 1), (2, 2), (3, 2)\}$ for $u$ \textbf{ or } $(I, I') \in \{(1, 2), (2, 1)\}$ for $v$}
                            \State Apply $P_L(u/v)$ on $\ket{0}_\mathrm{Spectral}^{\otimes \mathcal{S}}$ controlled by $\ket{0, k, \pm}_\mathrm{Diag}^{\otimes \log D+2}\ket{I, I'}_\mathrm{Block}^{\otimes 2\log(D+2)}$;
                            \State \textcolor{mypink}{ $\triangleright\mathcal{TD}=2\mathcal{S}+\lceil\log{(2\lceil\log{(D+2)}\rceil+\lceil\log D\rceil+2)}\rceil, \mathcal{TC}=2\mathcal{S}+2\lceil\log{(D+2)}\rceil+\lceil\log D\rceil+1$, $\mathcal{RD}=\mathcal{RC}=2\mathcal{S}+2$}
                            \State Apply $U_\mathrm{FBBE}$ on $\ket{0}_\mathrm{Cell}^{\otimes n}\ket{0}_\mathrm{Spectral}^{\otimes (2\mathcal{S}+3)}$ controlled by $\ket{0, k, \pm}_\mathrm{Diag}^{\otimes \log D+2}\ket{I, I'}_\mathrm{Block}^{\otimes 2\log(D+2)}$;
                            \State \textcolor{mypink}{ $\triangleright\mathcal{TD}=\mathcal{TC}=2\mathcal{S}$, $\mathcal{RD}=2\log\mathcal{S}, \mathcal{RC}=2n\log\mathcal{S}$}
                            \State Apply $P_R^\dagger(u/v)$ on $\ket{0}_\mathrm{Spectral}^{\otimes \mathcal{S}}$ controlled by $\ket{0, k, \pm}_\mathrm{Diag}^{\otimes \log D+2}\ket{I, I'}_\mathrm{Block}^{\otimes 2\log(D+2)}$;
                            \State \textcolor{mypink}{ $\triangleright\mathcal{TD}=2\mathcal{S}+\lceil\log{(2\lceil\log{(D+2)}\rceil+\lceil\log D\rceil+2)}\rceil, \mathcal{TC}=2\mathcal{S}+2\lceil\log{(D+2)}\rceil+\lceil\log D\rceil+1$, $\mathcal{RD}=\mathcal{RC}=2\mathcal{S}+2$}
                        \ElsIf{$(I, I') \in \{(3, 2), (2, 0)\}$ for $uv$} \textcolor{myred}{/* Product of block-encodings of $u$ and $v$ */}
                            \State Apply $P_L(u)\otimes P_L(v)$ on $\ket{0}_\mathrm{Spectral}^{\otimes 2\mathcal{S}}$ controlled by $\ket{0, k, \pm}_\mathrm{Diag}^{\otimes \log D+2}\ket{I, I'}_\mathrm{Block}^{\otimes 2\log(D+2)}$;
                            \State \textcolor{mypink}{ $\triangleright\mathcal{TD}=4\mathcal{S}+2\lceil\log{(2\lceil\log{(D+2)}\rceil+\lceil\log D\rceil+2)}\rceil, \mathcal{TC}=4\mathcal{S}+4\lceil\log{(D+2)}\rceil+2\lceil\log D\rceil+2$, $\mathcal{RD}=\mathcal{RC}=4\mathcal{S}+4$}
                            \State Apply two $U_\mathrm{FBBE}$ on $\ket{0}_\mathrm{Cell}^{\otimes n}\ket{0}_\mathrm{Spectral}^{\otimes (2\mathcal{S}+3)}$ controlled by $\ket{0, k, \pm}_\mathrm{Diag}^{\otimes \log D+2}\ket{I, I'}_\mathrm{Block}^{\otimes 2\log(D+2)}$;
                            \State \textcolor{mypink}{ $\triangleright\mathcal{TD}=\mathcal{TC}=4\mathcal{S}$, $\mathcal{RD}=4\log\mathcal{S}, \mathcal{RC}=4n\log\mathcal{S}$}
                            \State Apply $P_R^\dagger(u)\otimes P_R^\dagger(v)$ on $\ket{0}_\mathrm{Spectral}^{\otimes 2\mathcal{S}}$ controlled by $\ket{0, k, \pm}_\mathrm{Diag}^{\otimes \log D+2}\ket{I, I'}_\mathrm{Block}^{\otimes 2\log(D+2)}$;
                            \State \textcolor{mypink}{ $\triangleright\mathcal{TD}=4\mathcal{S}+2\lceil\log{(2\lceil\log{(D+2)}\rceil+\lceil\log D\rceil+2)}\rceil, \mathcal{TC}=4\mathcal{S}+4\lceil\log{(D+2)}\rceil+2\lceil\log D\rceil+2$, $\mathcal{RD}=\mathcal{RC}=4\mathcal{S}+4$}
                        \ElsIf{$(I, I') = (1, 0) \text{ for } \frac{\gamma-3}{2}u^2+\frac{\gamma-1}{2}v^2 \textbf{ or } (3, 1) \text{ for } H+(1-\gamma)u^2$} \textcolor{myred}{/* Linear combination of $u^2$, $v^2$, and $e$ */}
                            \State Apply $P_L(\cdot)$ on $\ket{0}_\mathrm{Spectral}^{\otimes 2}$ controlled by $\ket{0, k, \pm}_\mathrm{Diag}^{\otimes \log D+2}\ket{I, I'}_\mathrm{Block}^{\otimes 2\log(D+2)}$;
                            \State \textcolor{mypink}{ $\triangleright\mathcal{TD}=\lceil\log{(2\lceil\log{(D+2)}\rceil+\lceil\log D\rceil+2)}\rceil + 8, \mathcal{TC}=2\lceil\log{(D+2)}\rceil+\lceil\log D\rceil+9$, $\mathcal{RD}=\mathcal{RC}=10$}
                            \State Apply $\mathrm{QSVT}(U_u, \overrightarrow{\phi}_{x^2})$ on $\ket{0}_\mathrm{Cell}^{\otimes n}\ket{0}_\mathrm{Spectral}^{\otimes (\mathcal{S}+1)}$ controlled by $\ket{0, k, \pm}_\mathrm{Diag}^{\otimes \log D+2}\ket{I, I'}_\mathrm{Block}^{\otimes 2\log(D+2)}\ket{0}_\mathrm{Spectral}^{\otimes 2}$;
                            \State \textcolor{mypink}{ $\triangleright\mathcal{TD}= 4\lceil\log{(\lceil\log D\rceil+2)}\rceil+4\lceil\log\lceil\log\mathcal{S}\rceil\rceil,\mathcal{TC}=4\lceil\log\mathcal{S}D\rceil$, $\mathcal{RD}=4\mathcal{S}+4\log\mathcal{S}, \mathcal{RC}=4n\log\mathcal{S}+4\mathcal{S}$}
                            \State Apply $\mathrm{QSVT}(U_v, \overrightarrow{\phi}_{x^2})$ on $\ket{0}_\mathrm{Cell}^{\otimes n}\ket{0}_\mathrm{Spectral}^{\otimes (\mathcal{S}+1)}$ controlled by $\ket{0, k, \pm}_\mathrm{Diag}^{\otimes \log D+2}\ket{I, I'}_\mathrm{Block}^{\otimes 2\log(D+2)}\ket{0}_\mathrm{Spectral}^{\otimes 2}$;
                            \State \textcolor{mypink}{ $\triangleright\mathcal{TD}= 4\lceil\log{(\lceil\log D\rceil+2)}\rceil+4\lceil\log\lceil\log\mathcal{S}\rceil\rceil,\mathcal{TC}=4\lceil\log\mathcal{S}D\rceil$, $\mathcal{RD}=4\mathcal{S}+4\log\mathcal{S}, \mathcal{RC}=4n\log\mathcal{S}+4\mathcal{S}$}
                            \State Apply $U_e$ on $\ket{0}_\mathrm{Cell}^{\otimes n}\ket{0}_\mathrm{Spectral}^{\otimes (\mathcal{S})}$ controlled by $\ket{0, k, \pm}_\mathrm{Diag}^{\otimes \log D+2}\ket{I, I'}_\mathrm{Block}^{\otimes 2\log(D+2)}\ket{0}_\mathrm{Spectral}^{\otimes 2}$;
                            \State \textcolor{mypink}{ $\triangleright\mathcal{TD}=6\mathcal{S}+2\lceil\log{(2\lceil\log{(D+2)}\rceil+\lceil\log D\rceil+4)}\rceil, \mathcal{TC}=6\mathcal{S}+4\lceil\log{(D+2)}\rceil+2\lceil\log D\rceil+6$}
                            \State \textcolor{mypink}{ $\triangleright\mathcal{RD}=4\mathcal{S}+2\log\mathcal{S}+4, \mathcal{RC}=2n\log\mathcal{S}+4\mathcal{S}+4$}
                            \State Apply $P_R^\dagger(\cdot)$ on $\ket{0}_\mathrm{Spectral}^{\otimes 2}$ controlled by $\ket{0, k, \pm}_\mathrm{Diag}^{\otimes \log D+2}\ket{I, I'}_\mathrm{Block}^{\otimes 2\log(D+2)}$;
                            \State \textcolor{mypink}{ $\triangleright\mathcal{TD}=\lceil\log{(2\lceil\log{(D+2)}\rceil+\lceil\log D\rceil+2)}\rceil + 8, \mathcal{TC}=2\lceil\log{(D+2)}\rceil+\lceil\log D\rceil+9$, $\mathcal{RD}=\mathcal{RC}=10$}
                        \ElsIf{$(I, I') = (3, 0)$ for $\left(\frac{\gamma-1}{2}u^2+\frac{\gamma-1}{2}v^2-H\right)u$} \textcolor{myred}{/* Product of block-encodings of $u$ and $\left(\frac{\gamma-1}{2}u^2+\frac{\gamma-1}{2}v^2-H\right)$ */}
                            \State Apply $\mathrm{LCU}(u^2, v^2, e)$ on $\ket{0}_\mathrm{Cell}^{\otimes n}\ket{0}_\mathrm{Spectral}^{\otimes (\mathcal{S}+3)}$ controlled by $\ket{0, k, \pm}_\mathrm{Diag}^{\otimes \log D+2}\ket{I, I'}_\mathrm{Block}^{\otimes 2\log(D+2)}$;
                            \State \textcolor{mypink}{ $\triangleright\mathcal{TD}\leq6\mathcal{S}+4\lceil\log{(2\lceil\log{(D+2)}\rceil+\lceil\log D\rceil+4)}\rceil+8\lceil\log{(\lceil\log D\rceil+2)}\rceil+8\lceil\log\lceil\log\mathcal{S}\rceil\rceil + 16$}
                            \State \textcolor{mypink}{ $\triangleright \mathcal{TC}=6\mathcal{S}+8\lceil\log\mathcal{S}D\rceil+8\lceil\log{(D+2)}\rceil+4\lceil\log D\rceil+24,\mathcal{RD}=12\mathcal{S}+10\log\mathcal{S}+24, \mathcal{RC}=10n\log\mathcal{S}+12\mathcal{S}+24$}
                            \State Apply $U_u$ on $\ket{0}_\mathrm{Cell}^{\otimes n}\ket{0}_\mathrm{Spectral}^{\otimes (\mathcal{S})}$ controlled by $\ket{0, k, \pm}_\mathrm{Diag}^{\otimes \log D+2}\ket{I, I'}_\mathrm{Block}^{\otimes 2\log(D+2)}$;
                            \State \textcolor{mypink}{ $\triangleright\mathcal{TD}=6\mathcal{S}+2\lceil\log{(2\lceil\log{(D+2)}\rceil+\lceil\log D\rceil+2)}\rceil, \mathcal{TC}=6\mathcal{S}+4\lceil\log{(D+2)}\rceil+2\lceil\log D\rceil+2$}
                            \State \textcolor{mypink}{ $\triangleright\mathcal{RD}=4\mathcal{S}+2\log\mathcal{S}+4, \mathcal{RC}=2n\log\mathcal{S}+4\mathcal{S}+4$}
                        \EndIf
                    \EndFor
                    \State Apply $U_R^\dagger(\overrightarrow{\theta}_R^{\pm k})$ in Eq.~\eqref{aeq:u_r} on $\ket{I, I'}_\mathrm{Block}^{\otimes 2\log(D+2)}$ controlled by $\ket{0, k, \pm}_\mathrm{Diag}^{\otimes \log D+2}$;
                    \State \textcolor{mypink}{$\triangleright\mathcal{TD}=2D^2+4D+\lceil\log{(\lceil\log\rceil D+2)}\rceil, \mathcal{TC}=2D^2+4D+\lceil\log D\rceil+1$, $\mathcal{RD}=\mathcal{RC}=2D^2+6D+4$}
                    \State Apply Quantum Arithmetic on $\ket{0}_\mathrm{Cell}^{\otimes n}$ controlled by  $\ket{0, k, \pm}_\mathrm{Diag}^{\otimes \log D+2}$; \textcolor{mypink}{$\triangleright\mathcal{TD}=\mathcal{TC}=2\log n+9$}
                \EndFor
            \EndFor
            \State \textcolor{myred}{/* Block-encoding of numerical dissipation for viscosity*/}
            \State Apply $\mathrm{QSVT}(U_e, \overrightarrow{\phi}_{\mu})$ on $\ket{0}_\mathrm{Cell}^{\otimes n}\ket{0}_\mathrm{Spectral}^{\otimes (\mathcal{S}+1)}$ controlled by $\ket{1, 0, 0}_\mathrm{Diag}^{\otimes \log D+2}$;
            \State \textcolor{mypink}{ $\triangleright\mathcal{TD}= 2d_\mu(\lceil\log{(\lceil\log D\rceil+2)}\rceil+\lceil\log\lceil\log\mathcal{S}\rceil\rceil),\mathcal{TC}=2d_\mu\lceil\log\mathcal{S}D\rceil$, $\mathcal{RD}=2d_\mu(\mathcal{S}+\log\mathcal{S}), \mathcal{RC}=2d_\mu(n\log\mathcal{S}+\mathcal{S})$}
            \State Apply $\mathrm{QSVT}(U_\rho, \overrightarrow{\phi}_{1/\rho})$ on $\ket{0}_\mathrm{Cell}^{\otimes n}\ket{0}_\mathrm{Spectral}^{\otimes (\mathcal{S}+1)}$ controlled by $\ket{1, 0, 0}_\mathrm{Diag}^{\otimes \log D+2}$;
            \State \textcolor{mypink}{ $\triangleright\mathcal{TD}= 2d_\rho(\lceil\log{(\lceil\log D\rceil+2)}\rceil+\lceil\log\lceil\log\mathcal{S}\rceil\rceil),\mathcal{TC}=2d_\rho\lceil\log\mathcal{S}D\rceil$, $\mathcal{RD}=2d_\rho(\mathcal{S}+\log\mathcal{S}), \mathcal{RC}=2d_\rho(n\log\mathcal{S}+\mathcal{S})$}
            \State Apply $P_R^\dagger(\overrightarrow{\theta}_D)$ on $\ket{0}_\mathrm{DIAG}^{\otimes 3}$; \textcolor{mypink}{/* $\mathcal{RD}=\mathcal{RC}=D+1$ */}
            \State \Return  A $\left({4\overline{\mu}}/{\underline{\rho}}, 2\log{\mathcal{S}}+6, 2\overline{\mu}\epsilon_{1/\rho} + 2\epsilon_\mu/\underline{\rho}\right)$-block-encoding of Jacobian Matrix $A$.
        \end{algorithmic}
    \end{algorithm}

    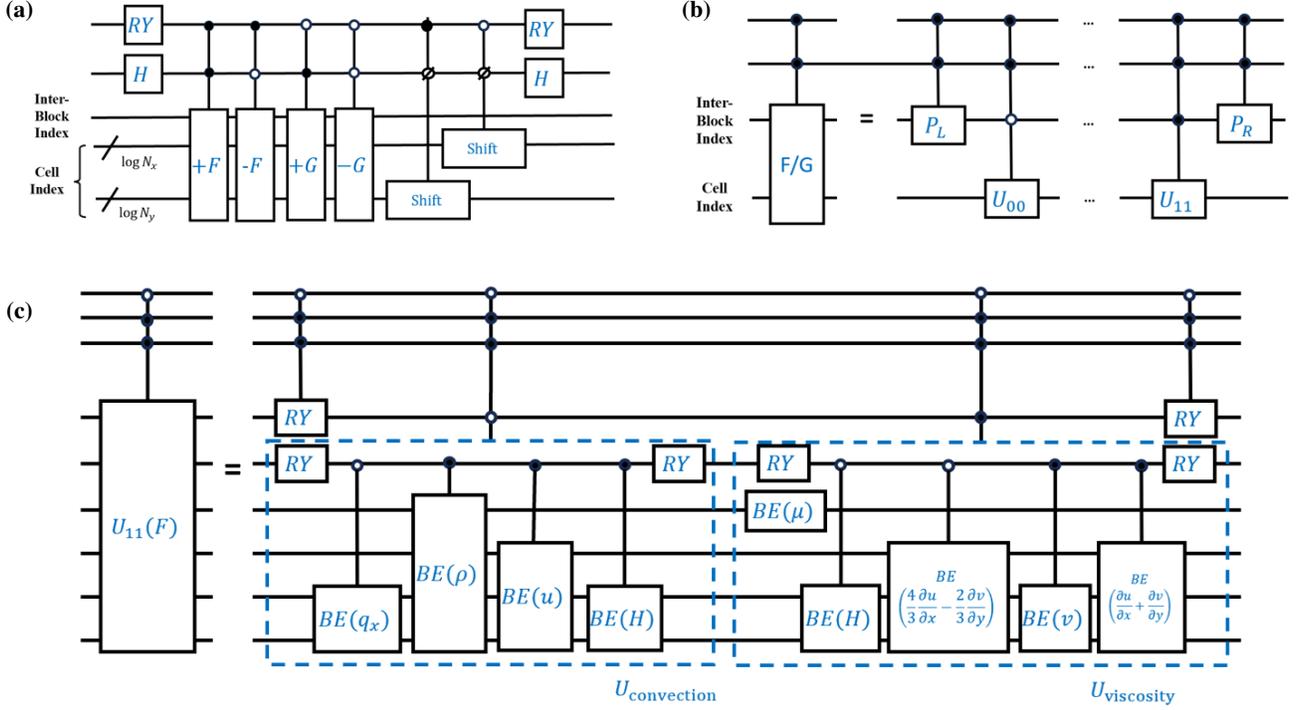
\begin{figure}
        \centering
        \begin{tikzpicture}
            \node{\includegraphics[width=0.95\textwidth]{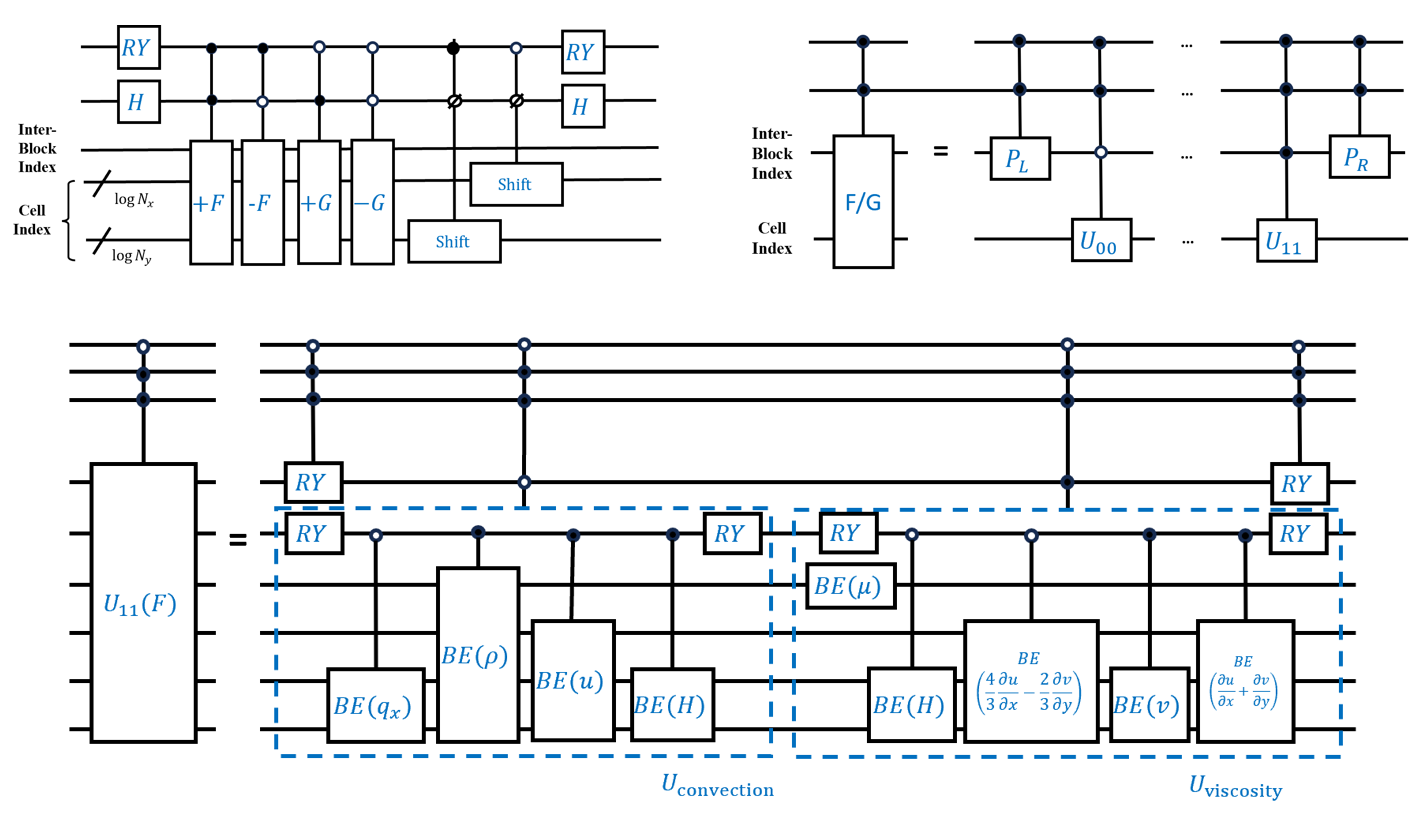}};
            \node[] at (-8.5, 4.5) {\textbf{(a)}};
            \node[] at (0.5, 4.5) {\textbf{(b)}};
            \node[] at (-8.5, 0.5) {\textbf{(c)}};
        \end{tikzpicture}
        \caption{Block-encoding circuit of residual vector $b$. (a) The global level circuit. (b) The horizontal/vertical flux level circuit. (c) The element-level circuit.}
        \label{fig:beb}
    \end{figure}
    We then describe the quantum circuit to block encode the residual vector $b$ in the middle of Fig.~\ref{fig:algorithm}(a). It is easy to check that $b$ can be block-encoded as a linear combination, given state preparation rotation angles
    \begin{equation}\label{aeq:dplus}
        \overrightarrow{\theta}_D'=2(\arctan{\left(\sqrt{\frac{\alpha_\mathcal{G}}{\alpha_\mathcal{F}}}\right)}, \frac{\pi}{4}, \frac{\pi}{4}),
    \end{equation}
    of $2D$ flux operators $\pm F, \pm G, \cdots$ followed by a shift operator, as illustrated in Fig.~\ref{fig:beb}~(a).
    Herein, each flux operator $F/G$ can be decomposed as a state preparation pair $(P_L, P_R^\dagger)$ with rotation angles
    \begin{equation}\label{aeq:angle_spp}
        \begin{split}
            \overrightarrow{\theta}_\mathrm{SPP}^{\pm1}=&\pm2\left(
            \arctan\left(\sqrt{\frac{\alpha_{\mathcal{F}(10)}+\alpha_{\mathcal{F}(11)}}{\alpha_{\mathcal{F}(00)}+\alpha_{\mathcal{F}(01)}}}\right),
            \arctan\left(\sqrt{\frac{\alpha_{\mathcal{F}(01)}}{\alpha_{\mathcal{F}(00)}}}\right),
            \arctan\left(\sqrt{\frac{\alpha_{\mathcal{F}(11)}}{\alpha_{\mathcal{F}(10)}}}\right)
            \right)\\
            \overrightarrow{\theta}_\mathrm{SPP}^{\pm2}=&\pm2\left(
            \arctan\left(\sqrt{\frac{\alpha_{\mathcal{G}(10)}+\alpha_{\mathcal{G}(11)}}{\alpha_{\mathcal{G}(00)}+\alpha_{\mathcal{G}(01)}}}\right),
            \arctan\left(\sqrt{\frac{\alpha_{\mathcal{G}(01)}}{\alpha_{\mathcal{G}(00)}}}\right),
            \arctan\left(\sqrt{\frac{\alpha_{\mathcal{G}(11)}}{\alpha_{\mathcal{G}(10)}}}\right)
            \right)
        \end{split}
    \end{equation}
    on the inter-block index and a sequence of element-level block-encoding unitaries $U_{00}, U_{01}, U_{10}, U_{11}$ on the cell index controlled by the inter-block index, as depicted in Fig.~\ref{fig:beb}~(b).
    Since in $b$ both convection and viscosity fluxes are considered simultaneously, each unitary $U_{11}(\cdot)$ contains two ingredients of $U_{convection}$ and $U_{viscosity}$, illustrated in Fig.~\ref{fig:beb}~(c). 
    A similar computation shows that the four components of the horizontal flux can be $(\alpha_{\mathcal{F}({k})}, 3\log\mathcal{S}+5, \epsilon_{\mathcal{F}({k})})$-block-encoded for $k\in\{00, 01, 10, 11\}$ with 
    \begin{align}
        \alpha_{\mathcal{F}({00})} &= \alpha_\rho\alpha_u \\
        \alpha_{\mathcal{F}({01})} &= \alpha_\rho\left(\alpha_u^2+(\gamma-1)\alpha_e\right) + \frac{2}{3}\overline{\mu}\alpha_{2\frac{\partial u}{\partial x}-\frac{\partial v}{\partial y}}\\
        \alpha_{\mathcal{F}({10})} &= \alpha_\rho\alpha_u\alpha_v + \overline{\mu}\alpha_{\frac{\partial u}{\partial x}+\frac{\partial v}{\partial y}} \\\label{aeq:alpha_flux}
        \alpha_{\mathcal{F}({11})} &=  \alpha_\rho\alpha_u\alpha_H + \alpha_{q_x} + \frac{2}{3}\overline{\mu}\alpha_H\alpha_{2\frac{\partial u}{\partial x}-\frac{\partial v}{\partial y}} + 
        \overline{\mu}\alpha_v\alpha_{\frac{\partial u}{\partial x}+\frac{\partial v}{\partial y}}
    \end{align}
    and 
    \begin{align}
        \epsilon_{\mathcal{F}({00})} &= 0 \\
        \epsilon_{\mathcal{F}({01})} &= \frac{2}{3}\alpha_{2\frac{\partial u}{\partial x}-\frac{\partial v}{\partial y}}\epsilon_{\mu}\\
        \epsilon_{\mathcal{F}({10})} &= \alpha_{\frac{\partial u}{\partial x}+\frac{\partial v}{\partial y}}\epsilon_{\mu} \\
        \epsilon_{\mathcal{F}({11})} &=  \left( \frac{2}{3}\alpha_H\alpha_{2\frac{\partial u}{\partial x}-\frac{\partial v}{\partial y}} + \alpha_v\alpha_{\frac{\partial u}{\partial x}+\frac{\partial v}{\partial y}} \right)\epsilon_\mu,
    \end{align}
    where $\alpha_{2\frac{\partial u}{\partial x}-\frac{\partial v}{\partial y}}$ and $\alpha_{\frac{\partial u}{\partial x}+\frac{\partial v}{\partial y}}$ are the $L_1$-norm of the Fourier spectral of these partial derivetives.
    Notably, these derivatives can be efficiently computed in the Fourier spectral space, and the spectral sparsity is also upper bounded by $\mathcal{S}$.
    Herein, the rotation angles are decided by the normalization constants of each underlying component. For example, the rotation angles for Eq.~\eqref{aeq:alpha_flux} can be computed as:
    \begin{equation}
        \begin{split}
            \overrightarrow{\theta}_{\mathrm{SPP}, I}^{\pm k} =& \mathrm{Amplitude-Encode}(\alpha_{{k}({I})})\\
            =& 2\left( \arctan\sqrt{\frac{\frac{2}{3}\overline{\mu}\alpha_H\alpha_{2\frac{\partial u}{\partial x}-\frac{\partial v}{\partial y}} + 
        \overline{\mu}\alpha_v\alpha_{\frac{\partial u}{\partial x}+\frac{\partial v}{\partial y}}}{\alpha_\rho\alpha_u\alpha_H + \alpha_{q_x}}}, 
        \arctan\sqrt{\frac{\alpha_\rho\alpha_u\alpha_H}{\alpha_{q_x}}}, 
        \arctan\sqrt{\frac{\overline{\mu}\alpha_v\alpha_{\frac{\partial u}{\partial x}+\frac{\partial v}{\partial y}}}{\frac{2}{3}\overline{\mu}\alpha_H\alpha_{2\frac{\partial u}{\partial x}-\frac{\partial v}{\partial y}}}}\right)
        \end{split}
    \end{equation}
    Consequently, the horizontal and vertical fluxes can be $(\alpha_\mathcal{F}, 3\log\mathcal{S}+7, \epsilon_\mathcal{F})$- and $(\alpha_\mathcal{G}, 3\log\mathcal{S}+7, \epsilon_\mathcal{G})$-block-encoded, wherein
    \begin{align}
        \alpha_\mathcal{F} =& \left(\alpha_u\alpha_H + \alpha_u^2 + \alpha_u\alpha_v + \alpha_u + \alpha_{q_x} + (\gamma-1)\alpha_e \right)\alpha_\rho + 
        \frac{2}{3}\overline{\mu}(\alpha_H+1)\alpha_{2\frac{\partial u}{\partial x}-\frac{\partial v}{\partial y}} + 
        \overline{\mu}(\alpha_v+1)\alpha_{\frac{\partial u}{\partial x}+\frac{\partial v}{\partial y}},\\
        \alpha_\mathcal{G} =& \left(\alpha_v\alpha_H + \alpha_v^2 + \alpha_u\alpha_v + \alpha_v + \alpha_{q_y} + (\gamma-1)\alpha_e \right)\alpha_\rho + 
        \frac{2}{3}\overline{\mu}(\alpha_H+1)\alpha_{2\frac{\partial v}{\partial y}-\frac{\partial u}{\partial x}} + 
        \overline{\mu}(\alpha_u+1)\alpha_{\frac{\partial u}{\partial x}+\frac{\partial v}{\partial y}},
    \end{align}
    and
    \begin{align}
        \epsilon_\mathcal{F} =& \left( \frac{4}{3}\alpha_H\alpha_{2\frac{\partial u}{\partial x}-\frac{\partial v}{\partial y}} + 2\alpha_v\alpha_{\frac{\partial u}{\partial x}+\frac{\partial v}{\partial y}} \right)\epsilon_\mu,\\
        \epsilon_\mathcal{G} =& \left( \frac{4}{3}\alpha_H\alpha_{2\frac{\partial v}{\partial y}-\frac{\partial u}{\partial x}} + 2\alpha_v\alpha_{\frac{\partial u}{\partial x}+\frac{\partial v}{\partial y}} \right)\epsilon_\mu.
    \end{align}
    In summary, $b$ can be $(\alpha_b, 3\log\mathcal{S}+9, \epsilon_b)$-block-encoded by Algorithm~\ref{alg:beb} with
    \begin{equation}
        \begin{split}
            \alpha_b =& 2\alpha_\mathcal{F} + 2\alpha_\mathcal{G}\\
        =& 2\left((\alpha_u+\alpha_v)\alpha_H + \alpha_u^2 + \alpha_v^2 + 2\alpha_u\alpha_v + \alpha_u + \alpha_v + \alpha_{q_x}+ \alpha_{q_y}  + 2(\gamma-1)\alpha_e \right)\alpha_\rho\\
        &+ \frac{4}{3}\overline{\mu}(\alpha_H+1)(\alpha_{2\frac{\partial u}{\partial x}-\frac{\partial v}{\partial y}}+\alpha_{2\frac{\partial v}{\partial y}-\frac{\partial u}{\partial x}}) + 2\overline{\mu}(\alpha_u+\alpha_v+2)\alpha_{\frac{\partial u}{\partial x}+\frac{\partial v}{\partial y}}\\
        =& 2\left[(\alpha_u+\alpha_v)(1.4\alpha_e+0.2\alpha_u^2+0.2\alpha_v^2) + \alpha_u^2 + \alpha_v^2 + 2\alpha_u\alpha_v + \alpha_u + \alpha_v + \alpha_{q_x}+ \alpha_{q_y}  + 0.8\alpha_e \right]\alpha_\rho\\
        &+ \frac{4}{3}\overline{\mu}(1.4\alpha_e+0.2\alpha_u^2+0.2\alpha_v^2+1)\left(\alpha_{2\frac{\partial u}{\partial x}-\frac{\partial v}{\partial y}}+\alpha_{2\frac{\partial v}{\partial y}-\frac{\partial u}{\partial x}}\right) + 
        2\overline{\mu}(\alpha_u+\alpha_v+2)\alpha_{\frac{\partial u}{\partial x}+\frac{\partial v}{\partial y}}
        \end{split}
    \end{equation}
    and
    \begin{equation}\label{aeq:epsilon_b}
        \begin{split}
            \epsilon_b =& 2\epsilon_\mathcal{F} + 2\epsilon_\mathcal{G}
            = \left[ \frac{4}{3}(1.4\alpha_e+0.2\alpha_u^2+0.2\alpha_v^2)\left(\alpha_{2\frac{\partial u}{\partial x}-\frac{\partial v}{\partial y}}+\alpha_{2\frac{\partial v}{\partial y}-\frac{\partial u}{\partial x}}\right) + 2(\alpha_u+\alpha_v)\alpha_{\frac{\partial u}{\partial x}+\frac{\partial v}{\partial y}} \right]\epsilon_\mu.
        \end{split}
    \end{equation}

    \begin{algorithm}[H]
        \caption{Hierarchy Spectral Block-Encoding of Residual Vector (Conventional Synthesis Only)}
        \label{alg:beb}
        \begin{algorithmic}[1]
            \Require Spectrum $\mathrm{Spec}(\rho)$, $\mathrm{Spec}(u)$, $\mathrm{Spec}(v)$, $\mathrm{Spec}(e)$,             
            Viscosity upper bound $\Bar{\mu}$, Truncation degrees $d_\mu$,
            
            Quantum registers $\ket{0}_\mathrm{Diag}^{\otimes \log D+1}$, $\ket{0}_\mathrm{Block}^{\otimes \log (D+2)}$, $\ket{0}_\mathrm{Cell}^{\otimes n}$, $\ket{0}_\mathrm{Spectral}^{\otimes (2\mathcal{S}+3)}$.
            \Ensure Block-encoding quantum circuit $\mathrm{BE}(b)$.
            \State Classically compute normalization constants $\alpha_\rho, \alpha_u, \alpha_v, \alpha_e, \alpha_{2\frac{\partial u}{\partial x}-\frac{\partial v}{\partial y}}, \alpha_{2\frac{\partial v}{\partial y}-\frac{\partial u}{\partial x}}, \alpha_{\frac{\partial v}{\partial y}+\frac{\partial u}{\partial x}}, \alpha_H, \alpha_{\mathcal{F}}, \alpha_{\mathcal{G}}, \alpha_b$; \textcolor{myblue}{$\triangleright\mathcal{O}(\mathcal{S})$}
            \State Classically compute rotation angles $\overrightarrow{\theta}_D'$ in Eq.~\eqref{aeq:dplus}, $\overrightarrow{\theta}_\mathrm{SPP}^{\pm k}$ in Eq.~(\ref{aeq:angle_spp}),  and phase angles $\overrightarrow{\phi}_\mu$, $\overrightarrow{\phi}_{x^2}$;
            \textcolor{myblue}{$\triangleright\mathcal{O}(D^2)$}
            \State Apply $P_L(\overrightarrow{\theta}_D')$ on $\ket{0}_\mathrm{Diag}^{\otimes \log D+1}$; \textcolor{mypink}{ $\triangleright\mathcal{RD}=\mathcal{RC}=D+1$}
            \For{$k\in[D]$}
                \For{$Sign \in \{+, -\}$} 
                    \State \textcolor{myred}{/* Block-encoding of convective flux*/}
                    \State Apply $P_L(\overrightarrow{\theta}_\mathrm{SPP}^{\pm k})$ on $\ket{0}_\mathrm{Block}^{\otimes \log (D+2)}$ controlled by $\ket{k, \pm}_\mathrm{Diag}^{\otimes \log D+1}$;
                    \State \textcolor{mypink}{$\triangleright\mathcal{TD}=2D+\lceil\log{(\lceil\log D\rceil+1)}\rceil, \mathcal{TC}=2D+\lceil\log D\rceil$, $\mathcal{RD}=\mathcal{RC}=2D+2$}
                    \For{$I \in [D+2]$}
                        \State Apply block-encoding $U_{I}^{\pm k}$ on $\ket{0}_\mathrm{Cell}^{\otimes n}\ket{0}_\mathrm{Spectral}^{\otimes (2\mathcal{S}+3)}$ controlled by $\ket{k, \pm}_\mathrm{Diag}^{\otimes \log D+1}\ket{I}_\mathrm{Block}^{\otimes \log (D+2)}$:
                        \State Apply block-encoding $P_L(\overrightarrow{\theta}_{\mathrm{SPP}, I}^{\pm k})$ on $\ket{0}_\mathrm{Spectral}^{\otimes 2}$ controlled by $\ket{k, \pm}_\mathrm{Diag}^{\otimes \log D+1}\ket{I}_\mathrm{Block}^{\otimes \log (D+2)}$
                        \State \textcolor{mypink}{$\triangleright\mathcal{TD}=\lceil\log{(\lceil\log{(D+2)}\rceil+\lceil\log D\rceil+1)}\rceil+4, \mathcal{TC}=\lceil\log{(D+2)}\rceil+\lceil\log D\rceil+2$, $\mathcal{RD}=\mathcal{RC}=2D+2$}
                        \State \textcolor{myred}{/* For example, convective term in component $\mathcal{F}(11)$ */}
                        \State Apply $U_\rho$ on $\ket{0}_\mathrm{Spectral}^{\otimes \mathcal{S}}$ controlled by $\ket{k, \pm}_\mathrm{Diag}^{\otimes \log D+1}\ket{I}_\mathrm{Block}^{\otimes \log (D+2)}\ket{00}_\mathrm{Spectral}^{\otimes 2}$;
                        \State \textcolor{mypink}{ $\triangleright\mathcal{TD}=2\mathcal{S}+\lceil\log{(\lceil\log{(D+2)}\rceil+\lceil\log D\rceil+3)}\rceil, \mathcal{TC}=\mathcal{S}+2\lceil\log{(D+2)}\rceil+\lceil\log D\rceil+2$, $\mathcal{RD}=\mathcal{RC}=2\mathcal{S}+2$}
                        \State Apply $U_u$ on $\ket{0}_\mathrm{Spectral}^{\otimes \mathcal{S}}$ controlled by $\ket{k, \pm}_\mathrm{Diag}^{\otimes \log D+1}\ket{I}_\mathrm{Block}^{\otimes \log (D+2)}\ket{00}_\mathrm{Spectral}^{\otimes 2}$;
                        \State \textcolor{mypink}{ $\triangleright\mathcal{TD}=2\mathcal{S}+\lceil\log{(\lceil\log{(D+2)}\rceil+\lceil\log D\rceil+3)}\rceil, \mathcal{TC}=\mathcal{S}+2\lceil\log{(D+2)}\rceil+\lceil\log D\rceil+2$, $\mathcal{RD}=\mathcal{RC}=2\mathcal{S}+2$}
                        \State Apply $U_H$ on $\ket{0}_\mathrm{Cell}^{\otimes n}\ket{0}_\mathrm{Spectral}^{\otimes (\mathcal{S}+3)}$ controlled by $\ket{k, \pm}_\mathrm{Diag}^{\otimes \log D+1}\ket{I}_\mathrm{Block}^{\otimes \log (D+2)}\ket{00}_\mathrm{Spectral}^{\otimes 2}$;
                        \State \textcolor{mypink}{ $\triangleright\mathcal{TD}\leq6\mathcal{S}+4\lceil\log{(\lceil\log{(D+2)}\rceil+\lceil\log D\rceil+3)}\rceil+8\lceil\log{(\lceil\log D\rceil+1)}\rceil+8\lceil\log\lceil\log\mathcal{S}\rceil\rceil + 16$}
                        \State \textcolor{mypink}{ $\triangleright \mathcal{TC}=6\mathcal{S}+8\lceil\log\mathcal{S}D\rceil+4\lceil\log{(D+2)}\rceil+4\lceil\log D\rceil+24,\mathcal{RD}=12\mathcal{S}+10\log\mathcal{S}+24, \mathcal{RC}=10n\log\mathcal{S}+12\mathcal{S}+24$}
                        \State Apply $U_{q_x}$ on $\ket{0}_\mathrm{Spectral}^{\otimes \mathcal{S}}$ controlled by $\ket{k, \pm}_\mathrm{Diag}^{\otimes \log D+1}\ket{I}_\mathrm{Block}^{\otimes \log (D+2)}\ket{01}_\mathrm{Spectral}^{\otimes 2}$;
                        \State \textcolor{mypink}{ $\triangleright\mathcal{TD}=2\mathcal{S}+\lceil\log{(\lceil\log{(D+2)}\rceil+\lceil\log D\rceil+3)}\rceil, \mathcal{TC}=\mathcal{S}+2\lceil\log{(D+2)}\rceil+\lceil\log D\rceil+2$, $\mathcal{RD}=\mathcal{RC}=2\mathcal{S}+2$}
                        \State \textcolor{myred}{/* Viscosity term in component $\mathcal{F}(11)$*/}
                        \State Apply $\mathrm{QSVT}(U_e, \overrightarrow{\phi}_{\mu})$ on $\ket{0}_\mathrm{Cell}^{\otimes n}\ket{0}_\mathrm{Spectral}^{\otimes (\mathcal{S}+1)}$ controlled by $\ket{k, \pm}_\mathrm{Diag}^{\otimes \log D+1}\ket{I}_\mathrm{Block}^{\otimes \log (D+2)}\ket{1}_\mathrm{Spectral}^{\otimes 1}$;
                        \State \textcolor{mypink}{ $\triangleright\mathcal{TD}\leq2d_\mu(\lceil\log{(\lceil2\log (D+2)\rceil+3)}\rceil+\lceil\log\lceil\log\mathcal{S}\rceil\rceil),\mathcal{TC}=2d_\mu\lceil\log\mathcal{S}(D+2)^2\rceil$}
                        \State \textcolor{mypink}{ $\triangleright\mathcal{RD}=2d_\mu(\mathcal{S}+\log\mathcal{S}), \mathcal{RC}=2d_\mu(n\log\mathcal{S}+2\mathcal{S})$}
                        \State Apply $U_{{2\frac{\partial u}{\partial x}-\frac{\partial v}{\partial y}}}$ on $\ket{0}_\mathrm{Spectral}^{\otimes \mathcal{S}}$ controlled by $\ket{k, \pm}_\mathrm{Diag}^{\otimes \log D+1}\ket{I}_\mathrm{Block}^{\otimes \log (D+2)}\ket{10}_\mathrm{Spectral}^{\otimes 2}$;
                        \State \textcolor{mypink}{ $\triangleright\mathcal{TD}=2\mathcal{S}+\lceil\log{(\lceil\log{(D+2)}\rceil+\lceil\log D\rceil+3)}\rceil, \mathcal{TC}=\mathcal{S}+2\lceil\log{(D+2)}\rceil+\lceil\log D\rceil+2$, $\mathcal{RD}=\mathcal{RC}=2\mathcal{S}+2$}
                        \State Apply $U_H$ on $\ket{0}_\mathrm{Cell}^{\otimes n}\ket{0}_\mathrm{Spectral}^{\otimes (\mathcal{S}+3)}$ controlled by $\ket{k, \pm}_\mathrm{Diag}^{\otimes \log D+1}\ket{I}_\mathrm{Block}^{\otimes \log (D+2)}\ket{10}_\mathrm{Spectral}^{\otimes 2}$;
                        \State \textcolor{mypink}{ $\triangleright\mathcal{TD}\leq6\mathcal{S}+4\lceil\log{(\lceil\log{(D+2)}\rceil+\lceil\log D\rceil+3)}\rceil+8\lceil\log{(\lceil\log D\rceil+1)}\rceil+8\lceil\log\lceil\log\mathcal{S}\rceil\rceil + 16$}
                        \State \textcolor{mypink}{ $\triangleright \mathcal{TC}=6\mathcal{S}+8\lceil\log\mathcal{S}D\rceil+4\lceil\log{(D+2)}\rceil+4\lceil\log D\rceil+24,\mathcal{RD}=12\mathcal{S}+10\log\mathcal{S}+24, \mathcal{RC}=10n\log\mathcal{S}+12\mathcal{S}+24$}
                        \State Apply $U_{\frac{\partial v}{\partial y}+\frac{\partial u}{\partial x}}$ on $\ket{0}_\mathrm{Spectral}^{\otimes \mathcal{S}}$ controlled by $\ket{k, \pm}_\mathrm{Diag}^{\otimes \log D+1}\ket{I}_\mathrm{Block}^{\otimes \log (D+2)}\ket{11}_\mathrm{Spectral}^{\otimes 2}$;
                        \State \textcolor{mypink}{ $\triangleright\mathcal{TD}=2\mathcal{S}+\lceil\log{(\lceil\log{(D+2)}\rceil+\lceil\log D\rceil+3)}\rceil, \mathcal{TC}=\mathcal{S}+2\lceil\log{(D+2)}\rceil+\lceil\log D\rceil+2$, $\mathcal{RD}=\mathcal{RC}=2\mathcal{S}+2$}
                        \State Apply $U_v$ on $\ket{0}_\mathrm{Spectral}^{\otimes \mathcal{S}}$ controlled by $\ket{k, \pm}_\mathrm{Diag}^{\otimes \log D+1}\ket{I}_\mathrm{Block}^{\otimes \log (D+2)}\ket{11}_\mathrm{Spectral}^{\otimes 2}$;
                        \State \textcolor{mypink}{ $\triangleright\mathcal{TD}=2\mathcal{S}+\lceil\log{(\lceil\log{(D+2)}\rceil+\lceil\log D\rceil+3)}\rceil, \mathcal{TC}=\mathcal{S}+2\lceil\log{(D+2)}\rceil+\lceil\log D\rceil+2$, $\mathcal{RD}=\mathcal{RC}=2\mathcal{S}+2$}
                        \State Apply block-encoding $P_R^\dagger(\overrightarrow{\theta}_{\mathrm{SPP}, I}^{\pm k})$ on $\ket{0}_\mathrm{Spectral}^{\otimes 2}$ controlled by $\ket{k, \pm}_\mathrm{Diag}^{\otimes \log D+1}\ket{I}_\mathrm{Block}^{\otimes \log (D+2)}$
                        \State \textcolor{mypink}{$\triangleright\mathcal{TD}=\lceil\log{(\lceil\log{(D+2)}\rceil+\lceil\log D\rceil+1)}\rceil+4, \mathcal{TC}=\lceil\log{(D+2)}\rceil+\lceil\log D\rceil+2$, $\mathcal{RD}=\mathcal{RC}=2D+2$}
                    \EndFor
                    \State Apply $P_R^\dagger(\overrightarrow{\theta}_\mathrm{SPP}^{\pm k})$ on $\ket{0}_\mathrm{Block}^{\otimes \log (D+2)}$ controlled by $\ket{k, \pm}_\mathrm{Diag}^{\otimes \log D+1}$;
                    \State \textcolor{mypink}{$\triangleright\mathcal{TD}=2D+\lceil\log{(\lceil\log D\rceil+1)}\rceil, \mathcal{TC}=2D+\lceil\log D\rceil$, $\mathcal{RD}=\mathcal{RC}=2D+2$}
                    \State Apply Quantum Arithmetic on $\ket{0}_\mathrm{Cell}^{\otimes n}$ controlled by  $\ket{0, k, \pm}_\mathrm{Diag}^{\otimes \log D+2}$; \textcolor{mypink}{$\triangleright\mathcal{TD}=\mathcal{TC}=2\log n+9$}
                \EndFor
            \EndFor
            \State Apply $P_R^\dagger(\overrightarrow{\theta}_D')$ on $\ket{0}_\mathrm{DIAG}^{\otimes \log D+1}$; \textcolor{mypink}{/* $\mathcal{RD}=\mathcal{RC}=D+1$ */}
            \State \Return  
        \end{algorithmic}
    \end{algorithm}

\subsection{Sparse spectral decoding}\label{sec:ssd}
In this section, we introduce a sparse spectral method for extracting information from the quantum state to facilitate iteration and final reconstruction of the fluid field.
By applying hierarchy spectral block-encoding  and quantum linear system solver in Theorem~\ref{theorem: qlss}, we can get a $2n$-qubit quantum state $\ket{W}$ representing the conservative variables on $N$ cells. 
However, the state in our algorithm is usually dense, so that the conventional state tomography procedure can be costly to read out the information, especially the signs of each amplitude (linearly dependent on $N$)~\cite{dalzell2023end}.

Since the state tomography complexity to extract amplitudes with signs grows linearly with the dimension, we aim to extract its spectral information instead. More precisely, the spectrum can be extracted once one can efficiently conduct a $D$-dimensional discrete Fourier transformation (DFT) on the original output state. Given a $D$-dimensional array $f[x_1, x_2, ..., x_D]$ of size $N = N_1\times N_2\times ... \times N_D$, its D-dimensional DFT is defined to be 
    \begin{equation}
        \Tilde{f}[k_1, k_2, ..., k_D] = \sum_{x_1=0}^{N_1-1}\sum_{x_2=0}^{N_2-1}\cdots\sum_{x_D=0}^{N_D-1}f[x_1, x_2, ..., x_D]\cdot e^{-2\pi i\left(\sum_{j=1}^D\frac{k_jx_j}{N_j}  \right)}.
    \end{equation}
Assume $f$ to be represented as a flattened vector
    \begin{equation}
        \Vec{f} = \left( f[0, 0, ..., 0], f[0, 0, ..., 1], ..., f[N_1-1, N_2-1, ..., N_D-1] \right)^\mathsf{T}
    \end{equation}
A direct computation shows that the $D$-dimensional DFT operator $F_{N_1, N_2, ..., N_D}$ targeting on $\Vec{f}$, due to the separability of high-dimensional DFT, can be decomposed into a tensor product of $D$ DFT operators as
    \begin{equation}\label{aeq:qft}
        F_{N_1, N_2, ..., N_D}=\bigotimes_{j=1}^D F_{N_j}= \bigotimes_{j=1}^D
        \begin{pmatrix}
            1 & 1 & 1 & \cdots & 1 \\
            1 & \omega_j & \omega_j^2 & \cdots & \omega_j^{N_j-1} \\
            1 & \omega_j^2 & \omega_j^4 & \cdots & \omega_j^{2(N_j-1)} \\
            \vdots & \vdots & \vdots & \ddots & \vdots \\
            1 & \omega_j^{N_j-1} & \omega_j^{2(N_j-1)} & \cdots & \omega_j^{(N_j-1)(N_j-1)}
        \end{pmatrix},
    \end{equation}
where $\omega_j=e^{2\pi i/N_j}$ is a primitive $N_j$-th root of unity.
Consequently, $D$ parallel inverse quantum fourier transformations on $n_j=\log{N_j} (1\leq j\leq D)$ qubits can transform the quantum linear system solver's original output state
    \begin{equation}
        \ket{W} = 
        \frac{1}{\mathcal{N}_W}\sum_{x_1=1}^{N_1}\cdots\sum_{x_D=1}^{N_D}{\underbrace{\left(\rho(\overrightarrow{x})\ket{0\cdots0}+\rho u(\overrightarrow{x})\ket{0\cdots1}+\cdots+\rho E(\overrightarrow{x})\ket{1\cdots1}\right)_\mathrm{inter-block}}_{(D+2)^2}}\ket{x_1, x_2, \cdots, x_D}_\mathrm{cell}
    \end{equation}
into
    \begin{equation}\label{aeq:w_spectral}
        \ket{\Tilde{W}} = \frac{1}{\mathcal{N}_{\Tilde{W}}}\sum_{k_1=0}^{N_1-1}\cdots\sum_{k_D=0}^{N_D-1}\underbrace{\left(\Tilde{\rho}(\overrightarrow{k})\ket{0\cdots0}+\Tilde{\rho}*\Tilde{u}(\overrightarrow{k})\ket{0\cdots1}+\cdots+\Tilde{\rho}*\Tilde{E} (\overrightarrow{k})\ket{1\cdots1}\right)}_{(D+2)^2}\ket{k_1, k_2, \cdots, k_D},
    \end{equation}
wherein $\Tilde{\rho}(\overrightarrow{k})$, $\Tilde{u}(\overrightarrow{k})$, $\Tilde{v}(\overrightarrow{k})$, ..., $\Tilde{\rho}(\overrightarrow{E})$ are the $(k_1, k_2, \cdots, k_D)$-th spectral coefficient of the corresponding variables, $*$ is the $D$-dimensional convolution product, and
\begin{equation}
    \mathcal{N}_W = \sqrt{\sum_{x_1=1}^{N_1}\cdots\sum_{x_D=1}^{N_D}\left(\rho^2(\overrightarrow{x})+(\rho u)^2(\overrightarrow{x})+(\rho v)^2(\overrightarrow{x})+\cdots+(\rho E)^2(\overrightarrow{x})\right)}
\end{equation}
and
\begin{equation}
    \mathcal{N}_{\Tilde{W}} = \sqrt{\sum_{k_1=0}^{N_1-1}\cdots\sum_{k_D=0}^{N_D-1}\left(\Tilde{\rho}^2(\overrightarrow{k})+(\Tilde{\rho}*\Tilde{u})^2(\overrightarrow{k})+(\Tilde{\rho}*\Tilde{v})^2(\overrightarrow{k})+\cdots+(\Tilde{\rho}*\Tilde{E})^2(\overrightarrow{k})\right)}
\end{equation}
are the corresponding normalization constants satisfying
\begin{equation}
    \mathcal{N}_W = \mathcal{N}_{\Tilde{W}}
\end{equation}
as a direct result of Parseval's theorem.
Notably, in our quantum Navier-Stokes solver, we assume that the output $\rho, \rho u, \rho v$, and $\rho E$ stay in the spectral space of sparsity $\mathcal{S}$. Consequently, $\rho u$, ..., $\rho v$'s spectral sparsities are all upper bounded by $\mathcal{S}^2$ as a direct result of convolution, and $\rho E=\rho(e+\frac{1}{2}u^2+\frac{1}{2}v^2)$'s spectral sparsity is upper bounded by $\mathcal{S}^3$. Consequently, the $4N$-dimensional dense state $\ket{W}$ is translated into an $(\mathcal{S}^3+2\mathcal{S}^2+\mathcal{S})$-dimensional state $\ket{\Tilde{W}}$.

While the extraction of these $\mathcal{O}(\mathrm{Poly}(\mathcal{S}))$ complex amplitudes has already been much cheaper than the original ones, the underlying symmetry can further reduce the tomography overhead.
Indeed, since the original output $\rho, \rho u, \rho v$, and $\rho E$ are real-valued, their Fourier spectrum satisfy the conjugate symmetry
\begin{equation}
    \Tilde{\rho}(k_1, k_2, \cdots, k_D) = \Tilde{\rho}^*(N_1-k_1, N_2-k_2, \cdots, N_D-k_D),
\end{equation}
wherein $*$ on the exponent denotes the complex conjugate.
Therefore the real and imaginary parts of $\Tilde{\rho}(k_1, k_2, \cdots, k_D)$ can be derived by
\begin{align}\label{aeq:spectral_re}
    \mathrm{Re}\Tilde{\rho}(k_1, k_2, \cdots, k_D) =& \frac{1}{2}\left( \Tilde{\rho}(k_1, k_2, \cdots, k_D) + \Tilde{\rho}(N_1-k_1, N_2-k_2, \cdots, N_D-k_D) \right),\\\label{aeq:spectral_im}
    \mathrm{Im}\Tilde{\rho}(k_1, k_2, \cdots, k_D) =& \frac{1}{2}\left( \Tilde{\rho}(k_1, k_2, \cdots, k_D) - \Tilde{\rho}(N_1-k_1, N_2-k_2, \cdots, N_D-k_D) \right).
\end{align}
By applying a sequence of parallel quantum inplace subtractor operators $O_\mathrm{SUB}$ on the cell index register in Eq.~\eqref{aeq:w_spectral}, we can derive
\begin{equation}
    \begin{split}
        &\ket{\Tilde{W'}} = \left(\bigotimes_{j=1}^D O_\mathrm{SUB}(N_j)\right) \ket{\Tilde{W}}\\
        =& \frac{1}{\mathcal{N}_{\Tilde{W}}}\sum_{k_1=0}^{N_1-1}\cdots\sum_{k_D=0}^{N_D-1}\left(\Tilde{\rho}(\overrightarrow{k})\ket{00}+\Tilde{\rho}*\Tilde{u}(\overrightarrow{k})\ket{01}+\Tilde{\rho}*\Tilde{v} (\overrightarrow{k})\ket{10}+\cdots+\Tilde{\rho}*\Tilde{E} (\overrightarrow{k})\ket{11}\right)\ket{\overrightarrow{N}-\overrightarrow{k}}\\
        =& \frac{1}{\mathcal{N}_{\Tilde{W}}}\sum_{k_1=0}^{N_1-1}\cdots\sum_{k_D=0}^{N_D-1}\left(\Tilde{\rho}(\overrightarrow{N}-\overrightarrow{k})\ket{00}+\Tilde{\rho}*\Tilde{u}(\overrightarrow{N}-\overrightarrow{k})\ket{01}+\Tilde{\rho}*\Tilde{v} (\overrightarrow{N}-\overrightarrow{k})\ket{10}+\cdots+\Tilde{\rho}*\Tilde{E}(\overrightarrow{N}-\overrightarrow{k})\ket{11}\right) \ket{\overrightarrow{k}}\\
    \end{split},
\end{equation}
where $\overrightarrow{N}-\overrightarrow{k} = (N_1-k_1, N_2-k_2, \cdots, N_D-k_D)$ for simplicity. Consequently, the desired real-valued information defined in Eq.~\eqref{aeq:spectral_re} and Eq.~\eqref{aeq:spectral_im} can be derived by a conditioned implementation of $\bigotimes_{j=1}^D O_\mathrm{SUB}(N_j)$ as
\begin{equation}\label{aeq:con_sub}
    \begin{split}
        &U_\mathrm{Con-Sub}\ket{\Tilde{W}}\\
        =&\left(H\otimes I_{n+2\log(D+2)}\right)\left(\ket{0}\bra{0}\otimes\left(\bigotimes_{j=1}^D O_\mathrm{SUB}(N_j)\right)+\ket{1}\bra{1}\otimes I_{n+2\log(D+2)}\right)\left(H\otimes I_{n+2\log(D+2)}\right)\ket{0}\ket{\Tilde{W}}\\
        =&\left(H\otimes I_{n+2\log(D+2)}\right)\frac{1}{\sqrt{2}}\left(\ket{0}\ket{\Tilde{W}} + \ket{1}\ket{\Tilde{W'}}\right)\\
        =&\frac{1}{\sqrt{2}}\left(\ket{0}\ket{\mathrm{Re}\Tilde{W}} + \ket{1}\ket{\mathrm{Im}\Tilde{W}}\right),
    \end{split}
\end{equation}
wherein
\begin{equation}
    \begin{split}\label{aeq:signed_w_state}
        \ket{\mathrm{Re}\Tilde{W}} &= \frac{1}{\sqrt{2}}\ket{\Tilde{W}} + \frac{1}{\sqrt{2}}\ket{\Tilde{W'}},\\
        \ket{\mathrm{Im}\Tilde{W}} &= \frac{1}{\sqrt{2}}\ket{\Tilde{W}} - \frac{1}{\sqrt{2}}\ket{\Tilde{W'}}.
    \end{split}
\end{equation}
To summarize, one suffices to tomography two $(\mathcal{S}^3+2\mathcal{S}^2+\mathcal{S})$-dimensional states defined in Eq.~\eqref{aeq:signed_w_state} to extract $(2\mathcal{S}^3+4\mathcal{S}^2+2\mathcal{S})$ signed real numbers. By Theorem~\ref{thm:tomography}, this can be implemented by $(115S^3 + 230S^2+115S)\ln{((12S^3 + 24S^2 + 12S)/\delta)}/(\epsilon^2(1-\epsilon^2/4)$ calls of our algorithm, where the error $\epsilon$ is bounded in later sections.
Given this low-dimensional state, we use the celebrated compressed-sensing quantum state tomography protocol:
\begin{atheorem}[Compressed-sensing quantum state tomography, Theorem 1 in \cite{gross2010quantum}]\label{thm:tomography}
    Let $\rho$ be an arbitrary state of rank $r$ and effective dimension $d_\mathrm{eff}$. If $m = Crd_\mathrm{eff}\log^2d_\mathrm{eff}$ randomly chosen Pauli expectations are known, then $\rho$ can be uniquely reconstructed  with probability
    of failure to be exponentially small in $c$.

\end{atheorem}

\subsection{End-to-end complexity and resource analysis}
Now, we can describe the end-to-end iterative quantum Navier-Stokes solver given the aforementioned ingredients. 
For the $k$-th iteration, we can update the block-encoding oracles $\mathcal{O}_A^{(k+1)}, \mathcal{O}_b^{(k+1)}$ by $\tilde{W}^{(k)}$.
In Theorem~\ref{theorem: qlss}, the QLSS is implemented in two steps. Firstly, randomly choose a guess of norm $t=\lvert x\rvert$ and apply a repeat-until-success process of \textit{Kernel Reflection} defined as a QSVT of the block-encoding of the operator
\begin{equation}
    G_t = \left(I-\frac{1}{2}(\ket{b}+\ket{e_m})(\bra{b}+\bra{e_m})\right)\left(A+t^{-1}\ket{e_m}\bra{e_n}\right)
\end{equation}
with phase angles determined by the polynomial
\begin{equation}\label{aeq:KR}
    K_{\kappa^{-1}, l}(x) = \frac{2T_l\left( \frac{-2\kappa^2x^2+\kappa^2+1}{\kappa^2-1} \right)+2}{T_l\left( \frac{\kappa^2+1}{\kappa^2-1} \right)+1}-1
\end{equation}
wherein $T_l$ is the $l$-th Chebyshev polynomial and
\begin{equation}\label{aeq:l}
    l = \left\lceil \frac{\mathrm{arccosh(}\eta^{-1}\mathrm{)}}{\mathrm{arccosh}\left(\frac{\kappa^2+1}{\kappa^2-1}\right)} \right\rceil.
\end{equation}
This step successfully outputs an ansatz solution $\ket{a}_\mathrm{ANC}\ket{x_\mathrm{ANS}}_\mathrm{WORK}$ when ancillary qubits are measured to be zero.
Then, apply a repeat-until-success process of \textit{Kernel Projection} defined as a QSVT of the block-encoding of the operator
\begin{equation}
    G = \left(I-\ket{b}\bra{b}\right)A
\end{equation}
with phase angles determined by the polynomial
\begin{equation}\label{aeq:KP}
    F_{\kappa^{-1}, l'}(x) = \frac{T_{l'}\left( \frac{-2\kappa^2x^2+\kappa^2+1}{\kappa^2-1} \right)}{T_{l'}\left( \frac{\kappa^2+1}{\kappa^2-1} \right)}
\end{equation}
where 
\begin{align}\label{aeq:lkp}
    l' =& \left\lceil \frac{\mathrm{arccosh}\left(\eta_\mathrm{KP}^{-1}\right)}{\mathrm{arccosh}\left(\frac{\kappa^2+1}{\kappa^2-1}\right)} \right\rceil,\\\label{aeq:eta}
    \eta_\mathrm{KP}=&\frac{\epsilon}{\sqrt{1-\epsilon^2}}\sqrt{\frac{(1-\eta)^2(3+2\ln\left(\frac{\mathcal{R},^2+\mathcal{L},^2}{2\mathcal{L},^2}\right))}{\eta^2}-1}.
\end{align}
This step successfully outputs a solution $\ket{a}_\mathrm{ANC}\ket{W^{(k+1)}}_\mathrm{WORK}$ when ancillary qubits are measured to be zero.
Then we can apply $F_{N_1, N_2, ..., N_D}^\dagger$ in Eq.~\eqref{aeq:qft} and $U_\mathrm{Con-Sub}$ in Eq.~\eqref{aeq:con_sub} to derive $\ket{\tilde{W}^{(k+1)}}$.
Finally, a random Pauli measurement is applied on $\ket{\tilde{W}^{(k+1)}}$. The state preparation and measurement process is repeated several times to reconstruct $\tilde{W}^{(k+1)}$ for the next iteration.
The end-to-end complexity for this Algorithm~\ref{alg:qnss} can be evaluated as:
\begin{proof}[Proof of Theorem~\ref{thm:2}]

    By Lemma~\ref{lem:hierarchy} and Lemma~\ref{lem:spectral} introduced in subsection~\ref{sec:hsbe}, we derive the block-encoding of flux Jacobian matrix $A$ and the residual vector $b$ with $\mathcal{O}(\mathcal{S}+\log\log N)$ depth and $\mathcal{O}(\mathcal{S}\log N)$ non-Clifford gate count in subsection~\ref{sec:Abbe}. By Theorem~\ref{theorem: qlss}, $\Tilde{\mathcal{O}}(\kappa)$ queries of $\mathcal{O}_A$ and $\mathcal{O}_b$ are required. Then we can apply the parallel inverse quantum Fourier transformation with an additional $\mathcal{O}(\log^2N)$ depth. To summarize, the end-to-end circuit depth is $\Tilde{\mathcal{O}}(\kappa(\mathcal{S}+\log\log N) + \log^2N)$. To also consider the sampling number given in subsection~\ref{sec:ssd}, the end-to-end time complexity is $\Tilde{\mathcal{O}}(\mathcal{S}\epsilon^{-2}(\kappa(\mathcal{S}+\log\log N) + \log^2N))$.
\end{proof}
The end-to-end quantum resource with conventional synthesis is sketched in Algorithm~\ref{alg:bea}, Algorithm~\ref{alg:beb}, and Algorithm~\ref{alg:qnss}.
We also develop a program to numerically compute the resources.
\begin{algorithm}[H]
        \caption{End-to-End Iterative Quantum Navier-Stokes Solver}
        \label{alg:qnss}
        \begin{algorithmic}[1]
            \Require $\tau$, $\kappa$, $\epsilon$, $\eta$, $\mathcal{L}$, $\mathcal{R}$, $\tilde{W}^{(0)}$
            \Ensure $\tilde{W}^{(\tau)}$
            \State Classically compute $l$ in Eq.~\eqref{aeq:l}, $l'$ in Eq.~\eqref{aeq:lkp},  $\eta_\mathrm{KP}$ in Eq.~\eqref{aeq:eta};
            \State Classically compute phase angles $\phi_\mathrm{KR}$ for Eq.~\eqref{aeq:KR} and $\phi_\mathrm{KP}$ for Eq.~\eqref{aeq:KP};
            \For{$k\in[\tau]$}
                \State Update $\mathcal{O}_A^{(k+1)}, \mathcal{O}_b^{(k+1)}$ by $\tilde{W}^{(k)}$;
                \For{$j\in[\mathcal{N}_\mathrm{sample}]$}
                    \State \textcolor{myred}{/* Apply QLSS */}
                    \Repeat
                        \Repeat
                            \State Choose uniformly random $t\in[\ln\mathcal{L}-\frac{1}{2}, \ln\mathcal{R}+\frac{1}{2}]$;
                            \State Apply $K_{\kappa^{-1}, l}(G_t)$) on $\ket{0}$ with output state $\ket{a}_\mathrm{ANC}\ket{x_\mathrm{ANS}}_\mathrm{WORK}$;
                            \State Measure on $\ket{a}_\mathrm{ANC}$;
                        \Until{$a=0$}
                        \State Apply $F_{\kappa^{-1}, l'}(G)$) on $\ket{x_\mathrm{ANS}}_\mathrm{WORK}$ with output state $\ket{a'}_\mathrm{ANC}\ket{W^{(k+1)}}_\mathrm{WORK}$;
                        \State Measure on $\ket{a'}_\mathrm{ANC}$;
                    \Until{$a'=0$}
                    \State \textcolor{myred}{/* Apply $\mathcal{D}$ */}
                    \State Apply $F_{N_1, N_2, ..., N_D}^\dagger$ in Eq.~\eqref{aeq:qft};
                    \State Apply $U_\mathrm{Con-Sub}$ in Eq.~\eqref{aeq:con_sub};
                    \State \textcolor{myred}{/* Random Pauli measurement */}
                    \State Apply $P_j$ and measure;
                \EndFor
                \State Reconstruct $\Tilde{W}^{(k+1)}$;
            \EndFor
            \State \Return  $\tilde{W}^{(\tau)}$
        \end{algorithmic}
    \end{algorithm}

\section{Circuit Synthesis and Logical Quantum Resource Reduction}

While spectral-based block-encoding and state tomography protocols mitigate the I/O bottleneck by reducing asymptotic complexity dependence on problem size, practical quantum advantage in CFD may still be undermined by inefficient circuit implementation. To address this, we develop circuit synthesis methods to minimize logical resource overhead. First, we outline conventional synthesis techniques employed in this work to reduce logical resources within each circuit module. Next, we introduce match-and-merge, an algorithm-inspired synthesis method that reduces logical resources by leveraging recurrent circuit modules. Finally, we propose mask-and-merge, a process that further optimizes resource usage by adaptively merging similar circuit substructures via a dynamically generated mask.
    \begin{figure}
        \centering
        \begin{tikzpicture}
            \node[]  at (-5, 0) {\includegraphics[width=0.46\textwidth]{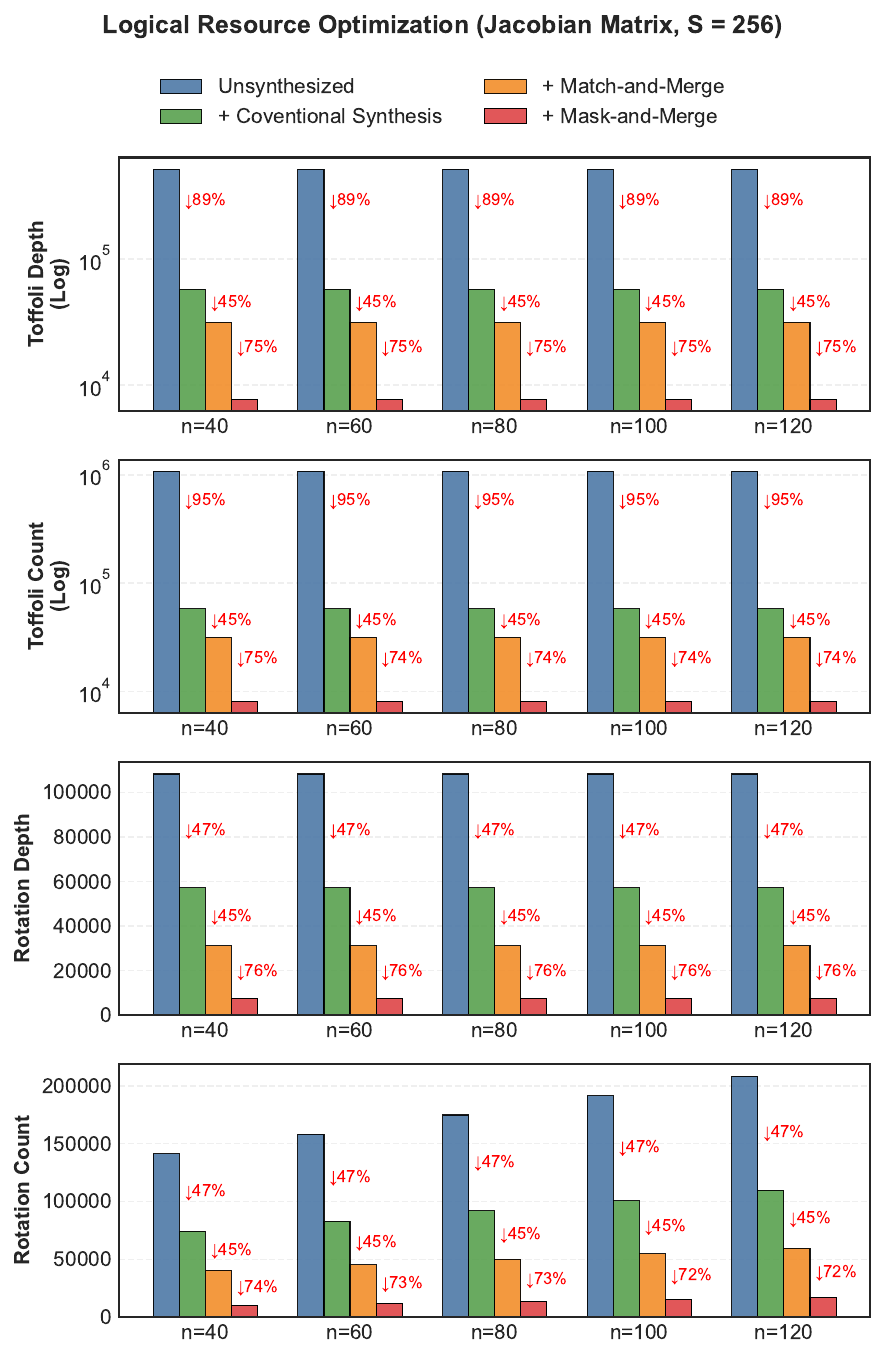}};
            \node[] at (4, 0) {\includegraphics[width=0.46\textwidth]{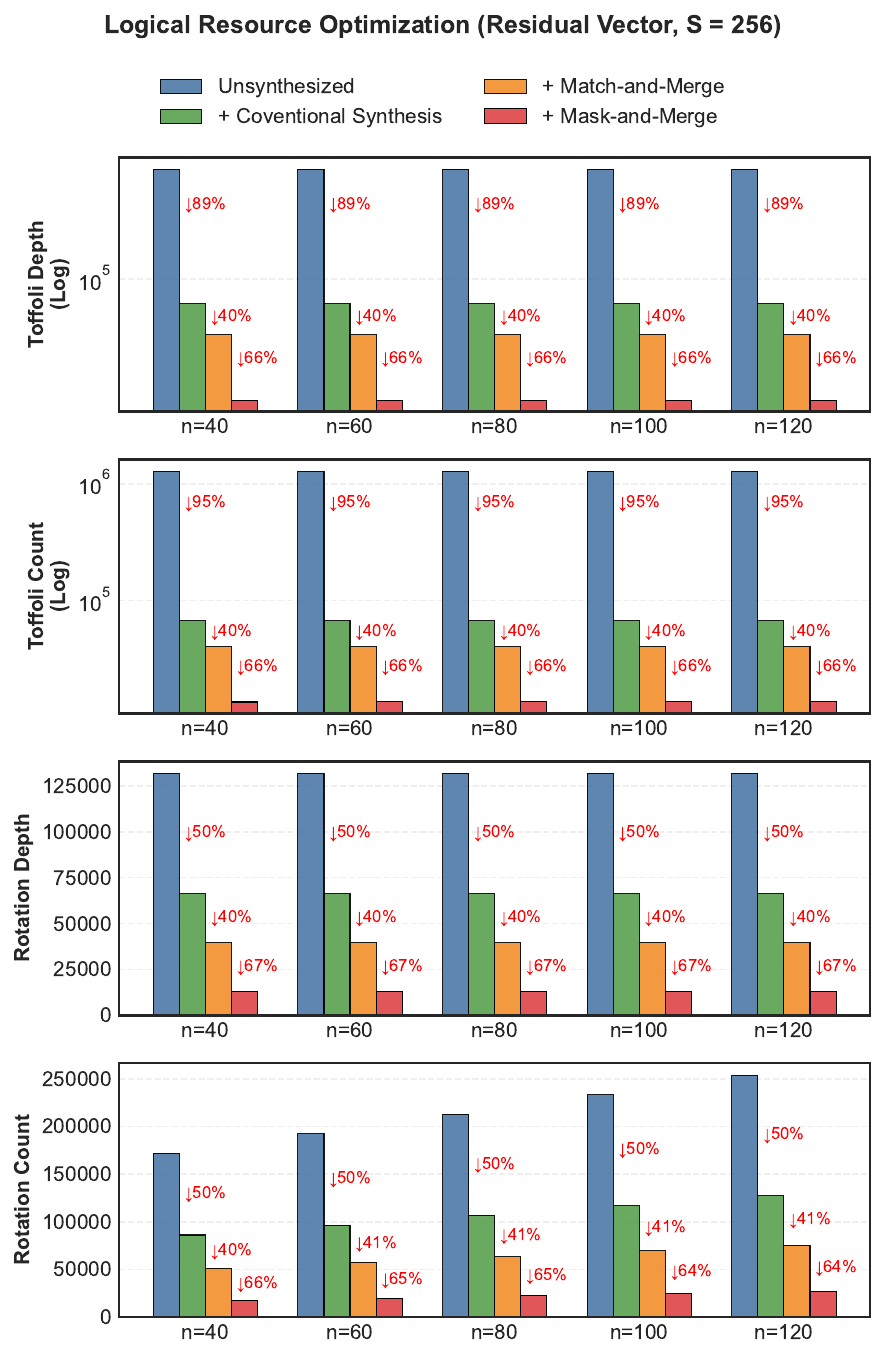}};
            \node[] at (-8.5, 6.1) {\textbf{a}};
            \node[] at (0.5, 6.1) {\textbf{b}};
        \end{tikzpicture}
        \caption{Logical resource optimization by circuit synthesis (a) Jacobian matrix: Toffoli depth is reduced by $98.5\%$, Toffoli count is reduced by $99.3\%$, Rotation depth and count are both reduced by approximately $92\%-93\%$. (b) Residual vector: Toffoli depth is reduced by $98\%$, Toffoli count is reduced by $99\%$, Rotation depth and count are both reduced by $90\%$.}
        \label{fig:synthesis_A}
    \end{figure}

    \begin{figure}
        \centering
        \begin{tikzpicture}
            \node{\includegraphics[width=0.8\textwidth]{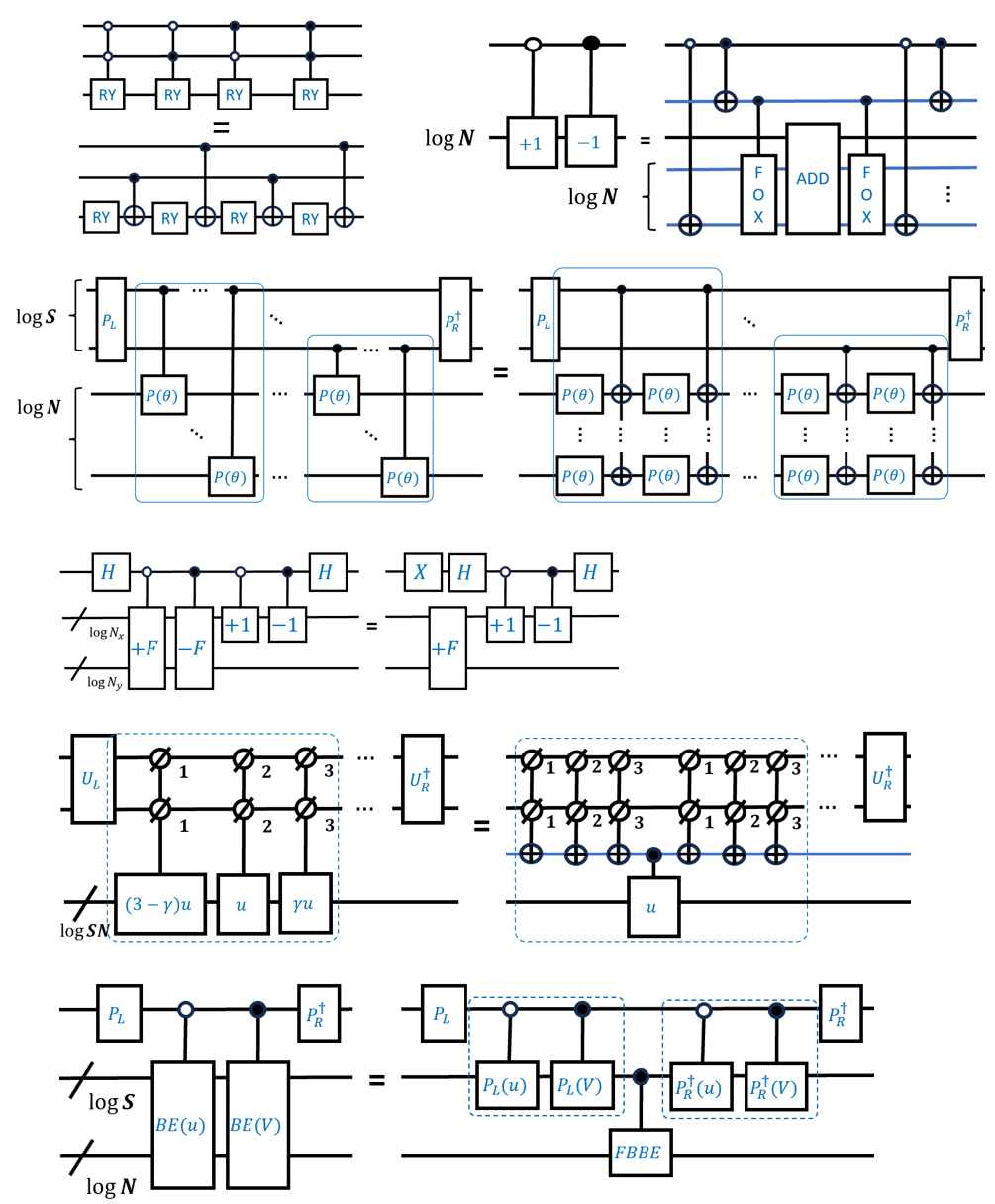}};
            \node[] at (-7, 8) {\textbf{(a1)}};
            \node[] at (-1, 8) {\textbf{(a3)}};
            \node[] at (-7, 4.75) {\textbf{(a2)}};
            \node[] at (-7, 0.5) {\textbf{(b1)}};
            \node[] at (-7, -2.25) {\textbf{(b2)}};
            \node[] at (-7, -5.75) {\textbf{(b3)}};
        \end{tikzpicture}
        \caption{Circuit synthesis: conventional techniques    match-and-merge patterns. (a1) Uniformly controlled rotation synthesis in the state preparation subroutine. (a2) Fan out $X$ gate in the Fourier basis block-encoding subroutine. 
        (a3) Controlled quantum  arithmetic in the shift operator.
        (b1) Match by signal pattern.
        (b2) Match by similar group pattern.
        (b3) Match by spectral structure pattern.}
        \label{fig:synthesis}
    \end{figure}
\subsection{Conventional Synthesis}
The hierarchy spectral quantum Navier-Stokes solver heavily utilizes state preparation modules to load spectrum information, Fourier basis block-encoding modules to load basis function information, and shift operators to compute and allocate matrix structure information. Henceforth, we first consider conventional quantum circuit synthesis methods to reduce logical resources within each subroutine as follows. 
\subparagraph{State preparation subroutine.} In each state preparation subroutine, $(\mathcal{S}-1)$ independent complex numbers (after normalization) are encoded in the amplitudes of a $\log\mathcal{S}$-qubit state. Conventionally, this state is prepared by $(\mathcal{S}-1)$ (multi-controlled) rotation gates: one rotation gate is implemented on the first qubit, and then $2^k$ rotation gates controlled by the first $k$ qubits are conducted on the $k+1$-th qubit for $1\leq \log\mathcal{S}-1$. 
While it is easy to check that at least $(\mathcal{S}-1)$ rotation gates are required to load $(\mathcal{S}-1)$ independent classical information, those costly (multi-controlled) Toffoli gates in their multi-control version can indeed be removed by the celebrated \textit{Uniformly Controlled Rotation} technique~\cite{mottonen2004transformation}.
Herein, those $2^k$ rotation gates on the same target qubit and controlled by strings from $\ket{00\cdots0}$ to $\ket{11\cdots1}$ can be transformed into $2^k$ rotation gates and $2^k$ $CX$ gates, as depicted in Fig.~\ref{fig:synthesis}~(a1).
In this way, $(\mathcal{S}-1)$ rotation depth and $(2\mathcal{S}-8)$ Toffoli depth can be reduced.
\subparagraph{Basis block-encoding subroutine.} In our FBBE subroutine, $n\log\mathcal{S}$ controlled rotation gates are conducted to block encode the $\log\mathcal{S}$ basis function on the $n$-qubit cell index register, where each controlled rotation gate can be decomposed into two single-qubit rotation gates and two $CX$ gates. Here, the key observation is that those $CX$ gates with a common control qubit can be regarded as a \textit{Fan Out X} ($FOX$) gate, and can be implemented by lattice surgery technique in surface code with low overhead~\cite{fowler2018low}. Consequently, the FBBE is transformed into $2\log\mathcal{S}$ layers of parallel rotation gates and $2\log\mathcal{S}$ $FOX$ gates, as illustrated in Fig.~\ref{fig:synthesis}~(a2), removing the circuit depth dependence on the problem size $2^n$.
\subparagraph{Shift subroutine.} In the shift subroutine, the control version of two quantum arithmetic operators on $n_x$ qubits and of two quantum arithmetic operators on $n_y$ qubits are required to allocate the horizontal and vertical convection flux sub-matrices $\frac{\partial \mathcal{F}_C}{\partial W}$ and $\frac{\partial \mathcal{G}_C}{\partial W}$, respectively. We apply an \textit{optimal Toffoli-depth quantum adder} to realize a $(2\log n+3)$-TOF depth quantum adder on the cell index register and an $n$-qubit ancillary register~\cite{wang2024optimal}. Since the control version of the quantum adder can consume much more Toffoli depth, we instead inject an auxiliary state on the ancillary register to further reduce the overhead, as in Fig.~\ref{fig:synthesis}~(a3). More specifically, we use a $CX$ gate to prepare $\ket{00\cdot1}$ on the ancillary register for $Shift(+1)$ and another $CX$ gate following the quantum adder to uncompute this register. Similarly, we use a $FOX$ gate to prepare $\ket{11\cdot1}$ on the ancillary register for $Shift(-1)$ and another $FOX$ to uncompute.

We design a program to quantify the logical resource reduction. And our numerical tests suggest that more than $10\times$ Toffoli depth reduction and approximate $2\times$ rotation depth reduction are realized for varying problem size $N$ and spectral sparsity $\mathcal{S}$, and the numerical result for $\mathcal{S}=256$ is shown in Fig.~\ref{fig:synthesis_A}.

    \begin{figure}
        \centering
        \begin{tikzpicture}
            \node at (0, 10.5) {\includegraphics[width=\textwidth]{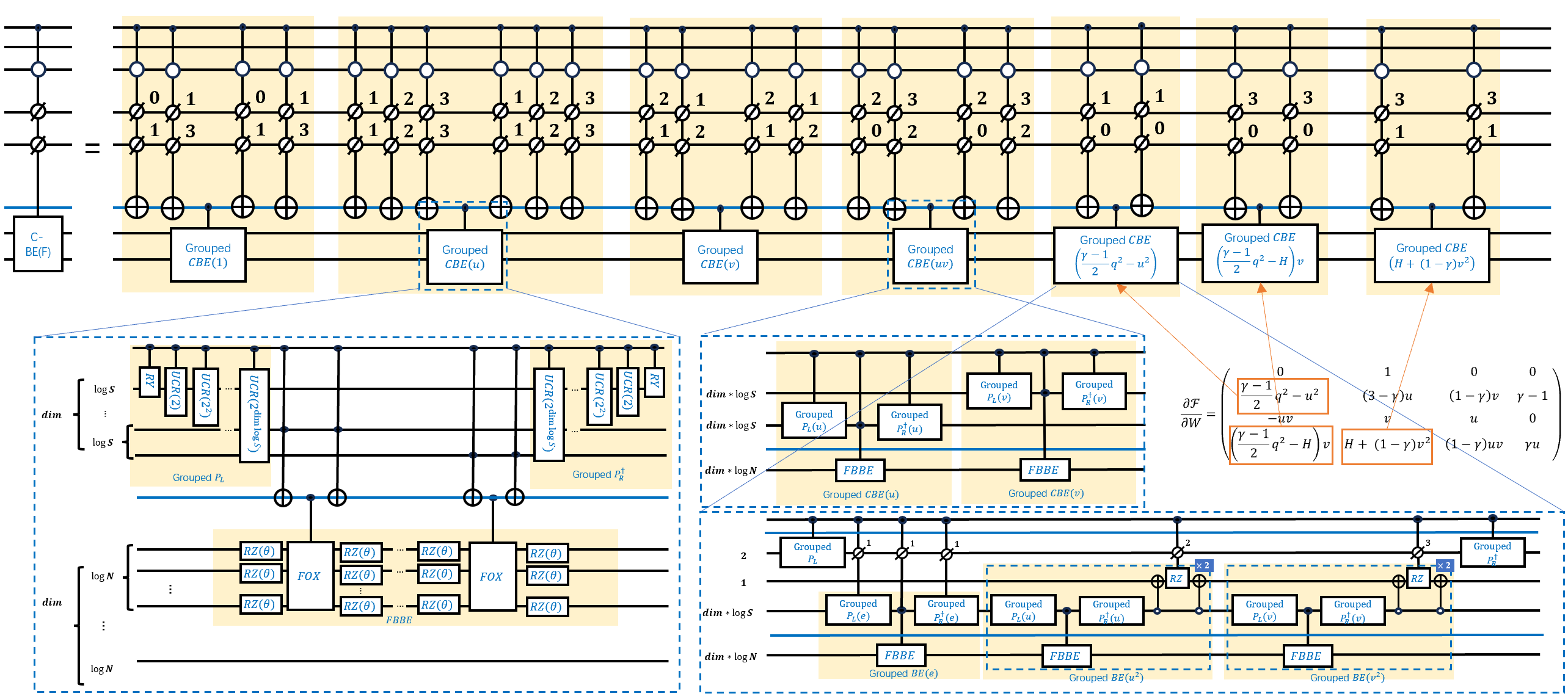}};
        \end{tikzpicture}
        \caption{Synthesized circuit after conventional techniques and the match-and-merge method.}
        \label{fig:synthesized_a}
    \end{figure}

\subsection{Match and Merge}
Besides conventional algorithm-agnostic synthesis techniques, we further reduce the logical resource by utilizing algorithm-aware methods to match redundant subroutines and merge them together. Herein, we study three matching patterns appearing in the hierarchy spectral quantum Navier-Stokes solver for search and merge.
\subparagraph{Match by signal.} We first observe that the subroutines to block encode convective flux Jacobian matrices $\pm\frac{\partial \mathcal{F}_C}{\partial W}$ and $\pm\frac{\partial \mathcal{G}_C}{\partial W}$ given in Eqs.~\eqref{aeq:shift_p1} and \eqref{aeq:shift_p1} share an identical circuit structure, differing only in their signals and locations. 
A similar pattern occurs when block-encoding different flux components at the faces in the residual vector $b$.
By matching these subroutines based on their signals, we can reuse the majority of the circuit structure, achieving an approximate $2\times$ reduction on the logical resources, as illustrated in Fig.~\ref{fig:synthesis}~(b1). 

\subparagraph{Match by similar group.} Our second fundamental observation is that those variables in Eqs.~\eqref{appendix_eq:horizontal_jacobian} and \eqref{appendix_eq:vertical_jacobian} can be divided into several similar groups, where those variables within the same group are proportional to each other. For example, the variables $(3-\gamma)u$ indexed by $(1, 1)$, $u$ indexed by $(2, 2)$, and $\gamma u$ indexed by $(3, 3)$ can be block-encoded by a single spectral block-encoding subroutine $BE(u)$ controlled by an ancillary group flag qubit, as well as three paris of multi-controlled Toffoli gates acting on the original inter-block index register, as depicted in Fig.~\ref{fig:synthesis}~(b2).
Herein, the coefficients $3-\gamma, 1, \gamma$ can be absorbed into the amplitudes and then encoded by the adapted state preparation pair $(P_L, P_R^\dagger)$ on the inter-block index register.
In this way, we can further reduce the number of calls for block-encoding subroutines of $u, v, uv$, etc. in the $A$-block-encoding circuit.

\subparagraph{Match by spectral structure.} We can go one step further, observing that the circuits of block-encoding subroutines with similar spectral structure are indeed the same except for those rotation angles in the state preparation pair. For example, $u$ and $v$ are both of spectral sparsity $\mathcal{S}$ with no product or QSVT architectures.
When implementing a linear combination of such subroutines, we can interchange state preparation and FBBE subroutines with different control strings to reduce circuit depth, as illustrated in Fig.~\ref{fig:synthesis}~(b3).
Herein, the controlled and zero-controlled FBBE subroutines are merged to remove not only a redundant FBBE module but also many costly multi-controlled Toffoli gates.
Also, the controlled and zero-controlled state preparation subroutines with different rotation angles can be merged by applying the \textit{Uniformly Controlled Rotation} technique to transform $\mathcal{O}(\mathcal{S})$ Toffoli gates into cheap $CX$ gates.

We design a process searching for the above three patterns to merge. The synthesized circuit is depicted in Fig.~\ref{fig:synthesized_a}. Our numerical tests suggest that approximately $2\times$ Toffoli and rotation depth reduction are realized for varying problem size $N$ and spectral sparsity $\mathcal{S}$, and the numerical result for $\mathcal{S}=256$ is shown in Fig.~\ref{fig:synthesis_A}.

\subsection{Mask and Merge}

Inspired by the match-and-merge technique to synthesize subroutines with similar patterns, we design an even more efficient synthesis process to merge subroutines with very different patterns, further compressing the circuit. The basic idea is to utilize a template block-encoding circuit for the most complicated spectral structure and then adaptively mask subcircuits within it to block encode variables of different spectral structures. Since the mask stage is implemented through a rotation angle configuration on an identical template circuit, we can merge these masked circuits to reduce the logical resources. More precisely, we have the following two-stage mask-and-merge process:

\subparagraph{Mask stage.} 
Given a template circuit as a linear combination of many subroutines, we define three types of masks. 
The first one is the \textit{turn-on} status, which keeps one subroutine in its current state, and consequently, this subroutine will be successfully executed as in the template circuit.
The second one is the \textit{turn-off} status, which can shut down one subroutine by changing the rotation angles in the state preparation pair so that the underlying subroutine will not be executed.
The third one is the \textit{partial-mask} status, which modifies the rotation angles inside the underlying subroutine to load a different spectrum.
For example, we consider the block-encoding circuit of the horizontal flux circuit in $b$. Among these four flux components in the bottom of Fig.~\ref{fig:algorithm}(a), we choose the most complicated term $\rho uH + q_x - \mu\left[ \left(\frac{4}{3}\frac{\partial u}{\partial x} - \frac{2}{3}\frac{\partial v}{\partial y}\right)H + \left(\frac{\partial u}{\partial x} +\frac{\partial v}{\partial y}\right)v \right]$ as our template circuit, as depicted in Fig.~\ref{fig:mask}~(a).
To block encode the first term $\rho u$, we just need to turn off the subroutines $BE(H), BE(q_x), BE(\mu), BE(\frac{4}{3}\frac{\partial u}{\partial x} - \frac{2}{3}\frac{\partial v}{\partial y}), BE(v)$, and $BE(\frac{\partial u}{\partial x} +\frac{\partial v}{\partial y})$. This mask can be implemented by setting the rotation angles on the first two qubits and in the first state preparation pair $(P_L(H), P_R^\dagger(H))$ to zero. 
To block encode the second term $\rho u^2 + p - \mu\left(\frac{4}{3}\frac{\partial u}{\partial x} - \frac{2}{3}\frac{\partial v}{\partial y}\right)$, we observe that $\rho u^2 + p = \rho\left[ (\gamma-1)e + \frac{3-\gamma}{2}u^2 + \frac{1-\gamma}{2}v^2 \right]$ is indeed the product of $\rho$ and a linear combination of $e, u^2$, and $v^2$. Consequently, we can modify those rotation angles in the first state preparation pair $(P_L(H), P_R^\dagger(H))$ to implement a partial mask on the first $BE(H)$, as well as turn off $BE(u)$, $BE(q_x)$, the second $BE(H)$, $BE(\frac{4}{3}\frac{\partial u}{\partial x} - \frac{2}{3}\frac{\partial v}{\partial y})$, $BE(v)$, and $BE(\frac{\partial u}{\partial x} +\frac{\partial v}{\partial y})$.
To block encode the third term $\rho uv -\mu(\frac{\partial u}{\partial x} +\frac{\partial v}{\partial y})$, we can partial-mask $BE(H)$, turn on $BE(\rho)$, $BE(u)$, and $BE(\frac{\partial u}{\partial x} +\frac{\partial v}{\partial y})$, and turn off the other subroutines.
To summarize, we have derived the block-encoding circuit for all four of these components as a single template circuit with four masks, as depicted in Fig.~\ref{fig:mask}~(b). Herein, the grey colored blocks are turned off, the grey hatched blocks are partially masked, and the remained blocks are turned on.

\subparagraph{Merge stage.}
We emphasize that all of these masks are indeed an alternative configuration of rotation angles, leaving the template circuit's structure unchanged. Consequently, we can merge four (masked) copies of the controlled template circuit, as depicted in Fig.~\ref{fig:mask}~(c).
For those subroutines that are partially masked, such as the first $BE(H)$ block in Fig.~\ref{fig:mask}~(d1) and the $BE(\frac{\partial u}{\partial x} +\frac{\partial v}{\partial y})$ block in Fig.~\ref{fig:mask}~(d4), the modified rotation angles are implemented in a control version and the unchanged part's control architecture is removed.
As for those subroutines that are only turned on or off with no modification, such as theose $BE(\rho)$, $BE(u)$, and $BE(q)$ blocks in Fig.~\ref{fig:mask}~(d2) and those $BE(H)$,  $BE(\frac{4}{3}\frac{\partial u}{\partial x} -\frac{2}{3}\frac{\partial v}{\partial y})$, and $BE(v)$ blocks in Fig.~\ref{fig:mask}~(d3), the control architecture can be directly removed.
In summary, we derive the synthesized circuit from the four copies of masked circuits with a roughly $4\times$ depth reduction, as in Fig.~\ref{fig:mask}~(e).

We design a process searching for the above three patterns to merge. Our numerical tests suggest that approximately $3\times-4\times$ Toffoli and rotation depth reduction are realized for varying problem size $N$ and spectral sparsity $\mathcal{S}$, aand the numerical result for $\mathcal{S}=256$ is shown in Fig.~\ref{fig:synthesis_A}. In summary, we reduce approximately $98\%$ of Toffoli resources and more than $90\%$ rotation resources.
    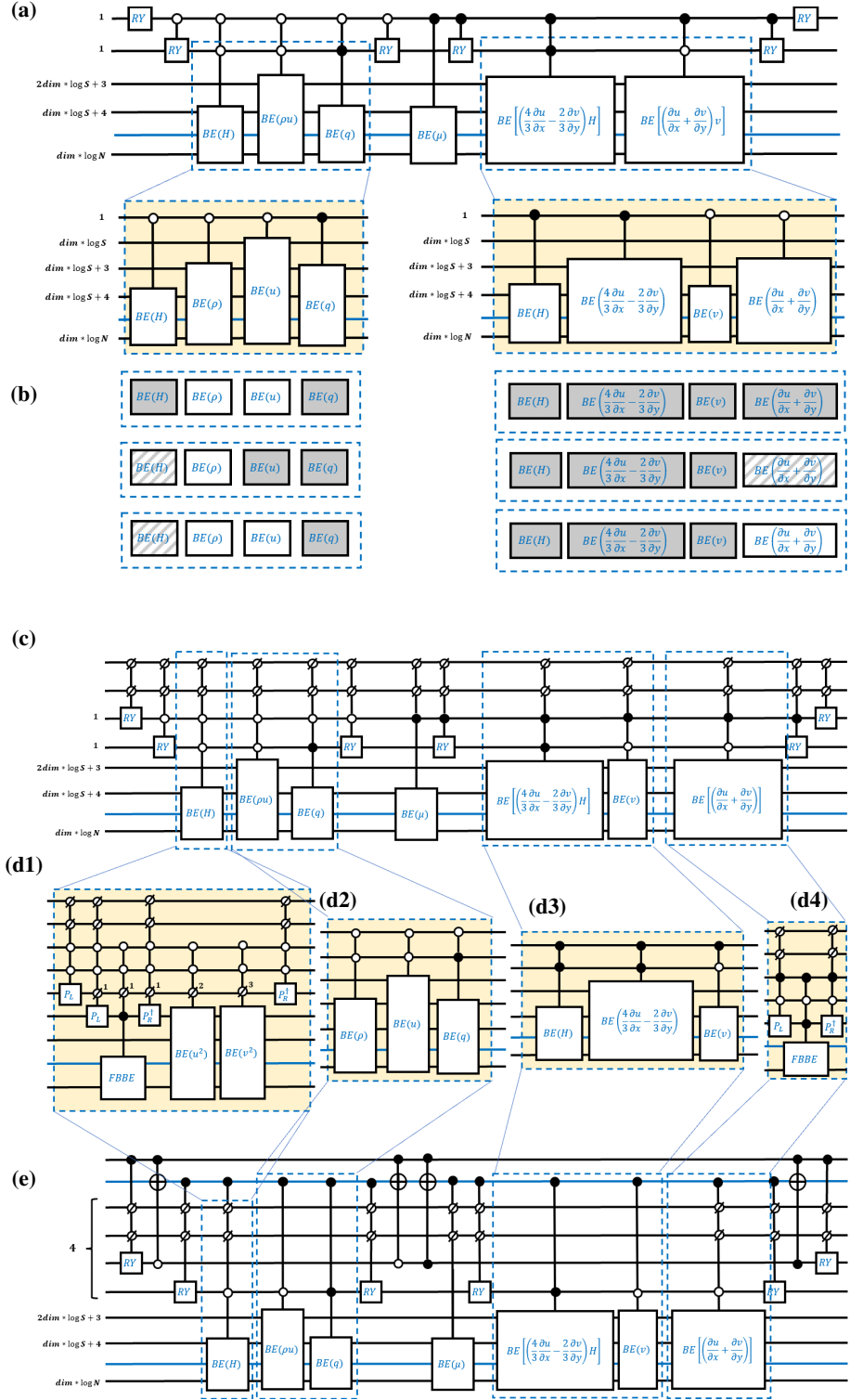
\begin{figure}
        \centering
        \begin{tikzpicture}
            \node at (0, 10.5) {\includegraphics[width=0.65\textwidth]{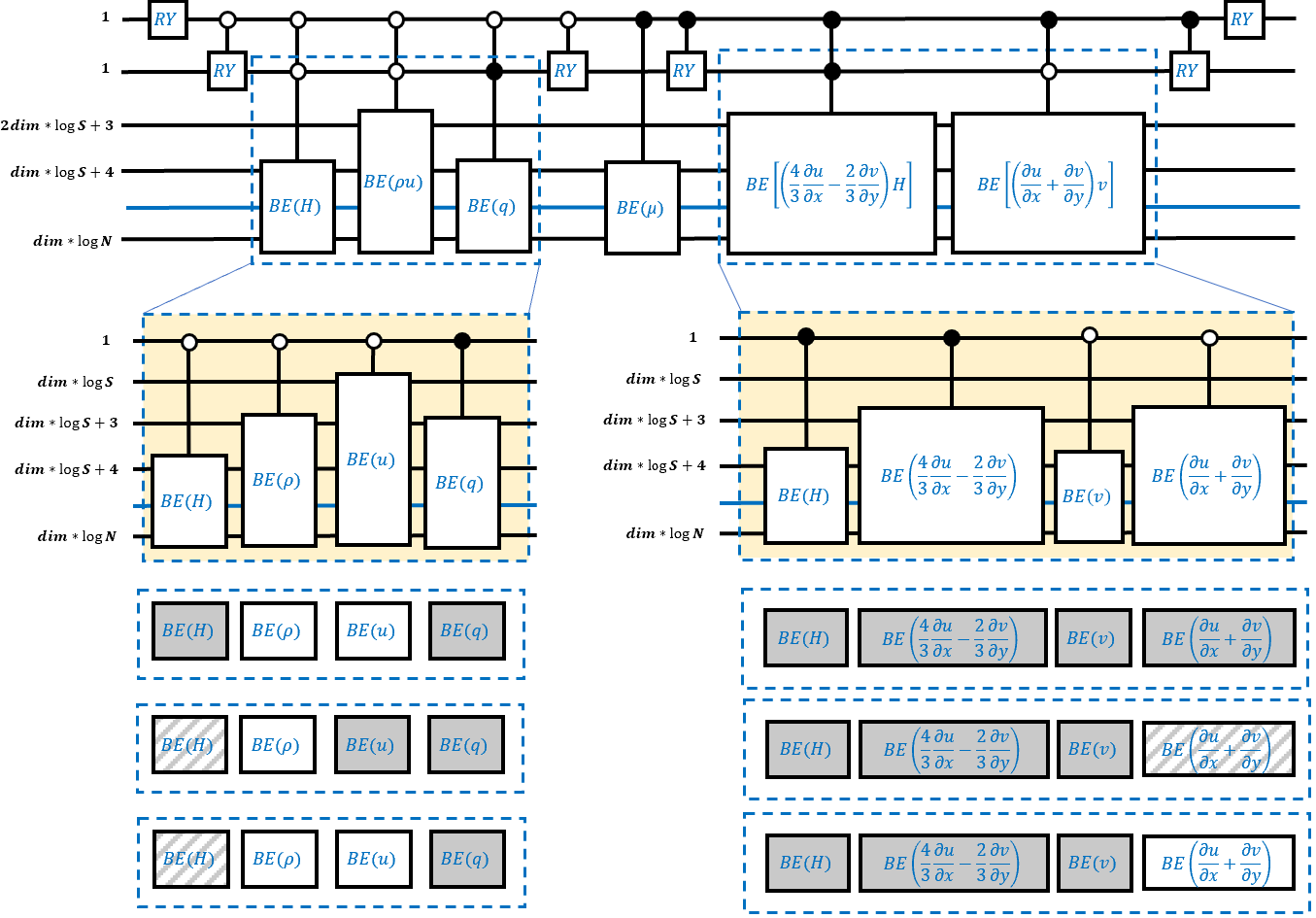}};
            \node{\includegraphics[width=0.65\textwidth]{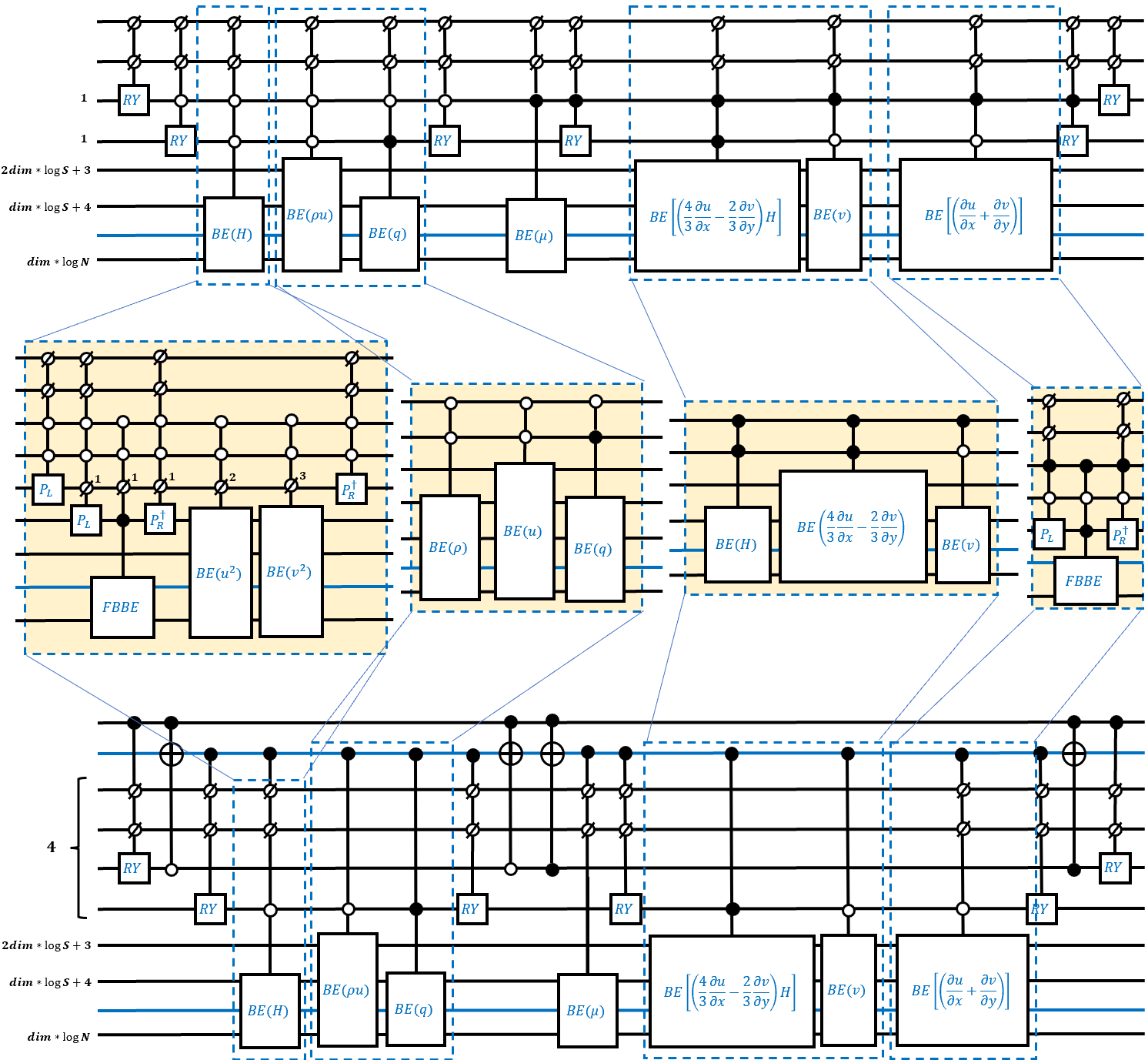}};
            \node[] at (-6, 14.5) {\textbf{(a)}};
            \node[] at (-6, 9) {\textbf{(b)}};
            \node[] at (-6, 5.5) {\textbf{(c)}};
            \node[] at (-6, -2.25) {\textbf{(e)}};
            \node[] at (-6, 2.25) {\textbf{(d1)}};
            \node[] at (-1.5, 1.75) {\textbf{(d2)}};
            \node[] at (1.5, 1.65) {\textbf{(d3)}};
            \node[] at (5.25, 1.75) {\textbf{(d4)}};
        \end{tikzpicture}
        \caption{Mask-and-merge circuit synthesis: 
        (a) Template circuit to block encode the most complicated term $\rho u^2 + p - \mu\left(\frac{4}{3}\frac{\partial u}{\partial x} - \frac{2}{3}\frac{\partial v}{\partial y}\right)$ by hierarchy spectral block-encoding. 
        (b) Masks to block encode other components: white colored blocks for subcircuits turned on, grey colored blocks for subcircuits masked, and grey hatched blocks for subcircuits partially masked. 
        (c) Masked template circuits to be merged, the double-controlled subcircuits can be further synthesized. 
        (d1-d4) Merged (masked) subcircuits: (d1) and (d3) are partially masked, and the rotation angles that need to be modified are double-controlled by the first two qubits. (d2) and (d4) are not partially masked and can be turned on or off by the state-preparation-pair at the beginning and end of the circuit. Consequently, the control qubits can be directly removed.
        (e) The masked and merged circuit with a roughly $4\times$ depth reduction.}
        \label{fig:mask}
    \end{figure}

\section{Quantum Error Correction and Physical Quantum Resource Reduction}

\subsection{Surface Code Implementation}

To enable the large-scale quantum computation, quantum error correction (QEC) is the crucial technique that protects information against the noise by encoding physical qubits into logical qubits of a certain type of QEC code. Among these QEC codes, surface codes are well-studied and generally considered to be one of the leading candidates for practical quantum computing~\cite{fowler2012surface}.

Here we focus the rotated surface code, which employs $d^2$ data qubits and $d^2-1$ syndrome qubits to encode a logical qubits, where $d$ is the code distance. Thus, the total number of physical qubits for each logical qubit is 
\begin{equation}
    n_\mathrm{phy} = 2d^2 - 1
\end{equation}

For each logical qubit of surface code qubit, the logical error rate per QEC cycle $p_\mathrm{L}$ can be estimated by the celebrated Fowler-Devitt-Jones formula
\begin{equation}
    P_\mathrm{L} = c{\left( 
\frac{p_\mathrm{phy}}{p_\mathrm{th}} \right)}^{(d+1)/{2}},
\end{equation}
where $p_\mathrm{phy}$ is the physical error rate, $p_\mathrm{th} \simeq 0.01$ is the threshold, and $c = 0.1$ is a constant, under the circuit-level noise model.~\cite{fowler2012surface, moussa2016transversal, litinski2019magic}.

Using lattice surgery~\cite{fowler2018low}, Clifford operations are relatively easy to implement, such as $H$~\cite{geher2024error} and $CX$~\cite{fowler2018low}. Another commonly used example is the multi-target controlled-not gate $CX_n$, which can be implemented in $2d$ QEC cycles using one ancilla logical qubit~\cite{fowler2018low}.

Non-Clifford gates on the surface code require significantly more resource overhead. Implementing these logic gates typically involves preparing corresponding ancilla states and applying gate teleportation circuits. The non-Clifford gates discussed in this work include the Toffoli gate and single-qubit rotation gates. The gate teleportation circuits for their implementation can be found in Ref.~\cite{gidney2019flexible} and Ref.~\cite{litinski2019game}, respectively.

\subsection{Magic State Preparationy}
The ancilla states required for implementing Toffoli gates and rotation gates are also known as magic states. To prepare these magic states, we consider the following preparation protocols.
\begin{alemma}
    \textbf{(T States from Magic State Cultivation,~\cite{gidney2024magic}.)} The state
    $\ket{T}=(\ket{0}+e^{i\pi/4}\ket{1})/\sqrt{2}$ can be prepared with volume $v=10^5$ to reach the logical error rate $\epsilon_L = 2\times 10^{-9}$ given physical error rate $p_\mathrm{phys}=10^{-3}$, and to reach the logical error rate $\epsilon_L = 4\times 10^{-11}$ given physical error rate $p_\mathrm{phys}=5\times10^{-4}$.
\end{alemma}

\begin{alemma}
    \textbf{(Toffoli states from $8\ket{T}\xrightarrow{28\epsilon^2}\ket{CCZ}$ distillation, \cite{gidney2019efficient}.)}
    The $\ket{CCZ}$ can be prepared in $5.5d$ rounds with output error rate bounded by $\epsilon_L \simeq 28\epsilon^2$ given 8 $\ket{T}$ states with input error rate $\epsilon$.
\end{alemma}

\begin{alemma}
    \textbf{(T states from catalyzed factory, \cite{gidney2019efficient}.)}
    Two $\ket{T}$ states can be converted from one $\ket{CCZ}$ state in $d$ rounds, along with one catalytic T-state, yielding an output error rate comparable to that of the input $\ket{CCZ}$ state.
\end{alemma}

\begin{alemma}
    \textbf{(Rotation states from MEK distillation, \cite{mlti, campbell2016efficient}.)}
    Rotation states $\ket{\theta_l}=e^{i\pi Z/2^{l}}\ket{+})$ in low Clifford hierarchy level $l$ can be prepared by iteratively executing MEK distillation protocol.
\end{alemma}

\begin{alemma}
    \textbf{(Rotation states from multi-level transversal injection, \cite{mlti}.)}
    Rotation states $\ket{\theta_l}$ in high Clifford hierarchy level $l$ can be prepared by multi-level transversal injection protocol.
\end{alemma}

In our design, cultivation protocols are employed to generate T-states with an error rate higher than $10^{-10}$. These T states serve as input states for distillation protocols to produce T states and CCZ states of sufficiently high fidelity. Furthermore, the T-states generated through both cultivation and distillation are consumed in the T-pumping stages of the multi-level transversal injection protocol and the MEK distillation protocol for preparing high-fidelity rotation states.

Using the above protocols, Toffoli states with a target infidelity below $2.8\times 10^{-17}$ and rotation states with a target infidelity below $5\times 10^{-12}$ for $l<30$ can be prepared. When implementing rotation gates with rotation angle precision up to $l=30$ using the circuit from Ref.~\cite{litinski2019game}, it requires all rotation states with $l<30$. Here, we use the spacetime volume, defined as the number of qubits multiplied by the number of QEC cycles, as the metric for overhead. In summary, the overhead of magic states required for each Toffoli gate or rotation gate is as follows:

\begin{equation}\label{aeq:volume_rot}
    \begin{split}
        V_\mathrm{TOF} &= 2.29\times10^6,\\
        V_\mathrm{ROT} &= \frac{1}{27}\sum_{3\leq L<30}\sum_{3\leq l\leq L} V(\ket{{\theta_l}})= 7.62\times10^7.
    \end{split}
\end{equation}

\section{Characterizing Quantum Advantage}
\subsection{Error and Resource Model}
\begin{figure}
    \centering
    \includegraphics[width=0.95\textwidth]{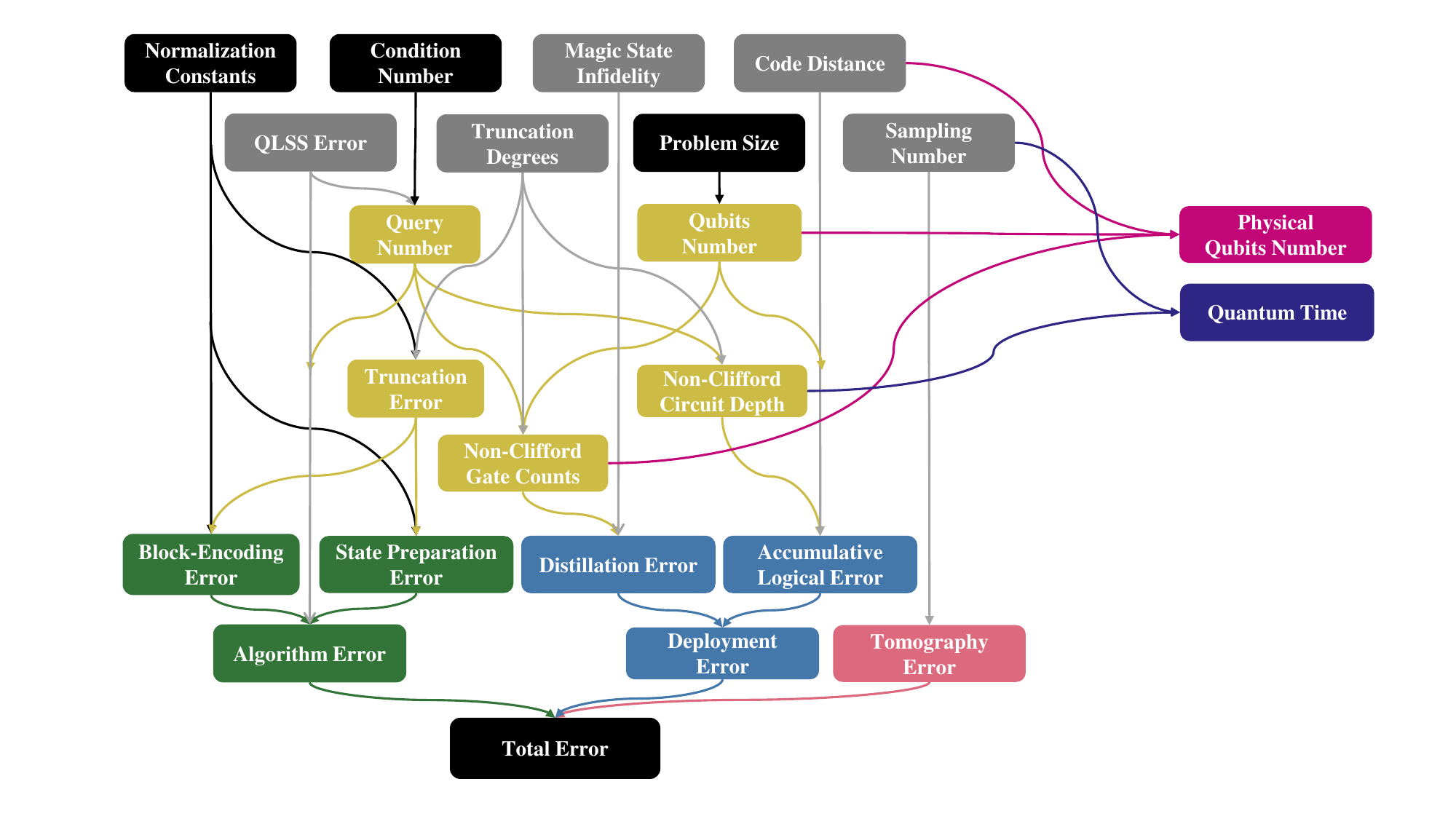}
    \caption{End-to-End Error Analysis Model.}
    \label{fig:error_model}
\end{figure}
In our error model shown in Fig.~\ref{fig:error_model}, there are two types of hyperparameters: the \textit{problem-relevant hyperparameters}  to be characterized by numerical tests (black colored), such as problem size, block-encoding normalization constants, and linear system condition number upper bound $\Bar{\kappa}$, and the algorithm-dependent \textit{free hyperparameters} to be decided by optimization (grey colored), including truncation degrees $d_\mu, d_\rho$, QLSS target error $\epsilon_\mathrm{QLSS}$, code distance $d$, magic state infidelity $\epsilon_T, \epsilon_R$, and tomography sampling number $\mathcal{N}_\mathrm{sample}$.
Given these hyperparameters, we can compute algorithm key parameters (golden colored) of the logical qubits number $\mathcal{QN}_L$, the query number $\mathcal{Q}$, truncation errors $\epsilon_\mu(d_\mu)$, $\epsilon_\rho(d_\rho)$, logical circuit rotation/Toffoli depth $\mathcal{TD}+\mathcal{RD}$, and rotation and Toffoli counts $\mathcal{TC}, \mathcal{RC}$.
Following that, we can evaluate the three sources of error:

In the algorithm phase (green colored),  errors arise from the imperfect block-encoding and state preparation subroutines, as well as the QLSS configuration. Intuitively, this algorithmic error originates in the quantum simulation of non-linear behaviour to approximate $\mu$, $1/\rho$, and ${1}/{A}$.
Given truncation degrees $d_\mu, d_\rho$ for the polynomial approximation of non-linear terms, the truncation errors $\epsilon_\mu$ and $\epsilon_\rho$ decay exponentially quickly.
Then these truncation errors are multiplied by a factor determined by block-encoding normalization constants to derive $\epsilon_A$ and $\epsilon_b$ as given in Eq.~\eqref{aeq:epsilon_A} and Eq.~\eqref{aeq:epsilon_b}.
Since each of these subroutines is called many times, we can bound the error of the imperfect block-encoding of $A$ and $b$ as
\begin{equation}
    \epsilon_\mathrm{BE/SP} \leq \mathcal{Q}\cdot(\epsilon_A + 2\epsilon_b).
\end{equation}
To also consider the pre-set QLSS error, the algorithm's algorithmic error can be bounded by
\begin{equation}
    \epsilon_\mathrm{alg} \leq \mathcal{Q}\cdot(\epsilon_A + 2\epsilon_b) + \epsilon_\mathrm{QLSS}.
\end{equation}
Then, in the deployment phase (blue colored), errors consist of the cumulative logical error and the distillation error.
Given the physical gate error $p_\mathrm{phys}$, surface code error threshold $p_\mathrm{th}=0.01$, and code distance $d$, the logical error rate can be estimated by
\begin{equation}
    P_L = 0.1(p_\mathrm{phys}/p_\mathrm{th})^{{d+1}/{2}}.
\end{equation}
Considering the logical qubit number $\mathcal{QN}_L$, the circuit Toffoli/rotation depth $(\mathcal{RD}+\mathcal{TD})$, and the $d$ error correction cycles for each Toffoli/rotation layer,  the accumulative logical error can be evaluated as
\begin{equation}
    \epsilon_\mathrm{a.l.e.} \leq \sqrt{2}P_L\cdot\mathcal{QN}_L\cdot(\mathcal{RD}+\mathcal{TD})\cdot d.
\end{equation}
And the distillation error in the magic state factories when implementing Toffoli / rotation gates can be evaluated as
\begin{equation}
    \epsilon_\mathrm{distillation} \leq \sqrt{2}(\mathcal{TC}\cdot\delta_\mathrm{TOF}+\mathcal{RC}\cdot\delta_\mathrm{ROT}).
\end{equation}
In summary, the deployment error is bounded by
\begin{equation}
    \epsilon_\mathrm{deploy} = \epsilon_\mathrm{accumulative} + \epsilon_\mathrm{distillation}.
\end{equation}
Finally, in the measurement phase (pink colored), the total error given sampling number $\mathcal{N}_\mathrm{sample}$ can be bounded by
\begin{equation}
    \epsilon_\mathrm{tomography}\leq\sqrt{\frac{cd_\mathrm{eff}\log^2d_\mathrm{eff}}{\mathcal{N}_\mathrm{sample}}}.
\end{equation}

In principle, given the end-to-end quantum Navier-Stokes error in each iteration bounded by 
\begin{equation}
    \epsilon_\mathrm{quantum} \leq \epsilon_\mathrm{alg} + \epsilon_\mathrm{deploy} + \epsilon_\mathrm{tomography}.
\end{equation}
and the problem-relevant hyperparameters determined by numerical tests, the Pareto frontier of  the end-to-end computing time overhead
\begin{equation}
    \mathcal{T}_\mathrm{quantum} = \mathcal{N}_\mathrm{sample}\times(\mathcal{RD}+\mathcal{TD})\times d\times\mathcal{T_\mathrm{cycle}}
\end{equation}
\begin{equation}
    \mathcal{QN}_\mathrm{physical} = v_\mathrm{TOF}\cdot V_\mathrm{TOF} + v_\mathrm{ROT}\cdot V_\mathrm{ROT}
\end{equation}
and physical qubits can be solved by a multi-objective constrained optimization of those free hyperparameters.
However, the dependency on those \textit{free hyperparameters} is often discontinuous and hard to evaluate.
For example, for each choice of magic state infidelity, a sequence of numerical experiments and optimizations is required to determine the magic state factory's configuration of footprint size and distillation times.
Consequently, we consider a more practical method to determine these free hyperparameters by a heuristic procedure described as follows:

\subsection{Characterizing Problem-Specific and Free Hyperparameters}\label{}
We first characterize those problem-specific hyperparameters as follows:
\subparagraph{Iteration-tolerant Error Bound.}
To model the impact of quantum error, we introduce a controlled, single-step noise in each iteration. Specifically, each component of the solution $x$ (the residual vector) is multiplied by a random factor drawn from the uniform distribution $\mathcal{U}(1 - \epsilon_\mathrm{noise}, 1 + \epsilon_\mathrm{noise})$. This per-iteration error accumulates throughout the simulation, leading to a deviation in the solution field at the final time. As detailed in Figure~\ref{fig:error-noise-n}, we find that for a per-iteration error 
\begin{equation}    
    \epsilon_\mathrm{noise} \leq = \epsilon_{\mathrm{threshold}} =  0.05,
\end{equation}

\noindent the final accumulated error is nearly indistinguishable from the baseline discretization error of the noiseless case ($\epsilon_\mathrm{total}=0$). This implies that noise within this iteration-tolerant error threshold  $\epsilon_{\mathrm{threshold}}$has a negligible impact on the overall accuracy, thus ensuring high-fidelity results.

In the noiseless case ($\epsilon_\mathrm{noise}=0$), the error systematically decreases as problem size $N$ increases, confirming the expected second-order spatial accuracy of our numerical scheme. When noise is introduced, the final error grows at an accelerating rate with the noise level $\epsilon_\mathrm{noise}$ for all problem sizes $N$, clearly illustrating the compounding effect of per-iteration inaccuracies. For the density and energy, however, an error floor is observed, which persists even for large $N$ in the noiseless case. This is an expected discrepancy, as the analytical solution assumes a perfectly uniform temperature, whereas our simulation correctly captures minor physical density and temperature fluctuations arising from compressibility effects.
\begin{figure}
    \centering
    \includegraphics[width=\linewidth]{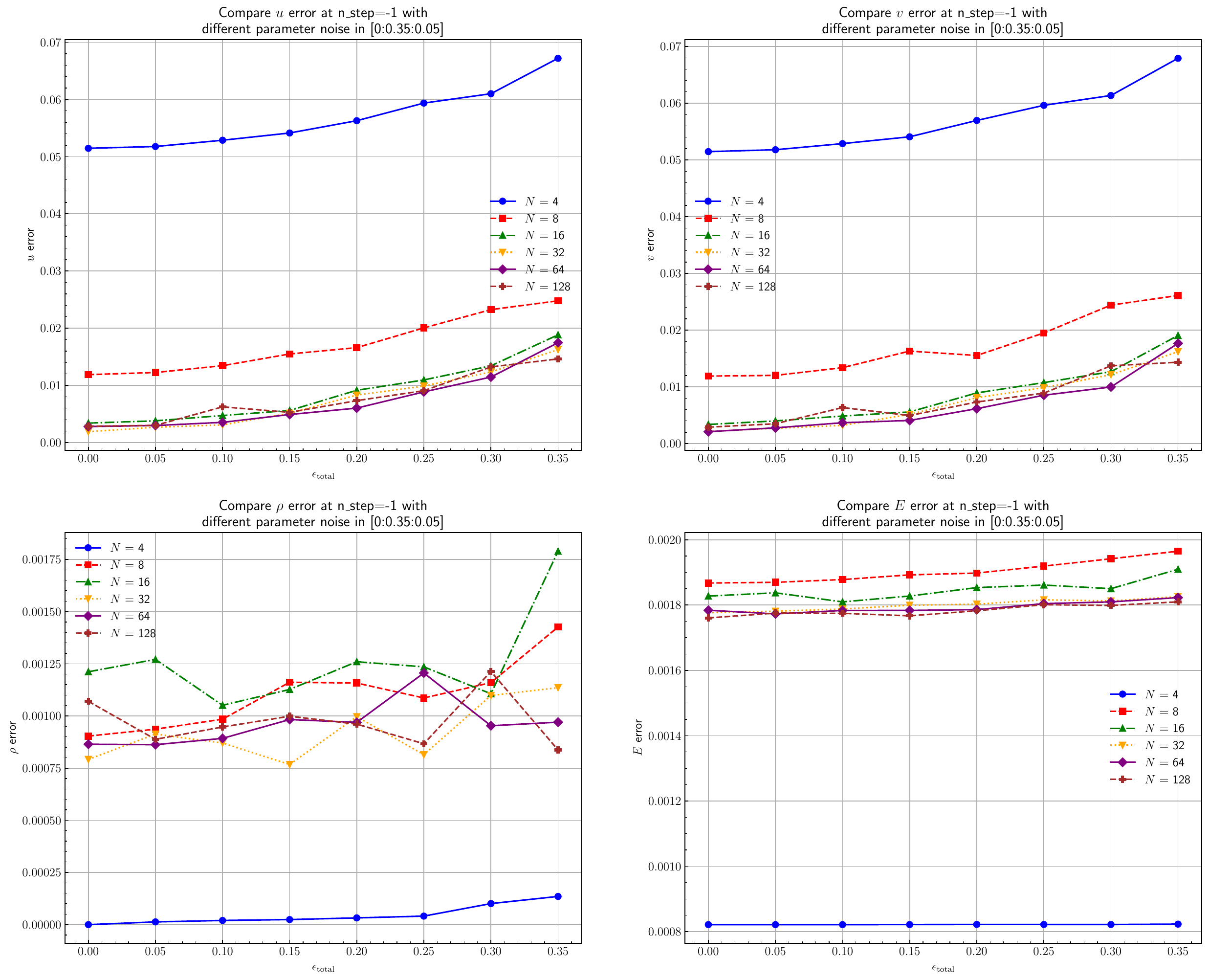}
    \caption{Final relative error as a function of noise magnitude $\epsilon_\mathrm{noise}$ for various problem sizes $N$.}
    \label{fig:error-noise-n}
\end{figure}

\subparagraph{Spectral Sparsity.}
We model the impact of spectral sparsity by retaining only the $\mathcal{S}$ largest spectral coefficients of a given field. This filtering is applied to the residual field at each iteration to mimic sparse spectral tomography, and it is also applied to the initial condition to represent sparse input data. As shown in Figure~\ref{fig:qty-ft-contour}, the residual field exhibits a spectrum that is naturally sparse.

\begin{figure}
    \centering
    \includegraphics[width=0.6\linewidth]{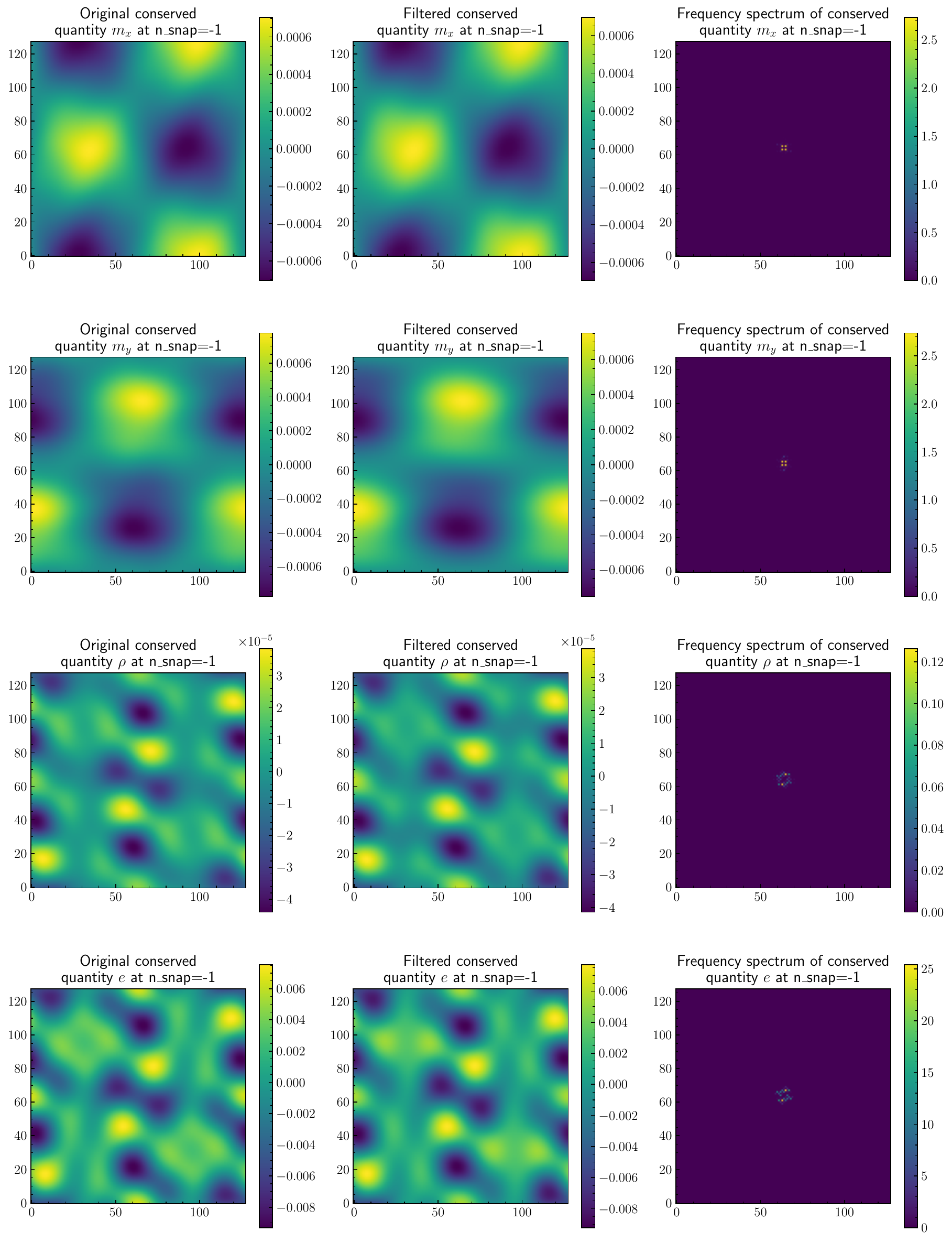}
    \caption{Residual fields of the four conserved quantities before and after spectral sparsification. The frequency spectra (right) demonstrate the sparsity structure.}
    \label{fig:qty-ft-contour}
\end{figure}

To quantify the effect of sparsity, we introduce a metric defined as the ratio of the $L^2$-norm of the field after spectral filtering to that of the original field. Figure~\ref{fig:norm-ratio-sparsity-case-n} illustrates how this norm ratio varies with the sparsity level $\mathcal{S}$ across different test cases and problem sizes $N$. The results show that the norm ratio is largely independent of the problem size $N$. Additionally, for test cases that simultaneously involve a broader range of frequencies, a higher sparsity level $\mathcal{S}$ is required to maintain accuracy. Notably, once the sparsity level reaches $\mathcal{S} \geq 64$, the norm ratio approaches unity. This indicates that the filtered field preserves the essential characteristics of the original field, confirming that this level of sparsity is sufficient to maintain accuracy for the problems under consideration.

\begin{figure}
    \centering
    \includegraphics[width=0.6\linewidth]{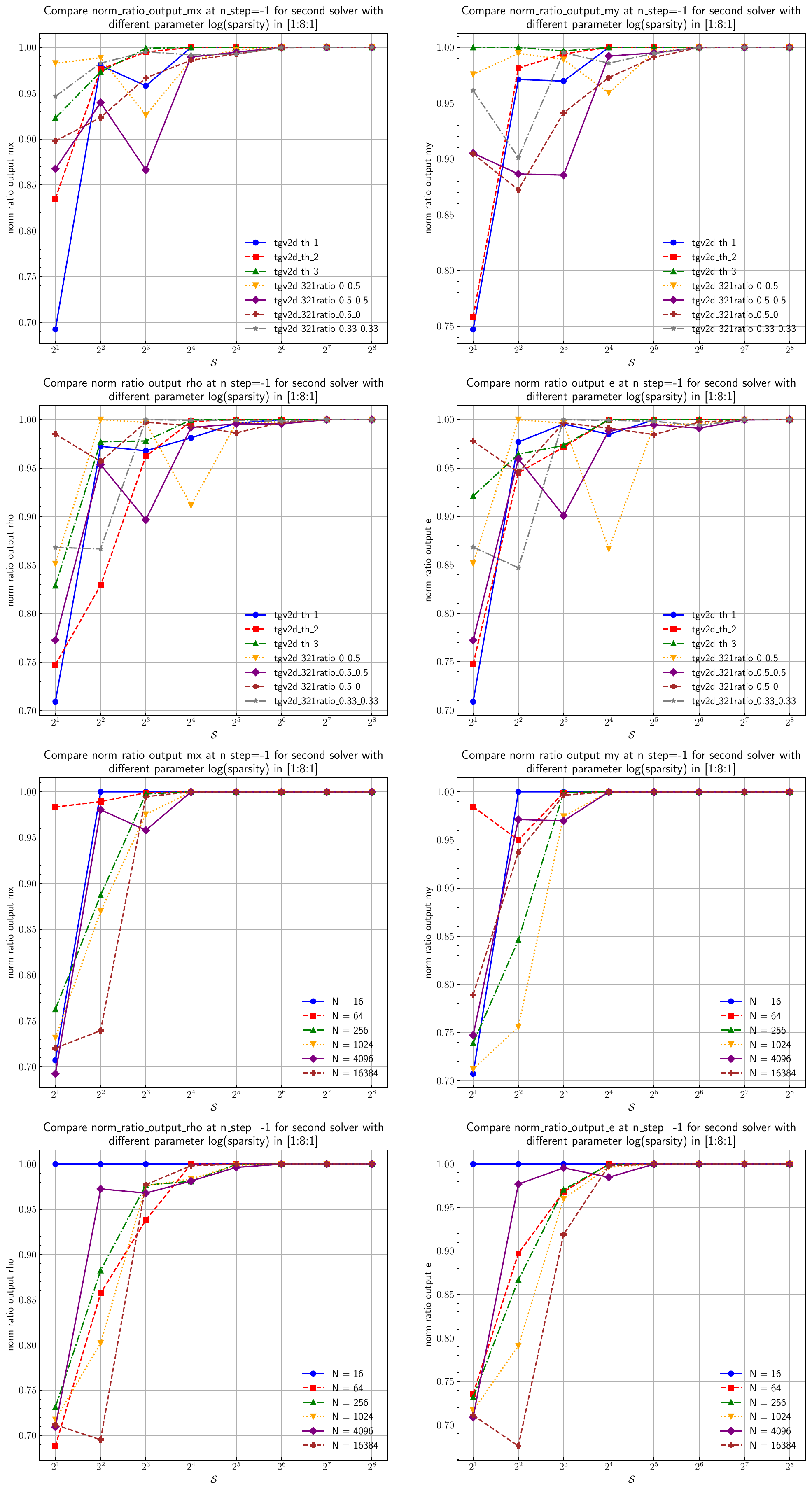}
    \caption{The norm ratio, obtained via spectral filtering, plotted as a function of the sparsity level $\mathcal{S}$. The results demonstrate robustness across various test cases and problem sizes. A sparsity level of $\mathcal{S} \geq 64$ is shown to effectively preserve accuracy, as indicated by the norm ratio approaching 1.0.}
    \label{fig:norm-ratio-sparsity-case-n}
\end{figure}

\subparagraph{Problem Size.}
We consider a grid of size $2^{25}\times2^{25}$ that corresponds to a linear system of $2^{52}$-dimension, i.e., a Jacobian matrix of size $2^{52}\times2^{52}$ and a residual vector of length $2^{52}$.
The corresponding logical qubits number in the quantum circuit is
\begin{equation}\label{aeq:numerical_logical_qubits_number}
    \mathcal{QN}_\mathrm{logical}=181.
\end{equation}

\subparagraph{Condition Number.}
We numerically characterize the condition number $\kappa$ by tracking its behavior under varying problem size $N$ and bandwidth $\mathcal{S}$ in Fig.~\ref{fig:kappa}.
For each $\mathcal{S}$, we generate a sequence of random initial conditions with spectral sparsity $\mathcal{S}$, and evaluate the condition number of the underlying flux Jacobian matrix with varying $N$.
From the numerical tests, we have two interesting observations:
Firstly, the condition number $\kappa$ grows with the problem size $N$ until it exceeds a critical level and saturates.
Secondly, the saturation level grows with the iteration-tolerant bandwidth (efficient spectral sparsity)  $\mathcal{S}$.
Our first key observation, the \textbf{condition number saturation}, reveals that for sufficiently large problem sizes, the condition number 
$\kappa$ remains bounded when $\mathcal{S}$ is fixed. This suggests that $\kappa$ can be reliably extrapolated from numerical tests performed at smaller scales.
In particular, we suppose that
\begin{equation}
    \kappa(\mathcal{S}\leq64)\leq550.
\end{equation}
The second observation, \textbf{sensitivity to spectral sparsity}, indicates that the spectral sparsity of the underlying variable can be a dominant factor influencing the condition number. While $\kappa$ is a well-known determinant of computational complexity in both classical and quantum linear system-related algorithms, its relationship with spectral sparsity has—to the best of our knowledge—received little attention in the CFD community. This insight opens a new quantum-inspired perspective on linear system complexity for conventional scientific computing.
\begin{figure}
    \centering
    \includegraphics[width=\linewidth]{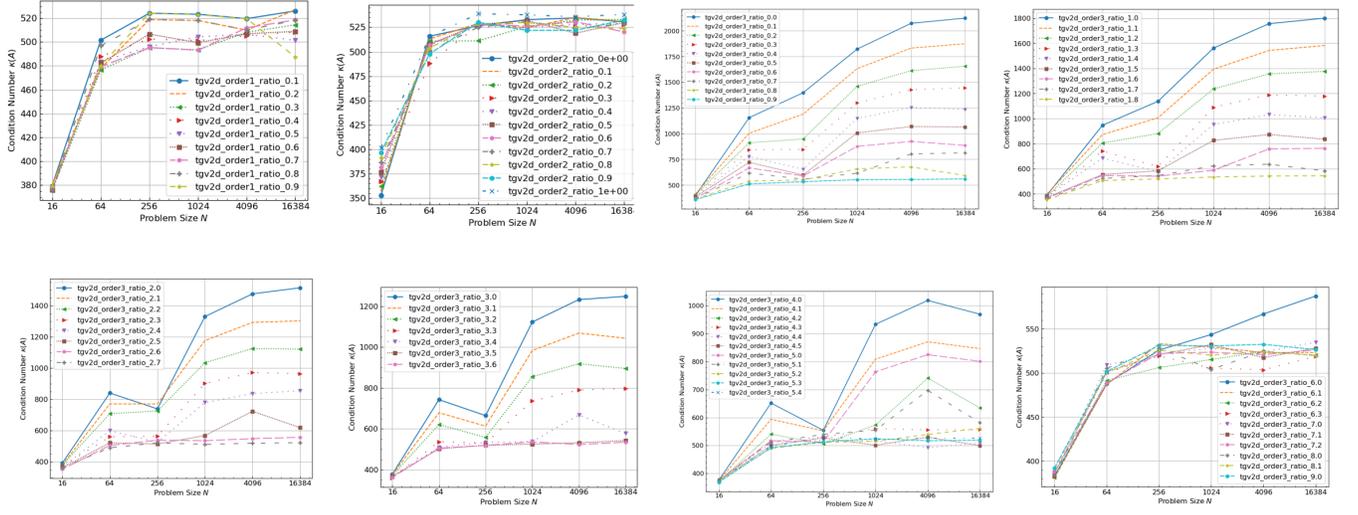}
    \caption{Condition number $\kappa$ with varying problem size $N$ and iteration-tolerant bandwidth $\mathcal{S}$.}
    \label{fig:kappa}
\end{figure}

\subparagraph{Normalization Constants.} We numerically monitor those normalization constants related to the algorithmic error terms in Eq.~\eqref{aeq:epsilon_A} and Eq.~\eqref{aeq:epsilon_b} with varying problem size and spectral sparsity as illustrated in Fig.~\ref{fig:normalization}.
The numerical result suggests stable bounds for each normalization constant as
\begin{equation}\label{aeq:error_ratios}
    \begin{split}
        r_b&\leq0.98<1,\\
        r_\mu&<0.001,\\
        r_\rho&<0.001.
    \end{split}
\end{equation}
\begin{figure}
    \centering
    \includegraphics[width=\linewidth]{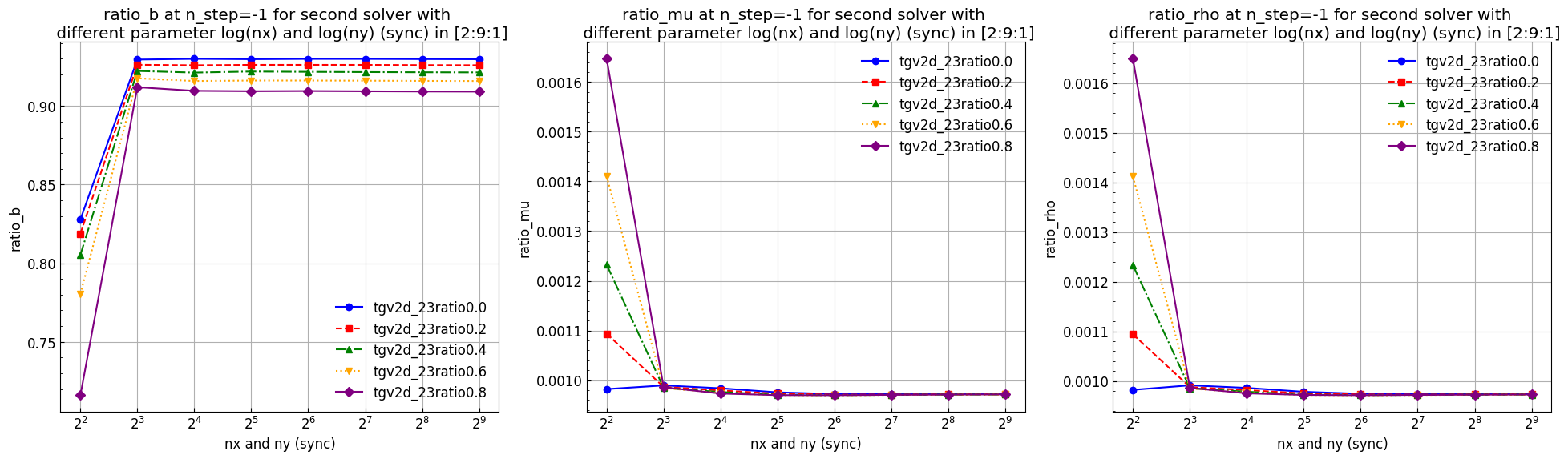}
    \caption{Normalization constants with varying problem size $N$ and iteration-tolerant bandwidth $\mathcal{S}$.}
    \label{fig:normalization}
\end{figure}

\subparagraph{QLSS Error.}
It is reported in Table~1 of Ref.~\cite{dalzell2024shortcut} that an augmented linear system kernel reflection with norm search methods of exhaustive search, Grover search, and adiabetic search.
We evaluate the precise query time of these three methods as depicted in Fig.~\ref{fig:query_time}.
The numerical result suggests that, while the adiabetic search method admits a tighter complexity upper bound, the first two methods present a much lower query complexity with $\kappa\leq5000$ since their prefactors are much smaller.
In particular, we choose the naive exhaustive norm search method with the shortest circuit depth.
Our numerical result shows that, in our case of $\kappa\leq550$ and $\epsilon_\mathrm{QLSS}\leq10^{-3}$, it is sufficient to suppose the query complexity 
\begin{equation}\label{aeq:Q_estimation}
    \mathcal{Q}=28301.
\end{equation}
\begin{figure}
    \centering
    \includegraphics[width=0.5\linewidth]{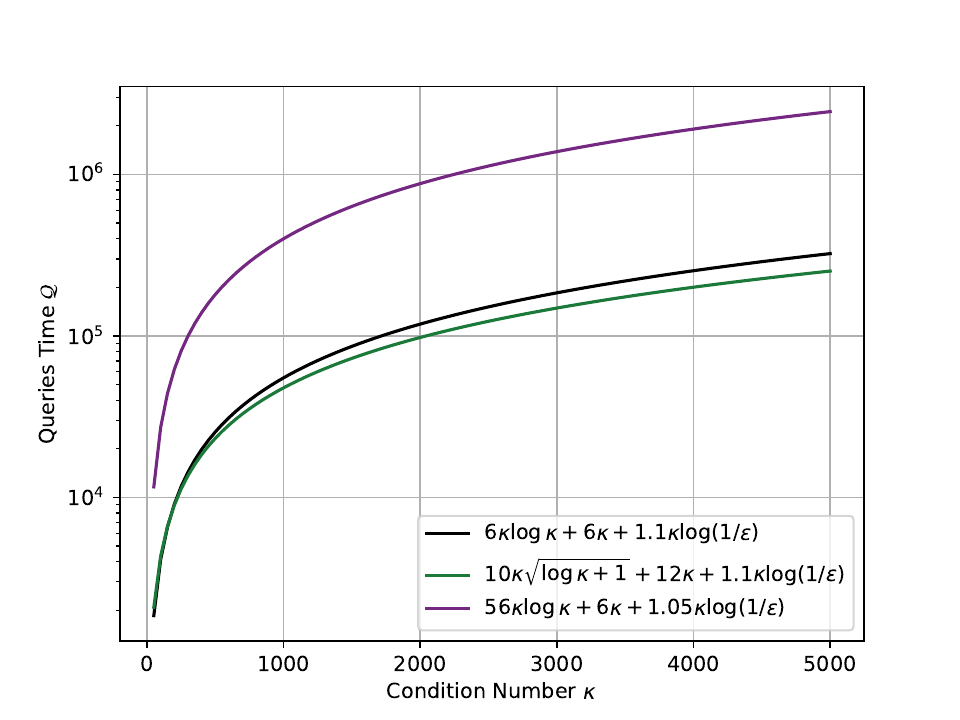}
    \caption{Quantum Linear System Solver Query Time $\mathcal{Q}$ with Error Bound $\epsilon_\mathrm{QLSS}=1\times10^{-3}$.}
    \label{fig:query_time}
\end{figure}

\subparagraph{Truncation Degree.}
We monitor the values of $T\propto e$ and $\rho$ across varying problem sizes and observe that, for a Mach number $\mathrm{Ma}=0.1$, they remain robustly constrained within the ranges $0.9997<T<1.0006$ and $0.991<\rho\leq 1$.
To implement the nonlinear functions $\mu$ and $1/\rho$, we employ Chebyshev polynomial approximations over the extended intervals $(0.999,1.002)$ and $(0.991,1)$, respectively. Numerical experiments confirm that
We track the value of $T\propto e$ and $\rho$ with varying problem size and find that $0.9997<T<1.0006$ and $0.991<\rho\leq 1$ with Mach speed $\mathrm{Ma}=0.1$ robustly.
We approximate the nonlinear function $\mu$ and $1/\rho$ by Chebyshev polynomials on wider intervals $\mu$ and $1/\rho$ , respectively.
The numerical experiment indicates that
\begin{equation}\label{aeq:truncation_degree}
    d_\mu = d_{\rho^{}} = 2
\end{equation}
suffice to suppress the truncation errors into
\begin{equation}\label{aeq:truncation_error}
    \begin{split}
        \epsilon_\mu &\leq 5.88\times10^{-11},\\
        \epsilon_\rho &\leq 3.20\times10^{-8}.
    \end{split}
\end{equation}
Insert Eqs.~\eqref{aeq:truncation_error} and Eqs.~\eqref{aeq:error_ratios} into Eq.~\eqref{aeq:epsilon_A} and Eq.~\eqref{aeq:epsilon_b} can bound the single-call block-encoding and state-preparation errors by
\begin{equation}
    \begin{split}
        \epsilon_A &\leq r_\rho\cdot\epsilon_\rho + r_\mu\cdot\epsilon_\mu\leq4.82\times10^{-11},\\
        \epsilon_b &\leq r_b\cdot\epsilon_\mu\leq5.88\times10^{-11}.
    \end{split}
\end{equation}
To also  consider the query time estimated in Eq.~\eqref{aeq:Q_estimation}, the accumulative block-encoding and state preparation error is bounded by
\begin{equation}
    \begin{split}
        \epsilon_{\mathrm{BE/SP}}\leq4.69\times10^{-6}.
    \end{split}
\end{equation}

\subparagraph{Magic State Infidelity.}
Given the logical qubits number in Eq.~\eqref{aeq:numerical_logical_qubits_number}, the query number in Eq.~\eqref{aeq:Q_estimation}, and the truncation degree in Eqs.~\eqref{aeq:truncation_degree}, the non-Clifford gate count is evaluated to be
\begin{equation}
    \begin{split}
        \mathcal{TC} &=9.41\times 10^7,\\ 
        \mathcal{RC} &=3.94\times 10^8.
    \end{split}
\end{equation}
Consequently, the infidelities of a single Toffoli gate and a rotation gate bounded by
\begin{equation}
    \begin{split}
        \delta_\mathrm{TOF} &\leq 2.8\times10^{-17},\\
        \delta_\mathrm{ROT} &\leq 3.0\times10^{-12}
    \end{split}
\end{equation}
are sufficient to suppress the distillation error into
\begin{equation}
    \begin{split}
        \delta_\mathrm{distillation}&\leq\mathcal{TC}\cdot\delta_\mathrm{TOF}+\mathcal{RC}\cdot\delta_\mathrm{ROT}<1.18\times10^{-3},\\
        \epsilon_\mathrm{distillation}&\leq\sqrt{2}\cdot\delta_\mathrm{distillation}<1.67\times10^{-3}.
    \end{split}
\end{equation}
Herein, we suppress the Toffoli gate infidelity much smaller than the rotation gate since its magic state factory's spacetime overhead is significantly lower.

\subparagraph{Code Distance.}
Given the logical qubits number in Eq.~\eqref{aeq:numerical_logical_qubits_number}, the query number in Eq.~\eqref{aeq:Q_estimation}, and the truncation degree in Eqs.~\eqref{aeq:truncation_degree}, the non-Clifford gate depth is evaluated to be 
\begin{equation}
    \mathcal{TD}+\mathcal{RD}<1.48\times10^8.
\end{equation}
Consequently, the code distance
\begin{equation}\label{aeq:d}
    d = 25
\end{equation}
is sufficient to suppress the accumulative logical error bounded by
\begin{equation}
    \begin{split}
        \delta_\mathrm{a.l.e.}&\leq0.1(P_\mathrm{physics}/P_\mathrm{threshold})^\frac{d+1}{2}\cdot\mathcal{QN}_L\cdot(\mathcal{RD}+\mathcal{TD})\cdot d<8.13\times10^{-7}\\
        \epsilon_\mathrm{a.l.e.}&\leq\sqrt{2}\cdot\delta_\mathrm{a.l.e.}<1.15\times10^{-6}.
    \end{split}
\end{equation}
Herein, we suppose the physical error rate $P_\mathrm{physics}=0.0005$.

\subparagraph{Sampling Number.}
We numerically test the required sampling numbers for compressed-sensing quantum state tomography as depicted in Fig.~\ref{fig:cs-qst}. For each effective dimension ranging from $4$ to $32$, we test the average infidelity between the reconstructed state and the target state with a varying number of measurements ranging from $500$ to $20000$, with each data point averaged over $10$ independent trials to ensure statistical robustness. The numerical result indicates that CS-QST can learn the target state with high fidelity, satisfying our requirement. In particular, only
\begin{equation}\label{aeq:N_estimation}
    \mathcal{N}_\mathrm{sample} = 1000
\end{equation}
shots can reconstruct the state with tomography error bounded by
\begin{equation}
    \begin{split}
        \delta_\mathrm{tomography}&\leq6\times10^{-5}\\
        \epsilon_\mathrm{tomography}&\leq\sqrt{2(1-\sqrt{1-\delta_\mathrm{tomography}})}<7.74\times10^{-3}.
    \end{split}
\end{equation}
To benchmark the constant $C$ in Theorem~\ref{thm:tomography}, we also perform a linear regression of the sampling number on the explanatory variable $rd_\mathrm{eff}\log^2d_\mathrm{eff}$. The fitting result shows a strong linear dependency as $d_\mathrm{eff}$ grows.
\begin{figure}
    \centering
    \includegraphics[width=0.95\textwidth]{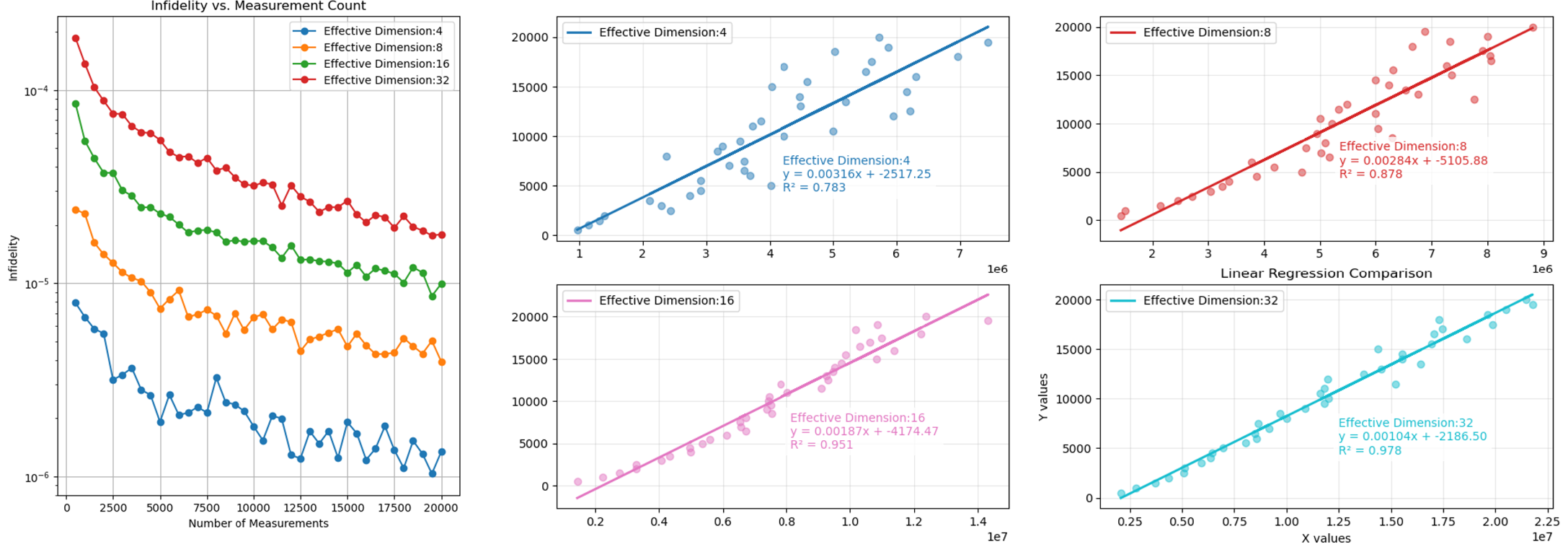}
    \caption{Compressed Sensing Quantum State Tomography.}\label{fig:cs-qst}
\end{figure}

\subsection{Quantum Computation Overhead}
Given the sample time in Eq.~\eqref{aeq:N_estimation}, the quantum error correction cycle number
\begin{equation}
    \mathcal{N}_\mathrm{cycle} = (\mathcal{RD} + \mathcal{TD})\times d < 3.70\times10^9, 
\end{equation}
and the standard QEC per-cycle time  on a superconducting quantum processor
\begin{equation}
    \mathcal{T}_\mathrm{cycle} = 10^{-6}\text{ s},
\end{equation}
the total quantum time can be evaluated as
\begin{equation}
    \begin{split}
        \mathcal{T}_\mathrm{quantum} = \mathcal{N}_\mathrm{sample}\times(\mathcal{RD}+\mathcal{TD})\times d\times\mathcal{T_\mathrm{cycle}}\leq3.68\times10^6 \text{ s}\simeq42.6\text{ days}.
    \end{split}
\end{equation}
By Eq.~\eqref{aeq:numerical_logical_qubits_number} and Eq.~\eqref{aeq:d}, the quantum circuit's physical qubits number is
\begin{equation}
    \mathcal{QN}_\mathrm{circuit} = 181\times2\times25^2=226250.
\end{equation}
The non-Clifford gate consumption rates are
\begin{equation}\label{aeq:consumption}
    \begin{split}
        v_\mathrm{TOF} &= 2.56\times10^{-2} \mathrm{ cycle}^{-1} \\
        v_\mathrm{ROT} &= 1.07\times10^{-1} \mathrm{ cycle}^{-1}
    \end{split}
\end{equation}
By Eqs.~\eqref{aeq:volume_rot} and Eqs.~\eqref{aeq:consumption}, the magic state factory's physical qubits number is
\begin{equation}
    \mathcal{QN}_\mathrm{factory} = v_\mathrm{TOF}\cdot V_\mathrm{TOF} + v_\mathrm{ROT}\cdot V_\mathrm{ROT} = 8.16\times10^6.
\end{equation}
Assuming a pessimistic estimation of routing qubits to be as many as circuit qubits, the total physical qubits is
\begin{equation}
    \mathcal{QN}_\mathrm{physics} = \mathcal{QN}_\mathrm{circuit} + \mathcal{QN}_\mathrm{routing} + \mathcal{QN}_\mathrm{factory} =8.71\times10^6.
\end{equation}

\subsection{Classical Computation Overhead}
We follow the analysis as in Ref.~\cite{tu2025towards} that, given the linear system dimension and sparsity to be
\begin{equation}
    \begin{split}
        s&=21,\\
        N_\mathrm{LS}&=2^{82},
    \end{split}
\end{equation}
the total floating-point operations (FLOPs) in solving an $N_\mathrm{LS}$-dimensional linear system by Conjugate Gradient and Cholesky Decomposition methods are respectively to be
\begin{equation}
    \mathcal{C}_\mathrm{CG} = (2s+7)N_\mathrm{LS}\cdot\kappa\log\left(\frac{2}{\epsilon}\right) \geq (2\times21+7)\times2^{82}\times 550 \times\log(200) > 9.96\times10^{29} \text{ FLOPs},
\end{equation}
and
\begin{equation}
    \mathcal{C}_\mathrm{CD} = N_\mathrm{LS}\cdot\left(3s^2+7s+5\right) \geq 2^{82}\times 1475 > 7.13\times10^{27} \text{ FLOPs}.
\end{equation}
The most powerful supercomputer achieves a maximal LINPACK performance of
\begin{equation}
    R_\mathrm{max} = 1.74\times10^{18}\text{ FLOP}\cdot \text{s}^{-1}
\end{equation}
and a theoretical peak performance 
\begin{equation}
    R_\mathrm{peak} = 2.75\times10^{18}\text{ FLOP}\cdot \text{s}^{-1},
\end{equation}
as reported in Ref.~\cite{top500}.
The corresponding classical computing time can be estimated as
\begin{equation}
    \mathcal{T}_\mathrm{El Capitan64} = \mathcal{C}_\mathrm{classical}/ R_\mathrm{max} \geq 4.10\times10^9 \text{ seconds}\simeq129.85 \text{ years},
\end{equation}
Even considering a single-precision floating-point computing, which is often insufficient for CFD problems, the classical time is
\begin{equation}
    \mathcal{T}_\mathrm{El Capitan32} = \mathcal{C}_\mathrm{classical} /R_\mathrm{max32} \geq 2.05\times10^9 \text{ seconds}\simeq 64.92 \text{ years}
\end{equation}
Alternatively, assuming an ideal parallel execution where we use the theoretical peak performance $R_\mathrm{peak}$ and neglect the law of diminishing marginal returns caused by communications among all high performance clusters, solving the same-sized problem within an equal time still requires
\begin{equation}
    \mathcal{T}_\mathrm{quantum}/\mathcal{T}_\mathrm{El Capitan, peak}\simeq720
\end{equation}
copies of the El Capitan supercomputers.

\section{Further Discussion}
\subsection{Further Discussion on Boundary Condition}
We briefly discuss how to handle boundary conditions beyond the periodic boundary condition in  two steps.
Firstly, we utilize the boundary operator $U_\mathrm{boundary}$ acting on the cell index register and the boundary flag register to mark whether it is on the boundary.
A standard rectangle boundary shape can be marked by
\begin{equation}\label{eq:boundary_operator}
    U_\mathrm{boundary}\ket{I_\mathrm{row}, I_\mathrm{col}}^{\otimes(n_x+n_y)}_\mathrm{CI}\ket{0}^{\otimes 4}_\mathrm{BF} = \ket{I_\mathrm{row}, I_\mathrm{col}}^{\otimes(n_x+n_y)}_\mathrm{CI}\left(\bigotimes_{\sigma \in \{\text{left,right,top,bottom}\}} \ket{\chi_\sigma(I_\mathrm{row}, I_\mathrm{col})}\right)_\mathrm{BF},
\end{equation}
as shown in Fig.~\ref{fig:boundary_operator}. Here the indicator function $\chi_\sigma$ is defined as

\begin{equation}
\chi_\sigma(I_\mathrm{row}, I_\mathrm{col}) =
\begin{cases}
1, & \text{if } (I_\mathrm{row}, I_\mathrm{col}) \text{ lies on the $\sigma$ boundary}, \\
0, & \text{otherwise}.
\end{cases}
\end{equation}

A more complicated boundary shape that can be efficiently computed by an algebraic function, such as a circular boundary, can be marked by either a quantum arithmetic or a quantum singular value transformation.
Secondly, we can apply the following process according to its boundary conditions:
\begin{itemize}
    \item \textbf{Wall Boundary Condition}: 
    In the standard wall boundary condition with no slip, such as in an airfoil simulation, the fluid sticks to the wall so that the velocity is forced to zero.
    For this condition, consider the following extended block-encoding with $4$ additional ancillary qubits (the cell index register):
    \begin{equation}
        \bra{0}^{\otimes4}_\mathrm{BF}\bra{0}^{\otimes a}_\mathrm{ANC}\bra{0}^{\otimes n}_\mathrm{WORK}U_\mathrm{boundary}U_A\ket{0}^{\otimes4}_\mathrm{BF}\ket{0}^{\otimes a}_\mathrm{ANC}\ket{0}^{\otimes n}_\mathrm{WORK}=\frac{1}{\alpha'}\Tilde{A},
    \end{equation}
    where $U_A$ is an $(\alpha, a, 0)$-block-encoding of the original Jacobian matrix $A$, and $\Tilde{A}$ is the Jacobian matrix with wall boundary condition. To discard the amplitudes for the basis not vanish on the boundary flag register, we succeed in including the wall boundary condition at a cost of (usually at a constant ratio of) smaller success probability.
    \item \textbf{Inlet Boundary Condition:}
    In the inlet boundary condition, such as in a wind tunnel simulation, the incoming flow properties are specified by $u=u_\mathrm{inlet}$. 
    The corresponding block encoding can be derived from Eq.~\eqref{eq:boundary_operator} by substituting the multi-controlled not gates with multi-controlled rotation gates conditioned on the cell index register. By modifying the rotation angle, we can adjust the amplitudes on each boundary adaptively for flexible inlet boundary conditions.
    \item \textbf{Outlet Boundary Condition:}
    In the outlet boundary condition, such as in  exhaust flows, the outlet pressure is specified by $p=p_\mathrm{outlet}$. In the conservative form, this can be implemented by a corrected energy
    \begin{equation}\label{aeq:corrected_energy}
        E_\mathrm{corrected} = \frac{p_\mathrm{outlet}}{\gamma-1} + \frac{1}{2}\rho\Vec{v}^2.
    \end{equation}
    We can encode the corrected energy in Eq.~\eqref{aeq:corrected_energy} by a unitary $U_\mathrm{corrected}$ which is similar to the construction of $U_\mathrm{BE}(H)$ in the main text, and then consider the linear combination of $U_\mathrm{corrected}$ and $U_\mathrm{boundary}$.
\end{itemize}
We leave the details and many other boundary conditions as our future work.
\begin{figure}
    \centering
    \begin{tikzpicture}
\node {
    \begin{quantikz}
        \lstick[2,label style={xshift=-0.5cm, yshift=1cm,rotate=90}]{Cell Index}\midstick[1,brackets=none]{$\ket{I_\mathrm{row}}_\mathrm{row}$}&\qwbundle{n_y}&\gate[6]{U_\mathrm{boundary}}&\midstick[6,brackets=none]{=}&\ctrl[open]{2} \gategroup[wires=5,steps=2,background,style={dashed,rounded corners, inner xsep=-2pt}]{} &&\ctrl{3}\gategroup[wires=6,steps=2,background,style={dashed,rounded corners, inner xsep=-2pt}]{}&&\\
        \midstick[1,brackets=none]{$\ket{I_\mathrm{col}\text{ }}_\mathrm{col}$}&\qwbundle{n_x}&&&&\ctrl[open]{3}&&\ctrl{4}&\\
        \lstick[4,label style={xshift=-0.5cm, yshift=1.25cm, rotate=90}]{Boundary Flag}\midstick[1,brackets=none]{$\ket{0}_\mathrm{top}$}\text{\ \ \ \ }&&&&\targ{}&&&&\\
        \midstick[1,brackets=none]{$\ket{0}_\mathrm{bottom}$}&&&&&&\targ&&&\\
        \midstick[1,brackets=none]{$\ket{0}_\mathrm{left}$\text{\ \ \ }}&&&&&\targ&&&&\\
        \midstick[1,brackets=none]{$\ket{0}_\mathrm{right}$\text{\ \ }}&&&&&&&\targ&&
    \end{quantikz}
};
\end{tikzpicture}
   \caption{Quantum circuit for a standard rectangle boundary operator: The top and bottom boundary flags are $X$ gates controlled by states $\ket{00...0}_\mathrm{row}$ and $\ket{11...1}_\mathrm{row}$, respectively. The left and right boundary flags are $X$ gates controlled by states $\ket{00...0}_\mathrm{col}$ and $\ket{11...1}_\mathrm{col}$, respectively. The gates in each dashed circled group can be implemented parallel.}
    \label{fig:boundary_operator}
\end{figure}
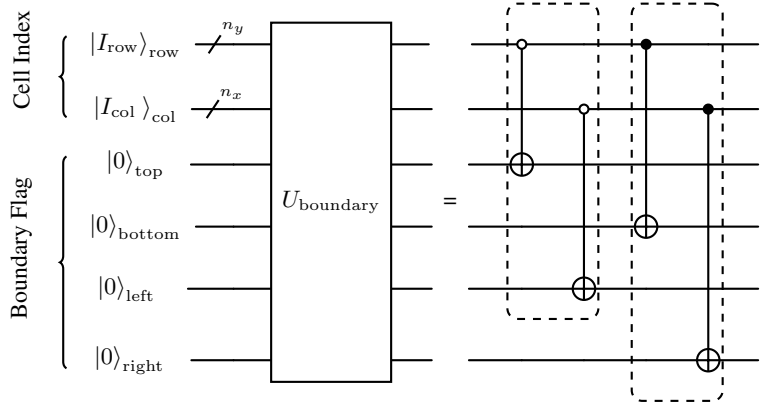

\subsection{Further Discussion on Riemann Problem}
In this subsection, we show how to design efficient quantum Riemann solvers with classic Lax-Friedrichs and Rusanov  schemes.
In the simplified introductory case of the classic Lax-Friedrichs Riemann solver, we illustrate the basic idea to solve the Riemann problem with boundary conditions.
Then, in the more complicated case of the Rusanov Riemann solver, we apply quantum signal processing to enable a more practical and flexible numerical scheme.
In principle, these techniques can be applied to solve many other Riemann solvers with interesting structures, which we leave as future work.

\subparagraph{Quantum Lax-Friedrichs Riemann Solver.}
We first consider the classic Lax-Friedrichs Riemann solver to derive the following residual vector:
\begin{equation}\label{eq:lf_residual}
    \begin{split}
        R^\mathrm{(L-F)}_{I,J}
        = &\mathcal{F}^\mathrm{(L-F)}_{I+\frac{1}{2},J} - \mathcal{F}^\mathrm{(L-F)}_{I-\frac{1}{2},J} + \mathcal{F}^\mathrm{(L-F)}_{I, J+\frac{1}{2}} - \mathcal{F}^\mathrm{(L-F)}_{I, J-\frac{1}{2}}\\
        = &\frac{1}{2}(\mathcal{F}_{I+1,J}-\mathcal{F}_{I-1,J}+\mathcal{F}_{I,J+1}-\mathcal{F}_{I,J-1})-\frac{\Delta x}{2\Delta t}(W_{I+1,J}+W_{I-1,J}+W_{I,J+1}+W_{I,J-1}-4W_{I,J}),
    \end{split}
\end{equation}
wherein the numerical fluxes are defined to be
\begin{equation}\label{eq:lf_scheme}
    \begin{split}
        \mathcal{F}^\mathrm{(L-F)}_{I+\frac{1}{2},J} &= \frac{1}{2}(\mathcal{F}_L+\mathcal{F}_R)-\frac{\Delta x}{2\Delta t}(W_R-W_L)=\frac{1}{2}(\mathcal{F}_{I,J}+\mathcal{F}_{I+1,J})-\frac{\Delta x}{2\Delta t}(W_{I+1,J}-W_{I,J}),\\
        \mathcal{F}^\mathrm{(L-F)}_{I-\frac{1}{2},J} &= \frac{1}{2}(\mathcal{F}_L+\mathcal{F}_R)-\frac{\Delta x}{2\Delta t}(W_R-W_L)=\frac{1}{2}(\mathcal{F}_{I-1,J}+\mathcal{F}_{I,J})-\frac{\Delta x}{2\Delta t}(W_{I,J}-W_{I-1,J}),\\
        \mathcal{F}^\mathrm{(L-F)}_{I,J+\frac{1}{2}} &= \frac{1}{2}(\mathcal{F}_U+\mathcal{F}_D)-\frac{\Delta x}{2\Delta t}(W_D-W_U)=\frac{1}{2}(\mathcal{F}_{I,J}+\mathcal{F}_{I,J+1})-\frac{\Delta x}{2\Delta t}(W_{I,J+1}-W_{I,J}),\\
        \mathcal{F}^\mathrm{(L-F)}_{I,J-\frac{1}{2}} &= \frac{1}{2}(\mathcal{F}_U+\mathcal{F}_D)-\frac{\Delta x}{2\Delta t}(W_D-W_U)=\frac{1}{2}(\mathcal{F}_{I,J-1}+\mathcal{F}_{I,J})-\frac{\Delta x}{2\Delta t}(W_{I,J}-W_{I,J-1}).\\
    \end{split}
\end{equation}
Note that Eq.~\eqref{eq:lf_residual} is indeed a linear combination of components with fixed coefficients and computable indices, and it can be efficiently implemented via Lemma~\ref{lem:lcu}.
In our case, the state-preparation-pair to prepare the fixed vector $y=(\frac{1}{2}, -\frac{1}{2}, \frac{1}{2}, -\frac{1}{2}, -\frac{\Delta x}{2\Delta t}, -\frac{\Delta x}{2\Delta t}, -\frac{\Delta x}{2\Delta t}, -\frac{\Delta x}{2\Delta t}, \frac{2\Delta x}{\Delta t})^\mathsf{T}$ is easy to construct, and the unitaries $\mathcal{F}$ and $W$ are implemented utilizing the sparse spectral block-encoding in Lemma~\ref{lem:spectral}.

The remaining problem is to address the neighboring cell's location, i.e., those operator subscripts.
Note that the classical information is encoded by rotation angles computed from a binary tree; hence, direct arithmetic of indices will not work.
We instead implement arithmetic operators on the basis after the state preparation subroutine, as when block-encoding $b$.

\subparagraph{Quantum Rusanov Riemann Solver.}
The Rusanov Riemann solver gives the following residual vector:
\begin{equation}\label{eq:rusanov_residual}
    \begin{split}
        R^\mathrm{(Rusanov)}_{I,J}
        = &\mathcal{F}^\mathrm{(Rusanov)}_{I+\frac{1}{2},J} - \mathcal{F}^\mathrm{(Rusanov)}_{I-\frac{1}{2},J} + \mathcal{F}^\mathrm{(Rusanov)}_{I, J+\frac{1}{2}} - \mathcal{F}^\mathrm{(Rusanov)}_{I, J-\frac{1}{2}}\\
        = &\frac{1}{2}(\mathcal{F}_{I+1,J}-\mathcal{F}_{I-1,J})-\frac{1}{2}S^+_{I+\frac{1}{2}, J}(W_{I+1,J}-W_{I,J})+\frac{1}{2}S^+_{I-\frac{1}{2}, J}(W_{I,J}-W_{I-1,J})\\
        &+\frac{1}{2}(\mathcal{F}_{I,J+1}-\mathcal{F}_{I,J-1})-\frac{1}{2}S^+_{I, J+\frac{1}{2}}(W_{I,J+1}-W_{I,J})+\frac{1}{2}S^+_{I, J-\frac{1}{2}}(W_{I,J}-W_{I,J-1}),
    \end{split}    
\end{equation}
with the numerical fluxes defined to be
\begin{equation}\label{eq:rusanov_scheme}
    \begin{split}
        \mathcal{F}^\mathrm{(Rusanov)}_{I+\frac{1}{2},J} &= \frac{1}{2}(\mathcal{F}_L+\mathcal{F}_R)-\frac{1}{2}S^+_{I+\frac{1}{2}, J}(W_R-W_L)=\frac{1}{2}(\mathcal{F}_{I,J}+\mathcal{F}_{I+1,J})-\frac{1}{2}S^+_{I+\frac{1}{2}, J}(W_{I+1,J}-W_{I,J}),\\
        \mathcal{F}^\mathrm{(Rusanov)}_{I-\frac{1}{2},J} &= \frac{1}{2}(\mathcal{F}_L+\mathcal{F}_R)-\frac{1}{2}S^+_{I-\frac{1}{2}, J}(W_R-W_L)=\frac{1}{2}(\mathcal{F}_{I-1,J}+\mathcal{F}_{I,J})-\frac{1}{2}S^+_{I-\frac{1}{2}, J}(W_{I,J}-W_{I-1,J}),\\
        \mathcal{F}^\mathrm{(Rusanov)}_{I,J+\frac{1}{2}} &= \frac{1}{2}(\mathcal{F}_U+\mathcal{F}_D)-\frac{1}{2}S^+_{I, J+\frac{1}{2}}(W_D-W_U)=\frac{1}{2}(\mathcal{F}_{I,J}+\mathcal{F}_{I,J+1})-\frac{1}{2}S^+_{I, J+\frac{1}{2}}(W_{I,J+1}-W_{I,J}),\\
        \mathcal{F}^\mathrm{(Rusanov)}_{I,J-\frac{1}{2}} &= \frac{1}{2}(\mathcal{F}_U+\mathcal{F}_D)-\frac{1}{2}S^+_{I, J-\frac{1}{2}}(W_D-W_U)=\frac{1}{2}(\mathcal{F}_{I,J-1}+\mathcal{F}_{I,J})-\frac{1}{2}S^+_{I, J-\frac{1}{2}}(W_{I,J}-W_{I,J-1}),
    \end{split}
\end{equation}
where the characteristic speeds $S^+$ are defined as
\begin{equation}
    \begin{split}
        S^+_{I\pm\frac{1}{2},J} &= \max\left\{ 
|u_{I,J}|, |u_{I,J}-c_{I,J}|, |u_{I,J}+c_{I,J}|, |u_{I\pm1,J}|, |u_{I\pm1,J}-c_{I\pm1,J}|, |u_{I\pm1,J}+c_{I\pm1,J}|\right\}\\
        S^+_{I,J\pm\frac{1}{2}} &= \max\left\{ 
|u_{I\pm1,J}|, |u_{I,J}-c_{I,J}|, |u_{I,J}+c_{I,J}|, |u_{I,J}|, |u_{I, J\pm1}-c_{I, J\pm1}|, |u_{I, J\pm1}+c_{I, J\pm1}|\right\},
    \end{split}
\end{equation}
and the local speed of sound $c$ is evaluated as
\begin{equation}\label{aeq:sound}
    c = \sqrt{\frac{\gamma p}{\rho}} = \sqrt{\gamma(\gamma-1)e}.
\end{equation}
The linear combination of those four intercell numerical fluxes is used to derive the Rusanov scheme residual vector.
Notably, the linear combination weights of form $S^+$ in Eq.~\eqref{eq:rusanov_residual} are no longer constants as those in Eq.~\eqref{eq:lf_residual}, and an efficient implementation of index-informed state-preparation-pairs is required.
Herein, we design the index-informed state-preparation-pair $(P_L', P_R')$ to recover the normalization information
\begin{equation}
    \left(\sqrt{1+S^+_{I-\frac{1}{2}, J}}, \sqrt{1+S^+_{I+\frac{1}{2}, J}}, \sqrt{1+S^+_{I, J-\frac{1}{2}}}, \sqrt{1+S^+_{I, J+\frac{1}{2}}}\right)^\mathsf{T}.
\end{equation}
in four steps:
(1) Firstly, we can apply QSVT on $U_u$ a polynomial approximation of $f(x)=\mathrm{abs}(x)$ to $(\alpha_1, a_1, \epsilon_1)$-block-encode $\lvert u\rvert$ and on $U_e$ a polynomial approximation of $g(x)=\sqrt{x}$ to $(\alpha_2, a_2, \epsilon_2)$-block-encode $c\propto\sqrt{e}$.
(2) Secondly, we apply an $(\alpha_1+\alpha_2, a_1+a_2+1, \alpha_2\epsilon_1+\alpha_1\epsilon_2)$-block encoding of
\begin{equation}
    |u|+c = \max\{|u|, |u-c|, |u+c|\}    
\end{equation}
by Lemma~\ref{lem:lcu}, and then apply a conditioned arithmetic to derive $|u_{I,J}|+c_{I,J}$ and $|u_{I+1,J}|+c_{I+1,J}$, respectively.
(3) Thirdly, we note that
\begin{equation}\label{eq:S+}
    \begin{split}
        S^+_{I+\frac{1}{2},J} &= \max\{|u_{I,J}|+c_{I,J}, |u_{I+1,J}|+c_{I+1,J}\}\\
        &=\frac{|u_{I,J}|+c_{I,J}}{2} + \frac{ |u_{I+1,J}|+c_{I+1,J}}{2} + \left\lvert \frac{|u_{I,J}|+c_{I,J}}{2} - \frac{ |u_{I+1,J}|+c_{I+1,J}}{2} \right\rvert
    \end{split}
\end{equation}
is a linear combination of $|u_{I,J}|+c_{I,J}$, $|u_{I+1,J}|+c_{I+1,J}$, and $\left\lvert \frac{|u_{I,J}|+c_{I,J}}{2} - \frac{ |u_{I+1,J}|+c_{I+1,J}}{2} \right\rvert$, and the last term can be block-encoded by a QSVT on $\frac{|u_{I,J}|+c_{I,J}}{2} - \frac{ |u_{I+1,J}|+c_{I+1,J}}{2}$ via the polynomial approximation of $f(x)=\mathrm{abs}(x)$.
(4) Finally, we apply a QSVT on Eq.~\eqref{eq:S+} via the polynomial approximation of $h(x)=\sqrt{\frac{x}{1+x}}$.
Consequently, the desired state-preparation-pair for each (normalized) intercell numerical flux
\begin{align}
    &P_\mathrm{L}\ket{0}^{\otimes 2} = \sqrt{\frac{1}{1+S^+}}\ket{00} + \sqrt{\frac{1}{1+S^+}}\ket{01} + \sqrt{\frac{S^+}{1+S^+}}\ket{10} + \sqrt{\frac{S^+}{1+S^+}}\ket{11},\\
    &P_\mathrm{R}\ket{0}^{\otimes 2} = \sqrt{\frac{1}{1+S^+}}\ket{00} + \sqrt{\frac{1}{1+S^+}}\ket{01} + \sqrt{\frac{S^+}{1+S^+}}\ket{10} - \sqrt{\frac{S^+}{1+S^+}}\ket{11}
\end{align}
can be constructed via
\begin{equation}
    \begin{split}
        P_\mathrm{L} &= U_{\sqrt{\frac{1}{1+S^+}}}\otimes H,\\
        P_\mathrm{R} &= (\ket{0}\bra{0}\otimes H + \ket{1}\bra{1}\otimes Z)U_{\sqrt{\frac{1}{1+S^+}}}.
    \end{split}
\end{equation}

\subsection{Further Discussion on Effect of Higher Mach Number}
The numerical experiments presented earlier were conducted at a low Mach number ($\mathrm{Ma} = 0.1$), a regime where the flow is nearly incompressible. As the Mach number increases, compressibility effects become dominant, and the fluctuations in density $\rho$ and temperature $T$ become much larger than in the $\mathrm{Ma} = 0.1$ case. Consequently, the numerical solution for a higher mach number flow diverges considerably from the analytical solution for incompressible flow.

These large fluctuations have a direct and significant impact on the normalization constants (one of the problem-specific hyperparameters), the variation of these parameters with increasing Mach number is illustrated in Figure~\ref{fig:factor-n-mach}.

\begin{figure}
    \centering
    \includegraphics[width=0.95\textwidth]{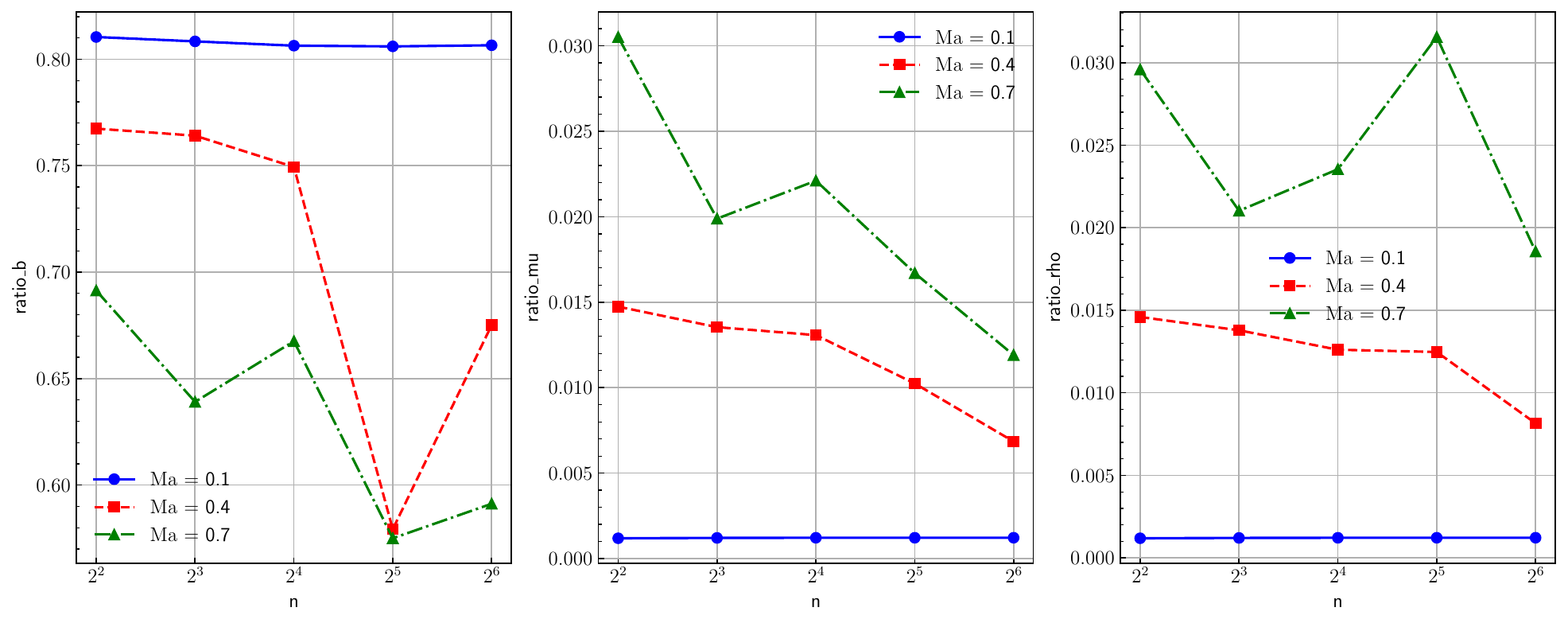}
    \caption{Variation of the normalization constants as a function of $n$ (the number of grid points in each dimension) for several Mach numbers. The plot illustrates that as the Mach number increases, the constants must be adjusted to account for the growing influence of compressibility effects in the flow.}
    \label{fig:factor-n-mach}
\end{figure}

From a quantum computing perspective, this presents a notable challenge. At higher Mach numbers, the spectra of density and temperature become richer. This increased complexity directly affects the implementation of quantum subroutines that rely on non-linear functions of these variables, which are typically realized using QSVT. For instance, calculating the reciprocal of the density ($1/\rho$) or implementing Sutherland's law for viscosity, which involves division by a function of temperature, becomes more demanding. To accurately approximate these functions over a wider and more complex range of values, QSVT requires higher-degree polynomials. This, in turn, translates to a greater demand on quantum resources, including increased circuit depth and potentially more ancilla qubits, making the simulation more challenging to execute.

\bibliographystyle{apsrev4-1}
\bibliography{main.bib}
\clearpage


\newcommand{\beginsupplement}{%
	\setcounter{table}{0}%
	\renewcommand{\thetable}{S\arabic{table}}%
	\setcounter{figure}{0}%
	\renewcommand{\thefigure}{S\arabic{figure}}%
	\setcounter{equation}{1}%
	\renewcommand{\theequation}{S\arabic{equation}}%
	\setcounter{section}{0}%
	\renewcommand{\thesection}{\arabic{section}}%
        \setcounter{page}{1}%
        \renewcommand{\thepage}{S\arabic{page}}%
}

\let\oldaddcontentsline\addcontentsline
\renewcommand{\addcontentsline}[3]{}

\let\addcontentsline\oldaddcontentsline
\resetlinenumber
\clearpage
\onecolumngrid

\beginsupplement
\renewcommand{\citenumfont}[1]{S#1}%
\renewcommand{\bibnumfmt}[1]{[S#1]}%

\NewDocumentCommand{\citesm}{>{\SplitList{,}} m }{%
  \def\temp{}%
  \ProcessList{#1}{\addSMprefix}%
  \expandafter\cite\expandafter{\temp}%
}
\newcommand{\addSMprefix}[1]{%
  \ifdefempty{\temp}%
    {\def\temp{SM_#1}}
    {\edef\temp{\temp,SM_#1}}%
}



\let\oldbibitem\bibitem
\renewcommand{\bibitem}[2][]{
  \ifstrempty{#1}{\oldbibitem{SM_#2}}{\oldbibitem[#1]{SM_#2}}
}

\end{document}